\documentclass[11pt]{article}
\usepackage[utf8]{inputenc}
\usepackage{amsmath,amsfonts,amsthm,amssymb,url}
\usepackage[dvipsnames]{xcolor}
\usepackage{fullpage}
\usepackage{thmtools, thm-restate}
\usepackage{algorithmic}
\usepackage{algorithm}
\usepackage[T1]{fontenc}
\usepackage[colorlinks=true, linkcolor=red, urlcolor=blue, citecolor=gray]{hyperref}
\usepackage{mathtools}
\usepackage{graphicx}
\usepackage{caption}
\usepackage{subcaption}
\usepackage{tikz}
\usetikzlibrary{calc,arrows.meta,positioning}
\usetikzlibrary{decorations.pathreplacing}
\usetikzlibrary{arrows}

\definecolor{ForestGreen}{RGB}{34,139,34}

\newtheorem{theorem}{Theorem}
\newtheorem{fact}{Fact}
\newtheorem{lemma}{Lemma}

\newtheorem{definition}{Definition}[section]

\newcommand{\Var}{\operatorname{Var}}

\let\Pr\relax
\DeclareMathOperator*{\Pr}{\mathbb{P}}

\DeclareMathOperator{\rank}{rank}

\newcommand{\R}{\mathbb{R}}
\newcommand{\C}{\mathbb{C}}

\newcommand{\E}{\mathbb{E}}
\newcommand{\poly}{\mathop\mathrm{poly}}

\newcommand{\nbyn}{n\times n}

\DeclareMathOperator{\nnz}{nnz}

\newcommand{\eqdef}{\mathbin{\stackrel{\rm def}{=}}}
\newcommand{\norm}[1]{\|#1\|}

\newcommand{\bs}[1]{\boldsymbol{#1}}
\newcommand{\bv}[1]{\mathbf{#1}}

\newcommand{\erdos}{Erd\"{o}s-R\'{e}nyi}

\usepackage{color}

\hypersetup{linktocpage}

\vspace{-1em}
\title{Sublinear Time Eigenvalue Approximation via Random Sampling}
\author{
Rajarshi Bhattacharjee\footnote{Manning College of Information and Computer Sciences, University of Massachusetts, Amherst, \texttt{\{rbhattacharj, cmusco, ray\}@cs.umass.edu}} 
\and 
Gregory Dexter\footnote{Department of Computer Science, Purdue University,
West Lafayette, \texttt{\{gdexter, pdrineas\}@purdue.edu}}
\and 
Petros Drineas\footnotemark[2]
\and 
Cameron Musco\footnotemark[1]
\and 
Archan Ray\footnotemark[1]
}
\date{}

  \usepackage{nth}
  \usepackage{intcalc}

\begin{document}
\sloppy

\maketitle

\begin{abstract}
We study the problem of approximating the eigenspectrum of a symmetric matrix $\bv A \in \mathbb{R}^{n \times n}$ with bounded entries (i.e., $\|\bv A\|_{\infty} \leq 1$). We present a simple sublinear time algorithm that approximates all eigenvalues of $\bv{A}$ up to additive error $\pm \epsilon n$ using those of a randomly sampled $\tilde {O}\left (\frac{\log^3 n}{\epsilon^3}\right ) \times \tilde O\left (\frac{\log^3 n}{\epsilon^3}\right )$ principal submatrix. Our result can be viewed as a concentration bound on the complete eigenspectrum of a random submatrix, significantly extending known bounds on just the singular values (the magnitudes of the eigenvalues). We give improved error bounds of $\pm \epsilon \sqrt{\nnz(\bv{A})}$ and $\pm \epsilon \norm{\bv A}_F$ when the rows of $\bv A$ can be sampled with probabilities proportional to their sparsities or their squared $\ell_2$ norms respectively. Here $\nnz(\bv{A})$ is the number of non-zero entries in $\bv{A}$ and $\norm{\bv A}_F$ is its Frobenius norm. Even for the strictly easier problems of approximating the singular values or testing the existence of large negative eigenvalues (Bakshi, Chepurko, and Jayaram, FOCS '20), our results are the first that take advantage of non-uniform sampling to give improved error bounds. From a technical perspective, our results require several new eigenvalue concentration and perturbation bounds for matrices with bounded entries. Our non-uniform sampling bounds require a new algorithmic approach, which judiciously zeroes out entries of a randomly sampled submatrix to reduce variance, before computing the eigenvalues of that submatrix as estimates for those of $\bv A$. We complement our theoretical results with numerical simulations, which demonstrate the effectiveness of our algorithms in practice. 
\end{abstract}
\clearpage

\tableofcontents

\pagebreak


\section{Introduction}\label{sec:introduction}

Approximating the eigenvalues of a symmetric matrix is a fundamental problem -- with applications in engineering, optimization, data analysis, spectral graph theory, and beyond. For an $n \times n$ matrix, all eigenvalues can be computed to high accuracy using direct eigendecomposition in $O(n^{\omega})$ time, where $\omega \approx 2.37$ is the exponent of matrix multiplication \cite{demmel2007fast,alman2021refined}. 
When just a few of the largest magnitude eigenvalues are of interest, the power method and other iterative Krylov methods can be applied \cite{saad2011numerical}.
 These methods repeatedly multiply the matrix of interest by query vectors, requiring $O(n^2)$ time per multiplication when the matrix is dense and unstructured.

For large $n$, it is desirable to have even faster eigenvalue approximation algorithms, running in $o(n^2)$ time -- i.e., sublinear in the size of the input matrix.  Unfortunately, for general matrices, no non-trivial approximation can be computed in $o(n^2)$ time: without reading $\Omega(n^2)$ entries, it is impossible to distinguish with reasonable probability if all entries (and hence all eigenvalues) are equal to zero, or if there is a single pair of arbitrarily large entries at positions $(i,j)$ and $(j,i)$, leading to a pair of arbitrarily large eigenvalues. Given this, we seek to address the following question:

\begin{center}
\emph{Under what assumptions on a symmetric $n \times n$ input matrix, can we compute non-trivial approximations to its eigenvalues in $o(n^2)$ time?}
\end{center}
 
It is well known that $o(n^2)$ time eigenvalue computation is possible for highly structured inputs, like tridiagonal or Toeplitz matrices \cite{gu1995divide}. For sparse or structured matrices that admit fast matrix vector multiplication, one can compute a small number of the largest in magnitude eigenvalues in $o(n^2)$ time using iterative methods. Through the use of robust iterative methods, fast top eigenvalue estimation is also possible for matrices that admit fast \emph{approximate} matrix-vector multiplication, such as kernel similarity matrices \cite{greengard1991fast,hardt2014noisy,backurs2021faster}. Our goal is to study simple, sampling-based sublinear time algorithms that work under {much weaker} assumptions on the input matrix.
 
\subsection{Our Contributions}\label{sec:contributions}

Our main contribution is to show  that a very simple algorithm can be used to approximate \emph{all eigenvalues} of any symmetric matrix with \emph{bounded entries}. In particular, for any $\bv A \in \R^{n \times n}$ with maximum entry magnitude $\norm{\bv A}_\infty \le 1$, sampling an $s \times s$ principal submatrix $\bv{A}_S$ of $\bv{A}$ with $s =  \tilde O\left(\frac{\log^3 n}{\epsilon^3}\right )$ and scaling its eigenvalues by $n/s$ yields a $\pm \epsilon n$ additive error approximation to {all eigenvalues of $\bv A$} with good probability.\footnote{Here and throughout, $\tilde O(\cdot)$ hides logarithmic factors in the argument. Note that by scaling, our algorithm gives a $\pm \epsilon n \cdot \norm{\bv A}_\infty$ approximation for  any $\bv A$.}  This result is formally stated below, where $[n] \eqdef \{1,\ldots,n\}$.

\begin{restatable}[Sublinear Time Eigenvalue Approximation]{theorem}{eigvalApprox}
\label{thm:main_bound}
Let $\bv A \in \mathbb{R}^{n\times n}$ be symmetric with $\|\bv A\|_\infty \leq 1$ and eigenvalues $\lambda_1(\bv{A}) \ge \ldots \ge \lambda_n(\bv{A})$. Let $S \subseteq [n]$ be formed by including each index independently with probability $s/n$ as in Algorithm \ref{alg:eigenvalue estimate}. Let $\bv A_S$ be the corresponding principal submatrix of $\bv{A}$, with eigenvalues $\lambda_1(\bv{A}_S) \ge \ldots \ge \lambda_{|S|}(\bv{A}_S)$.

For all $i \in [|S|]$ with $\lambda_i(\bv{A}_S) \ge 0$, let $\tilde \lambda_i(\bv{A}) = \frac{n}{s} \cdot \lambda_i(\bv{A}_S)$. For all $i \in [|S|]$ with $\lambda_i(\bv{A}_S) < 0$, let $\tilde \lambda_{n-(|S|-i)}(\bv{A}) = \frac{n}{s} \cdot \lambda_i(\bv{A}_S)$. For all other $i \in [n]$, let $\tilde \lambda_i(\bv{A}) = 0$. 
If $s \geq \frac{c \log(1/(\epsilon \delta)) \cdot \log^3 n}{\epsilon^3 {\delta}}$, for large enough constant $c$,
 then with probability $\ge 1-\delta$, for all $i \in [n]$,
\begin{align*}
    \lambda_i(\bv A) -\epsilon n \leq \tilde \lambda_i(\bv A) \leq \lambda_i(\bv A) +\epsilon n.
\end{align*}
\end{restatable}

\noindent See Figure \ref{fig:illustration} for an illustration of how the $|S|$ eigenvalues of $\bv{A}_S$ are mapped to estimates for all $n$ eigenvalues of $\bv{A}$.
Note that the principal submatrix $\bv{A}_S$ sampled in Theorem \ref{thm:main_bound} will have $ O(s) = \tilde O\left (\frac{\log^3 n}{\epsilon^3 \delta} \right )$ rows/columns  with high probability. Thus, with high probability, the algorithm reads just $\tilde O \left (\frac{\log^6n}{\epsilon^6 \delta^2}\right )$ entries of $\bv{A}$ and runs in $\poly(\log n, 1/\epsilon,1/\delta)$ time. 
Standard matrix concentration bounds imply that one can sample $O \left (\frac{s \log(1/\delta)}{\epsilon^2} \right )$ random entries from the $O(s) \times O(s)$ random submatrix $\bv A_S$ and preserve its eigenvalues  to error $\pm \epsilon  s$ with probability $1-\delta$ \cite{achlioptas2007fast}. See Appendix \ref{app:entry} for a proof. This  can be directly combined with Theorem \ref{thm:main_bound} to give improved  sample complexity:
\begin{restatable}[Improved Sample Complexity via Entrywise Sampling]{corollary}{entrywise}
\label{cor:entrywise}
Let $\bv A \in \R^{\nbyn}$ be symmetric  with $\|\bv A\|_\infty \leq 1$ and eigenvalues $\lambda_1(\bv A) \ge \ldots \ge \lambda_n(\bv A )$. For any $\epsilon,\delta \in (0,1)$, there is an algorithm that reads $\Tilde O\left ( \frac{\log^3 n}{ \epsilon^5 \delta}\right )$ entries of $\bv A$ and returns, with probability at least $1-\delta$, $\tilde \lambda_i(\bv{A})$ for each $i \in [n]$ satisfying $| \tilde \lambda_i(\bv A) - \lambda_i(\bv{A})| \le \epsilon n$.\end{restatable}
Observe that the dependence on $\delta$ in Theorem \ref{thm:main_bound} and Corollary \ref{cor:entrywise} can be improved via standard arguments: running  the algorithm with failure probability $\delta' = 2/3$, repeating $O(\log(1/\delta))$ times, and taking the median estimate for each $\lambda_i(\bv{A})$. This guarantees that the algorithm will succeed with probability at most $1-\delta$ at the expense of a $\log(1/\delta)$ dependence in the complexity.

\begin{figure}[h]
\vspace{-.5em}
\centering
\begin{tikzpicture}[
point/.style = {circle, draw=black, inner sep=0.07cm,
                ball color=#1,
                node contents={}},
every label/.append style = {font=\small}
                    ]
                  
\draw[latex-latex] (0,5) -- (16,5) ;
\draw[shift={(2,5)},color=black] (0pt,3pt) -- (0pt,-3pt);
\draw[shift={(2,5)},color=black] (0pt,0pt) -- (0pt,-3pt) node[below] 
{${-}n$};
\draw[shift={(8,5)},color=black] (0pt,3pt) -- (0pt,-3pt);
\draw[shift={(8,5)},color=black] (0pt,0pt) -- (0pt,-3pt) node[below] 
{$0$};
\draw[shift={(14,5)},color=black] (0pt,3pt) -- (0pt,-3pt);
\draw[shift={(14,5)},color=black] (0pt,0pt) -- (0pt,-3pt) node[below] 
{$n$};

\draw[latex-latex] (0,3) -- (16,3) ;
\draw[shift={(3,3)},color=black] (0pt,3pt) -- (0pt,-3pt);
\draw[shift={(3,3)},color=black] (0pt,0pt) -- (0pt,-3pt) node[below] 
{${-}s$};
\draw[shift={(8,3)},color=black] (0pt,3pt) -- (0pt,-3pt);
\draw[shift={(8,3)},color=black] (0pt,3pt) -- (0pt,-3pt) node(midpoint)[below] 
{$0$};
\draw[shift={(13,3)},color=black] (0pt,3pt) -- (0pt,-3pt);
\draw[shift={(13,3)},color=black] (0pt,0pt) -- (0pt,-3pt) node[below] 
{$s$};

\node (v1) at (2.3,5) [point=cyan, label=above:$\tilde{\lambda}_n(\bv A)$];
\node (v2) at (13.7,5) [point=cyan, label=above:$\tilde{\lambda}_1(\bv A)$];
\node (v3) at (5.6,5) [point=cyan, label=above:$\tilde{\lambda}_{n-(|S|-p)}(\bv A)$];
\node (v4) at (10.3,5) [point=cyan, label=above:$\tilde{\lambda}_{p-1}(\bv A)$];
\node (v5) at (6.7,5) [point=cyan];
\node (v6) at (7.7,5) [point=cyan];
\node (v7) at (8.3,5) [point=cyan];
\node (v8) at (9.3,5) [point=cyan];


\node (a1) at (4.2,3) [point=red, label=below:$\lambda_{|S|}(\bv A_S)$];
\node (a2) at (12,3) [point=red, label=below:$\lambda_1(\bv A_S)$];
\node (a3) at (7,3) [point=red, label=below:$\lambda_{p}(\bv A_S)$];
\node (a4) at (9,3) [point=red, label=below:$\lambda_{p-1}(\bv A_S)$];
\node (a5) at (8,3) [point=red];

\draw[-Latex](a1) edge (v1);
\draw[-Latex](a2) edge (v2);
\draw[-Latex](a3) edge (v3);
\draw[-Latex](a4) edge (v4);

\draw[-Latex](a5) edge (v5);
\draw[-Latex](a5) edge (v6);
\draw[-Latex](a5) edge (v7);
\draw[-Latex](a5) edge (v8);

  
\draw [decorate,decoration={brace,amplitude=10pt,raise=2ex}]
  (v5) -- (v8) node[midway,yshift=3em]{\small{$\tilde{\lambda}_{t}(\bv{A})\text{ for } t \in (n-(|S|-p+1), p)$}};
  
  

\end{tikzpicture}
\caption{\textbf{Alignment of eigenvalues in Thm.~\ref{thm:main_bound} and Algo.~\ref{alg:eigenvalue estimate}.} We illustrate how the eigenvalues of $\bv{A}_S$, scaled by $\frac{n}{s}$, are used to approximate all eigenvalues of $\bv A$. If $\bv{A}_S$ has $p-1$ positive eigenvalues, they are set to the top $p-1$ eigenvalue estimates. Its $|S|-p+1$ negative eigenvalues are set to the bottom eigenvalue estimates. All remaining eigenvalues are simply approximated as zero. 
}
\label{fig:illustration}
\end{figure}
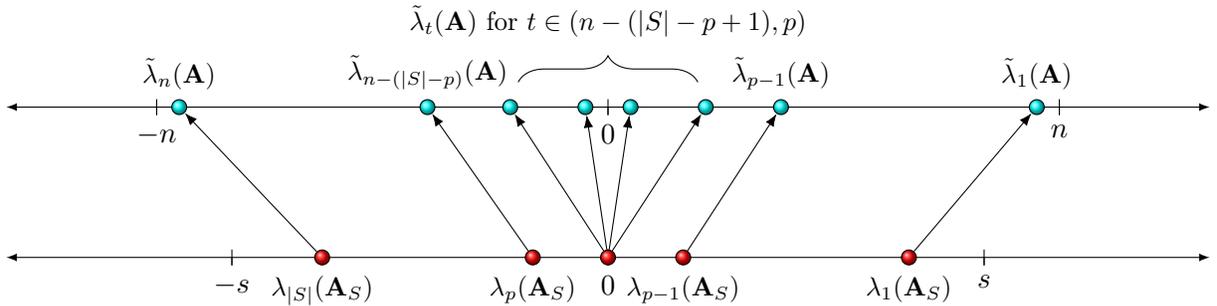

\noindent{\textbf{Comparison to known bounds.}} 
Theorem \ref{thm:main_bound} can be viewed as a concentration inequality on the full eigenspectrum of a random principal submatrix $\bv{A}_S$ of $\bv{A}$. This significantly extends prior work, which was able to bound just the spectral norm (i.e., the magnitude of the top eigenvalue) of a random principal submatrix \cite{rudelson2007sampling,tropp2008norms}. Bakshi, Chepurko, and Jayaram \cite{BakshiChepurkoJayaram:2020} recently identified developing such full eigenspectrum concentration inequalities as an important step in expanding our knowledge of sublinear time property testing algorithms for bounded entry matrices. 

Standard matrix concentration bounds \cite{gittens2011tail} can be used to show that the \emph{singular values} of $\bv{A}$ (i.e., the magnitudes of its eigenvalues) are approximated by those of a $O \left (\frac{\log n }{\epsilon^2} \right) \times O \left (\frac{\log n }{\epsilon^2} \right)$ random submatrix (see Appendix \ref{sec:singular}) with independently sampled rows and columns.  However, such a random matrix will not be symmetric or even have real eigenvalues in general, and thus no analogous bounds were previously known for the eigenvalues themselves.

Recently, Bakshi, Chepurko, and Jayaram \cite{BakshiChepurkoJayaram:2020} studied the closely related problem of testing positive semidefiniteness in the bounded entry model. They show how to test whether the minimum eigenvalue of $\bv A$ is either greater than $0$ or smaller than $-\epsilon n$ by reading just  $\Tilde{O}(\frac{1}{\epsilon^2})$ entries. They show that this result is optimal in terms of query complexity, up to logarithmic factors. Like our approach, their algorithm is based on random principal submatrix sampling. Our eigenvalue approximation guarantee strictly strengthens the testing guarantee -- given $\pm \epsilon n$ approximations to all eigenvalues, we immediately solve the testing problem. Thus, our query complexity is tight up to a $\poly(\log n,1/\epsilon)$ factor. It is open if  our higher sample complexity is necessary to solve the harder full eigenspectrum estimation problem. See Section \ref{optimal} for further discussion.

\medskip

\noindent\textbf{Improved bounds for non-uniform sampling.} 
Our second main contribution is to show that, when it is possible to efficiently sample rows/columns of $\bv A$ with probabilities proportional to their sparsities or their squared $\ell_2$ norms, significantly stronger eigenvalue estimates can be obtained. In particular, letting $\nnz(\bv A)$  denote the number of nonzero entries in $\bv A$ and $\norm{\bv A}_F$ denote its Frobenius norm,
we show that sparsity-based sampling yields eigenvalue estimates with error $\pm \epsilon \sqrt{\nnz(\bv{A})}$ and norm-based sampling gives error $\pm \epsilon \norm{\bv A}_F$. See Theorems \ref {thm:nnz_main_bound} and \ref{thm:l2_main_bound} for formal statements. Observe that when $\|\bv A\|_\infty \leq 1$, its eigenvalues are bounded in magnitude by $\norm{\bv A}_2 \le \norm{\bv A}_F \le \sqrt{\nnz(\bv A)} \le n$. Thus, Theorems \ref{thm:nnz_main_bound} and \ref{thm:l2_main_bound} are natural strengthenings of  Theorem \ref{thm:main_bound}. Row norm-based sampling (Theorem \ref{thm:l2_main_bound}) additionally removes the bounded entry requirement of Theorems \ref{thm:main_bound} and \ref{thm:nnz_main_bound}.

As discussed in Section \ref{sec:sparsityOverview}, sparsity-based sampling can be performed in sublinear time when $\bv{A}$ is stored in a slightly augmented sparse matrix format or when $\bv{A}$ is the adjacency matrix of a graph accessed in the standard graph query model of the sublinear  algorithms literature \cite{GoldreichRon:1997}. Norm-based sampling  can also be performed efficiently with an augmented matrix format, and is commonly studied  in randomized and  `quantum-inspired'  algorithms for linear algebra~\cite{FriezeKannanVempala:2004,Tang:2018}. 

\begin{restatable}[Sparse Matrix Eigenvalue Approximation]{theorem}{nnzEigVal}
\label{thm:nnz_main_bound}
Let $\bv A \in \mathbb{R}^{n\times n}$ be symmetric with $\|\bv A\|_\infty \leq 1$ and eigenvalues $\lambda_1(\bv{A}) \ge \ldots \ge \lambda_n(\bv{A})$. Let $S \subseteq [n]$ be formed by including the $i$\textsuperscript{th} index independently with probability $p_i=\min \left (1,\frac{s\nnz(\bv{A}_i)}{\nnz(\bv{A})}\right)$ as in Algorithm~\ref{alg:nnz eigenvalue estimate}. Here $\nnz(\bv{A}_i)$ is the number of non-zero entries in the $i^{th}$ row of $\bv A$. Let $\bv A_S$ be the corresponding principal submatrix of $\bv{A}$, and let $\tilde \lambda_i(\bv{A})$ be the estimate of $\lambda_i(\bv{A})$ computed from $\bv{A}_S$ as in Algorithm~\ref{alg:nnz eigenvalue estimate}. If $s \geq \frac{c\log^8 n}{\epsilon^8\delta^4}$, for large enough constant $c$,
 then with probability $\ge 1-\delta$, for all $i \in [n]$, $|\tilde  \lambda_i(\bv A) - \lambda_i(\bv A) | \le \epsilon  \sqrt{\nnz(\bv{A})}$.
\end{restatable}

\begin{restatable}[Row Norm Based Matrix Eigenvalue Approximation]{theorem}{rownormSamp}
\label{thm:l2_main_bound}
Let $\bv A \in \mathbb{R}^{n\times n}$ be symmetric and eigenvalues $\lambda_1(\bv{A}) \ge \ldots \ge \lambda_n(\bv{A})$. Let $S \subseteq [n]$ be formed by including the $i$\textsuperscript{th} index independently with probability $p_i=\min \left (1,\frac{s\norm{\bv A_i}^2_2}{\norm{\bv A}^2_F}+\frac{1}{n^2}\right )$ as in Algorithm~\ref{alg:l2_eig_est}. Here $\norm{\bv A_i}_2$ is the $\ell_2$ norm of the $i^{th}$ row of $\bv A$. Let $\bv A_S$ be the corresponding principal submatrix of $\bv{A}$, and let $\tilde \lambda_i(\bv{A})$ be the estimate of $\lambda_i(\bv{A})$ computed from $\bv{A}_S$ as in Algorithm~\ref{alg:l2_eig_est}. If $s \geq  \frac{c \log^{10} n }{\epsilon^8\delta^4}$, for large enough constant $c$, then with probability $\ge 1-\delta$, for all $i \in [n]$, $|\tilde  \lambda_i(\bv A) - \lambda_i(\bv A) | \le \epsilon \norm{\bv{A}}_F.$
\end{restatable}
The above non-uniform sampling theorems immediately yield algorithms for testing the presence of a negative eigenvalue with magnitude at least $\epsilon \sqrt{\nnz(\bv{A})}$ or $\epsilon \norm{\bv A}_F$ respectively, strengthening the testing results of \cite{BakshiChepurkoJayaram:2020}, which require eigenvalue magnitude at least $\epsilon n$.  In the graph property testing literature, there is a rich line of work exploring the testing of bounded degree or sparse graphs \cite{GoldreichRon:1997,BenjaminiSchrammShapira:2010}. Theorem \ref{thm:nnz_main_bound} can be thought of as  first step in establishing a related theory of sublinear time approximation algorithms and property testers for sparse matrices.

Surprisingly, in the non-uniform sampling case, the eigenvalue estimates derived from $\bv{A}_S$ cannot simply be its scaled eigenvalues, as in Theorem \ref{thm:main_bound}. 
 E.g., when $\bv{A}$ is the identity, our row sampling probabilities are uniform in all cases. However, the scaled submatrix $\frac{n}{s} \cdot \bv{A}_S$ will be a scaled identity, and have eigenvalues equal to $n/s$ -- failing to give a $\pm \epsilon \sqrt{\nnz(\bv A)} = \pm \epsilon \norm{\bv A}_F = \pm \epsilon \sqrt{n}$ approximation to the true eigenvalues (all of which are $1$) unless $s \gtrsim \frac{\sqrt{n}}{\epsilon}$. To handle this, and related cases, we must argue that selectively zeroing out entries in sufficiently low probability rows/columns of $\bv{A}$ (see Algorithms \ref{alg:nnz eigenvalue estimate} and \ref{alg:l2_eig_est}) does not significantly change the spectrum, and ensures concentration of the submatrix eigenvalues. It is  not hard to see that simple random submatrix sampling fails even for the easier problem of singular value estimation. Theorems \ref{thm:nnz_main_bound} and \ref{thm:l2_main_bound}  give the first results of their kinds for this problem as well. 

\subsection{Related Work}\label{sec: related work}

Eigenspectrum estimation is a key primitive in numerical linear algebra, typically known as \emph{spectral density estimation}. The eigenspectrum is viewed as a distribution with mass $1/n$ at each of the $n$ eigenvalues, and the goal is to approximate this distribution \cite{weisse2006kernel,lin2016approximating}. Applications include identifying motifs in social networks \cite{dong2019network}, studying Hessian and weight matrix spectra in deep learning \cite{sagun2016eigenvalues,yao2018hessian,ghorbani2019investigation}, `spectrum splitting' in parallel eigensolvers \cite{li2019eigenvalues}, and the study of many systems in experimental physics and chemistry \cite{wang1994calculating,silver1994densities,helsen2019spectral}.

Recent work has studied sublinear time spectral density estimation for graph structured matrices -- Braverman, Krishnan, and Musco \cite{braverman2021linear} show that the spectral density of a normalized graph adjacency or Laplacian matrix can be estimated to $\epsilon$ error in the Wasserstein distance in $\tilde O(n/\poly(\epsilon))$ time. Cohen-Steiner, Kong, Sohler, and Valiant study a similar setting, giving runtime $2^{O(1/\epsilon)}$ \cite{cohen2018approximating}. We note that the additive error eigenvalue approximation result of Theorem \ref{thm:main_bound} (analogously Theorems \ref{thm:nnz_main_bound} and \ref{thm:l2_main_bound}) directly gives an $\epsilon n$ approximation to the spectral density in the Wasserstein distance -- extending the above results to a much broader class of matrices. When $\norm{\bv A}_\infty \le 1$, $\bv{A}$ can have eigenvalues as large as $n$, while the normalized adjacency matrices studied in \cite{cohen2018approximating,braverman2021linear} have eigenvalues in $[-1,1]$. So, while the results are not directly comparable, our Wasserstein error can be thought as on order of their error of $\epsilon$ after scaling. 

Our work is also closely related to a line of work on sublinear time property testing for bounded entry matrices, initiated by Balcan et al.~\cite{balcan2019testing}. In that work, they study testing of rank, Schatten-$p$ norms, and several other global spectral properties. Sublinear time testing algorithms for the rank and other properties have also been studied under low-rank and bounded row norm assumptions on the input matrix \cite{krauthgamer2003property,li2014improved}. Recent work studies positive semidefiniteness testing and eigenvalue estimation in the matrix-vector query model, where each query computes $\bv{A} \bv{x}$ for some $\bv{x} \in \R^{n \times n}$. As in Theorem \ref{thm:l2_main_bound},  $\pm \epsilon \norm{\bv A}_F$ eigenvalue estimation can be achieved with $\poly(\log n,1/\epsilon)$ queries  in this model \cite{Needell:2022uq}. Finally, several works  study streaming algorithms for eigenspectrum approximation
\cite{andoni2013eigenvalues,li2014sketching,li2016approximating}. These algorithms are not sublinear time -- they require at least linear time to process the input matrix. However, they  use sublinear working memory. Note that Theorem \ref{thm:main_bound} immediately gives a sublinear space streaming algorithm for eigenvalue estimation. We can simply store the sampled submatrix $\bv{A}_S$ as its entries are  updated. 

\subsection{Technical Overview}

In this section, we overview the main techniques used to prove Theorems~\ref{thm:main_bound}, and then how these techniques are extended to prove Theorems \ref{thm:nnz_main_bound} and \ref{thm:l2_main_bound}. We start by defining a  decomposition of any symmetric $\bv A$ into the sum of two matrices containing its large and small magnitude eigendirections.

\begin{definition}[Eigenvalue Split]\label{def:split}
Let $\bv{A} \in \R^{n \times n}$ be symmetric. For any $\epsilon,\delta \in (0,1)$, let $\bv A_o= \bv V_o\bv \Lambda_o \bv V_o^T$ where $\bv \Lambda_o$ is diagonal, with the eigenvalues of $\bv A$ with magnitude $\ge \epsilon \sqrt{\delta} n$ on its diagonal, and  $\bv{V}_o$ has the corresponding eigenvectors as columns. Similarly, let $\bv A_m=\bv V_m\bv \Lambda_m \bv V_m^T$ where $\bv \Lambda_m$ has the eigenvalues of $\bv A$ with magnitude $ <\epsilon \sqrt{\delta} n$ on its diagonal and $\bv{V}_m$ has the corresponding eigenvectors as columns. Then, $\bv A$ can be decomposed as $$\bv A = \bv A_o + \bv A_m=  \bv V_o\bv \Lambda_o\bv V_o^T + \bv V_m\bv \Lambda_m\bv V_m^T.$$
Any principal submatrix of $\bv A$, $\bv{A}_S$, can be similarly written as $$\bv{A}_S = \bv{A}_{o,S} + \bv{A}_{m,S}=\bv V_{o,S}\bv\Lambda_{o}\bv V_{o,S}^T+\bv V_{m,S}\bv\Lambda_{m}\bv V_{m,S}^T,$$ where $\bv V_{o,S},\bv V_{m,S}$ are the corresponding submatrices obtained by sampling rows of $\bv{V}_o,\bv{V}_m$.
\end{definition}

Since $\bv A_S$, $\bv A_{m,S}$ and $\bv A_{o,S}$ are all symmetric, we can use Weyl's eigenvalue perturbation theorem~\cite{weyl1912asymptotic} to show that for all eigenvalues of $\bv{A}_S$, 
\begin{align}\label{eq:wel}\lvert \lambda_i(\bv A_S) -\lambda_i(\bv A_{o,S}) \rvert \leq \|\bv A_{m,S} \|_2.
\end{align}

We will argue that the eigenvalues of $\bv{A}_{o,S}$ approximate those of $\bv{A}_o$ -- i.e. all eigenvalues of  $\bv{A}$ with magnitude $\ge \epsilon \sqrt{\delta} n$. Further, we will show that $\norm{\bv{A}_{m,S}}_2$ is small with good probability. Thus, via \eqref{eq:wel}, the eigenvalues of $\bv{A}_S$ approximate those of $\bv{A}_o$. In the estimation procedure of Theorem \ref{thm:main_bound},  all other  small magnitude eigenvalues of $\bv{A}$ are estimated to be $0$, which will immediately give our $\pm \epsilon n$ approximation bound when the original eigenvalue has magnitude $\le \epsilon n$.

\medskip

\noindent\textbf{Bounding the eigenvalues of $\bv A_{o,S}$.} The first step is to show that the eigenvalues of $\bv{A}_{o,S}$ well-approximate those of $\bv{A}_o$. As in \cite{BakshiChepurkoJayaram:2020}, we critically use  that the eigenvectors corresponding to large eigenvalues are \emph{incoherent} -- intuitively, since $\norm{\bv{A}}_\infty$ is bounded, their mass must be spread out in order to witness a large eigenvalue. Specifically,  \cite{BakshiChepurkoJayaram:2020} shows that for any eigenvector $\bv{v}$ of $\bv{A}$ with corresponding eigenvalue $\ge \epsilon \sqrt{\delta} n$, $\norm{\bv{v}}_\infty \le \frac{1}{\epsilon \sqrt{\delta} \cdot \sqrt{n}}$. We give related bounds on the Euclidean norms of the rows of $\bv{V}_o$ (the \emph{leverage scores} of $\bv{A}_o$), and on these rows after  weighting by $\bv \Lambda_o$.

Using these incoherence bounds, we argue that the eigenvalues of $\bv{A}_{o,S}$ approximate those of $\bv{A}_{o}$ up to $\pm \epsilon n$ error. A key idea is to bound the eigenvalues of $\bv{\Lambda}_o^{1/2} \bv V_{o,S}^T \bv V_{o,S} \bv{\Lambda}_o^{1/2}$, which are identical to the non-zero eigenvalues of $\bv{A}_{o,S} = \bv V_{o,S} \bv{\Lambda}_o \bv V_{o,S}^T$. Via a matrix Bernstein bound and our incoherence bounds on $\bv V_{o}$, we show that this matrix is close to  $\bv{\Lambda}_o$ with high probability. However, since $\bv{\Lambda}_o^{1/2}$ may be complex, the matrix is \emph{not necessarily Hermitian} and standard perturbation bounds~\cite{SS1990,HJ2012} do not apply. Thus, to derive an eigenvalue bound, we apply a perturbation bound of Bhatia \cite{bhatia2013matrix}, which generalizes Weyl's inequality to the non-Hermitian case, with a $\log n$ factor loss. To the best of our knowledge, this is the first time that perturbation theory bounds for non-Hermitian matrices have been used to prove improved algorithmic results in the theoretical computer science literature.

We note that in Appendix \ref{app:alt_bound}, we give an alternate bound, which instead analyzes the Hermitian matrix $(\bv{V}_{o,S}^T  \bv{V}_{o,S})^{1/2} \bv{\Lambda}_{o} (\bv{V}_{o,S}^T \bv{V}_{o,S})^{1/2}$, whose eigenvalues are again identical to those of $\bv{A}_{o,S}$. This approach only requires Weyl's inequality, and yields an overall bound of $s = O \left (\frac{\log n}{\epsilon^4 \delta} \right )$, improving the $\log n$ factors of Theorem \ref{thm:main_bound} at the cost of worse $\epsilon$ dependence.

\medskip

\noindent\textbf{Bounding the spectral norm of $\bv A_{m,S}$.} 
The next step is to show that all eigenvalues of $\bv A_{m,S}$ are small provided a sufficiently large submatrix is sampled. This means that the ``middle'' eigenvalues of $\bv A$, \emph{i.e.} those with magnitude $\le \epsilon \sqrt{\delta} n$ do not contribute much to any eigenvalue $\lambda_i(\bv A_S)$. To do so, we apply a theorem of \cite{rudelson2007sampling,tropp2008norms} which shows concentration of the spectral norm of a uniformly random submatrix of an entrywise bounded matrix
Observe that while $\norm{\bv{A}}_\infty \le 1$, such a bound will not in general hold for $\norm{\bv{A}_m}_\infty$. Nevertheless, we can use the incoherence of $\bv{V}_o$ to show that $\norm{\bv A_o}_\infty$ is bounded, which via triangle inequality, yields a bound on $\norm{\bv{A}_m}_\infty \le \norm{\bv{A}}_\infty+ \norm{\bv{A}_o}_\infty$. In the end, we show  that if  $s \geq {O}(\frac{ \log n}{\epsilon^2 \delta})$, with probability at least $1-\delta$, $\|\bv A_{m,S} \|_2 \leq \epsilon s$. After the $n/s$ scaling in the estimation procedure of Theorem \ref{thm:main_bound}, this spectral norm bound translates into an additive $\epsilon n$ error in approximating the eigenvalues of $\bv{A}$. 
\medskip

\noindent\textbf{Completing the argument.}
Once we establish the above bounds on $\bv{A}_{o,S}$ and $\bv{A}_{m,S}$, Theorem \ref{thm:main_bound} is essentially complete. Any eigenvalue in $\bv{A}$ with magnitude $\ge \epsilon n$ will correspond to a nearby eigenvalue in $\frac{n}{s} \cdot \bv{A}_{o,S}$ and in turn, $\frac{n}{s} \cdot \bv{A}_S$ given our spectral norm bound on $\bv{A}_{m,S}$. An eigenvalue in $\bv{A}$ with magnitude $\le \epsilon n$ may or may not correspond to a nearby by eigenvalue in $\bv{A}_{o,S}$ (it will only if it lies in the range $[\epsilon \sqrt{\delta} n,\epsilon n]$). In any case however, in the estimation procedure of Theorem \ref{thm:main_bound}, such an eigenvalue will either be estimated using a small eigenvalue of $\bv{A}_S$, or be estimated as $0$. In both instances, the estimate will give $\pm \epsilon n$ error, as required.

\medskip

\noindent\textbf{Can we beat additive error?} It is natural to ask if our approach can be improved to yield sublinear time algorithms with stronger relative error approximation guarantees for $\bv{A}$'s eigenvalues. Unfortunately, this is not possible -- consider a matrix with just a single pair of entries $\bv{A}_{i,j},\bv{A}_{j,i}$ set to $1$. To obtain relative error approximations to the two non-zero eigenvalues, we must find the pair $(i,j)$, as otherwise we cannot distinguish $\bv{A}$ from the all zeros matrix. This requires reading a $\Omega(n^2)$ of $\bv{A}$'s entries. More generally, consider $\bv{A}$ with a random $n/t \times n/t$ principal submatrix populated by all $1$s, and with all other entries equal to $0$. $\bv{A}$ has largest eigenvalue $n/t$. However, if we read $s \ll t^2$ entries of $\bv{A}$, with good probability, we will not see even a single one, and thus we will not be able to distinguish $\bv{A}$ from the all zeros matrix. This example establishes that any sublinear time algorithm with query complexity $s$ must incur additive error at least $\Omega(n/\sqrt{s})$.

\subsubsection{Improved Bounds via Non-Uniform Sampling}\label{sec:sparsityOverview}

We now discuss how to give improved approximation bounds via non-uniform sampling. We focus  on the $\pm \epsilon \sqrt{\nnz(\bv{A})}$  bound of Theorem \ref{thm:nnz_main_bound} using sparsity-based sampling. The proof of Theorem \ref{thm:l2_main_bound} for row norm sampling follows the same general ideas, but with some additional complications.

 Theorem \ref{thm:nnz_main_bound} requires sampling a submatrix $\bv{A}_S$, where each index $i$  is included in $S$ with probability $p_i = \min (1, \frac{s \nnz(\bv{A}_i)}{\nnz(\bv{A})})$. We reweight each sampled row by $\frac{1}{\sqrt{p_i}}$. Thus, if entry $\bv{A}_{ij}$ is sampled, it is scaled by $\frac{1}{\sqrt{p_i \cdot p_j}}$. When the rows have uniform sparsity (so all $p_i = s/n$), this ensures that the full submatrix is scaled  by $n/s$, as in Theorem \ref{thm:main_bound}.

The proof of Theorem \ref{thm:nnz_main_bound} follows the same outline as that of Theorem \ref{thm:main_bound}: we first argue that the outlying eigenvectors in $\bv{V}_o$ are incoherent, giving a bound on the norm of each row of $\bv{V}_o$ in terms of $\nnz(\bv A_i)$. We then apply a matrix Bernstein bound and Bhatia's non-Hermitian eigenvalue perturbation bound to show that the eigenvalues of $\bv{A}_{o,S}$ approximate those of $\bv{A}_o$ up to $\pm \epsilon \sqrt{\nnz(\bv{A})}$.

\medskip

\noindent\textbf{Bounding the spectral norm of $\bv{A}_{m,S}$.}
The major challenge is showing that the subsampled middle eigendirections do not significantly increase the approximation error by bounding the $\norm{\bv{A}_{m,S}}_2$ by $\epsilon \sqrt{\nnz(\bv{A})}$. This is difficult since the indices in $\bv{A}_{m,S}$ are sampled nonuniformly, so existing bounds \cite{tropp2008norms} on the spectral norm of uniformly random submatrices do not apply. We extend these bounds to the non-uniform sampling case, but still face an issue due to the rescaling of entries by $\frac{1}{\sqrt{p_i p_j}}$. In fact, without additional algorithmic modifications, $\norm{\bv{A}_{m,S}}_2$ is simply not bounded by $\epsilon \sqrt{\nnz(\bv{A})}$! For example, as already discussed, if $\bv{A} = \bv{I}$ is the identity matrix, we get $\bv{A}_{m,S} = \frac{n}{s} \cdot \bv{I}$ and so $\norm{\bv{A}_{m,S}}_2 = \frac{n}{s} >\epsilon \sqrt{\nnz(\bv{A})}$, assuming $s < \frac{\sqrt{n}}{\epsilon}$. 
Relatedly, suppose that $\bv{A}$ is tridiagonal, with zeros on the diagonal and ones on the first diagonals above and below the main diagonal. Then, if $s \ge \sqrt{n}$, with constant probability, one of the ones will be sampled and scaled by $\frac{n}{s}$. Thus, we will again have $\norm{\bv{A}_{m,S}}_2 \ge \frac{n}{s} \ge \epsilon \sqrt{\nnz(\bv A)}$, assuming $s < \frac{\sqrt{n}}{2\epsilon}$. Observe that this issue arrises even when trying to approximate just the singular values (the eigenvalue magnitudes) via sampling. Thus, while an analogous bound to the uniform sampling result of Theorem \ref{thm:main_bound} can easily be given for singular value estimation via matrix concentration inequalities (see Appendix \ref{sec:singular}), to the best of our knowledge, Theorems \ref{thm:nnz_main_bound} and \ref{thm:l2_main_bound} are the first of their kind even for singular value estimation.

\medskip

\noindent\textbf{Zeroing out entries in sparse rows/columns.}
To handle the above cases, we prove a novel perturbation bound, arguing that if we zero out any entry $\bv{A}_{ij}$ of $\bv{A}$ where $\sqrt{\nnz(\bv{A}_i) \cdot \nnz(\bv{A}_j)} \le  \frac{\epsilon \sqrt{\nnz(\bv{A})}}{c\log n}$, then the eigenvalues of $\bv{A}$ are not perturbed by more than $\epsilon \sqrt{\nnz(\bv{A})}$. This can be thought of as a strengthening of Girshgorin's circle theorem, which would ensure that zeroing out entries in rows/columns with $\nnz(\bv{A}_i) \le \epsilon \sqrt{\nnz(\bv A)}$ does not perturb the eigenvalues by more than $\epsilon \sqrt{\nnz(\bv A)}$. Armed with this perturbation bound, we argue that if we zero out the appropriate entries of $\bv{A}_{S}$ before computing its eigenvalues, then since we have removed  entries in very sparse rows and columns which would be scaled by a large $\frac{1}{\sqrt{p_i p_j}}$ factor in $\bv{A}_S$, we can bound $\norm{\bv{A}_{m,S}}_2$. This requires relating the magnitudes of the entries in $\bv{A}_{m,S}$ to those in $\bv{A}_S$ using the incoherence of the top eigenvectors, which gives bounds on the entries of $\bv{A}_{o,S} = \bv{A}_S - \bv{A}_{m,S}$. 

\medskip

\noindent \textbf{Sampling model.} We note that the sparsity-based sampling of Theorem \ref{thm:nnz_main_bound} can be efficiently implemented in several natural settings.
Given a matrix stored in sparse format, i.e., as a list of nonzero entries, we can easily sample a row with probability $\frac{\nnz(\bv{A}_i)}{\nnz(\bv{A})}$ by sampling a uniformly random non-zero entry and looking at its corresponding row. Via standard techniques, we can convert several such samples into a sampled set $S$ close in distribution to having each $i \in [n]$ included independently with probability $\min \left (1, \frac{s \nnz(\bv{A}_i)}{\nnz(\bv{A})}\right)$. If we store the values of $\nnz(\bv A), \nnz(\bv{A}_1),\ldots,\nnz(\bv{A}_n)$, we can also efficiently access each $p_i$, which is needed for rescaling and zeroing out entries. Also observe that if $\bv A$ is the adjacency matrix of a graph, in the standard graph query model \cite{GoldreichRon:1997}, it is well known how to approximately count edges and sample them uniformly at random, i.e., compute $\nnz(\bv{A})$ and sample its nonzero entries, in sublinear time~\cite{GoldreichRon:2008,EdenRosenbaum:2018}.  Further, it is typically assumed that one has access to the node degrees, i.e., $\nnz(\bv{A}_1),\ldots,\nnz(\bv{A}_n)$. Thus, our algorithm can naturally be used to estimate spectral graph properties in sublinear time.

The $\ell_2$ norm-based sampling of Theorem \ref{thm:l2_main_bound} can also be performed efficiently using an augmented data structure for storing $\bv A$. Such data structures have been used extensively in the literature on quantum-inspired algorithms, and require just $O(\nnz(\bv A))$ time to construct, $O(\nnz(\bv A))$ space, and $O(\log n)$  time to update give an update to an entry of $\bv A$ \cite{Tang:2018,Chepurko:2020tj}.

\subsection{Towards Optimal Query Complexity}\label{optimal} 

As discussed, Bakshi et al.~\cite{BakshiChepurkoJayaram:2020} show that any algorithm which can test with good probability whether $\bv{A}$ has an eigenvalue $\le -\epsilon n$ or else has all non-negative eigenvalues must read $\tilde \Omega \left (\frac{1}{\epsilon^2} \right )$ entries of $\bv{A}$.
This testing problem is strictly easier than outputting $\pm \epsilon n$ error estimates of all eigenvalues, so gives a lower bound for our setting.
If the queried entries are restricted to fall in a submatrix, \cite{BakshiChepurkoJayaram:2020} shows that this 
 submatrix must have dimensions $ \Omega \left (\frac{1}{\epsilon^2}\right ) \times \Omega \left (\frac{1}{\epsilon^2}\right )$, giving total query complexity $\Omega \left (\frac{1}{\epsilon^4} \right)$.
 Closing the gap between our upper bound of $\tilde O\left (\frac{\log^3 n}{\epsilon^3} \right ) \times  \tilde O\left (\frac{\log^3 n}{\epsilon^3} \right )$ and the lower bound  of $\Omega \left (\frac{1}{\epsilon^2} \right ) \times \Omega \left (\frac{1}{\epsilon^2} \right )$ for submatrix queries is an intriguing open question.  
 
 We show in Appendix \ref{app:psd} that this gap can be easily closed via a surprisingly simple argument if  $\bv{A}$ is positive semidefinite (PSD). In that case, $\bv{A} = \bv{B} \bv{B}^T$ with $\bv{B} \in \R^{n \times n}$. Writing $\bv{A}_S = \bv{S}^T \bv{A} \bv{S}$ for a sampling matrix $\bv{S} \in \R^{n \times |S|}$,  the non-zero eigenvalues of $\bv{A}_S$ are identical to those of $\bv{B}\bv{S}\bv{S}^T \bv{B}^T$. Via a standard approximate matrix multiplication analysis \cite{drineas2001fast}, one can then show that, for $s \ge \frac{1}{\epsilon^2 \delta}$, with probability at least $1-\delta$, $\norm{\bv{BB}^T - \bv{B}\bv{S}\bv{S}^T \bv{B}}_F \le \epsilon n$. Via Weyl's  inequality, this shows that the eigenvalues of $\bv{B}\bv{S}\bv{S}^T \bv{B}$, and hence $\bv{A}_S$, approximate those of $\bv{A}$ up to $\pm \epsilon n$ error.\footnote{In fact, via more refined eigenvalue perturbation bounds~\cite{bhatia2013matrix} one can show an $\ell_2$ norm bound on the eigenvalue approximation errors, which can be much stronger than the $\ell_\infty$ norm bound of Theorem~\ref{thm:main_bound}.}
  
 Unfortunately, this approach breaks down when $\bv{A}$ has negative eigenvalues, and so cannot be factored as $\bv{BB}^T$ for real $\bv{B} \in \R^{n \times n}$. This is more than a technical issue: observe that when $\bv{A}$ is PSD and has $\norm{\bv{A}}_\infty \le 1$, it can have at most $1/\epsilon$ eigenvalues larger than $\epsilon n$ -- since its trace, which is equal to the sum of its eigenvalues, is bounded by $n$, and since all eigenvalues are non-negative. When $\bv{A}$ is not PSD, it can have $\Omega(1/\epsilon^2)$ eigenvalues with magnitude larger than $\epsilon n$. In particular, if $\bv{A}$ is the tensor product of a $1/\epsilon^2 \times 1/\epsilon^2$ random $\pm 1$ matrix and the $ \epsilon^2 n \times \epsilon^2 n$ all ones matrix, the bulk of its eigenvalues (of which there are $1/\epsilon^2$) will concentrate around $1/\epsilon \cdot \epsilon^2 n  = \epsilon n$. As a result it remains unclear whether we can match the $1/\epsilon^2$ dependence of the PSD case, or if a stronger lower bound can be shown for indefinite matrices. 
 
Outside the $\epsilon$ dependence, it is unknown if full eigenspectrum approximation can be performed with sample complexity independent of the matrix size $n$. \cite{BakshiChepurkoJayaram:2020} achieve this for the easier positive semidefiniteness testing problem, giving sample complexity $\tilde O(1/\epsilon^2)$. However our bounds have additional $\log n$ factors. As discussed, in Appendix \ref{app:alt_bound} we give an alternate analysis for Theorem \ref{thm:main_bound}, which shows that sampling a $ O \left (\frac{\log n}{\epsilon^4 \delta} \right ) \times  O \left (\frac{\log n}{\epsilon^4 \delta} \right )$ submatrix suffices for $\pm \epsilon n$ eigenvalue approximation, saving a $\log^2 n$ factor at the cost of worse $\epsilon$ dependence. However, removing the final $\log n$ seems difficult -- it arises when bounding  $\norm{\bv A_{m,S}}_2$ via bounds on the spectral norms of random principal submatrices \cite{rudelson2007sampling}. Removing it seems as though it would require either improving such bounds, or taking a different algorithmic approach.
   
Also note that our $\log n$ and $\epsilon$ dependencies for non-uniform sampling (Theorems \ref{thm:nnz_main_bound} and \ref{thm:l2_main_bound}) are likely not tight. It is not hard to check that the lower bounds of \cite{BakshiChepurkoJayaram:2020} still hold in these settings. For example, in the sparsity-based sampling setting, by simply having the matrix entirely supported on a $\sqrt{\nnz(\bv{A})} \times \sqrt{\nnz(\bv{A})}$ submatrix, the lower bounds of \cite{BakshiChepurkoJayaram:2020} directly carry over. Giving tight query complexity bounds here would also be interesting. Finally, it would  be interesting to go beyond principal submatrix based algorithms, to achieve improved query complexity, as in Corollary \ref{cor:entrywise}. Finding an algorithm matching the $\tilde O\left (\frac{1}{\epsilon^2} \right )$ overall query complexity lower bound of \cite{BakshiChepurkoJayaram:2020} is open even in the much simpler PSD setting.
  
\section{Notation and Preliminaries}

We now define notation and foundational results that we use throughout our work. For any integer $n$, let $[n]$ denote the set $\{1,2,\ldots,n\}$. 
We write matrices and vectors in bold literals -- e.g., $\bv A$ or $\bv{x}$. We denote the eigenvalues of a symmetric matrix $\bv{A} \in \R^{n \times n}$ by $\lambda_1(\bv A) \ge \ldots \ge \lambda_n(\bv{A})$, in decreasing order.
A symmetric matrix is positive semidefinite if all its eigenvalues are non-negative. For two matrices $\bv A, \bv B$, we let $\bv A \succeq \bv B$ denote that $\bv A - \bv B$ is positive semidefinite.
For any matrix $\bv A \in \R^{n\times n}$ and $i\in [n]$, we let $\bv{A}_i$ denote the $i^{th}$ row of $\bv{A}$, $\nnz(\bv A_i)$ denote the number of non-zero elements in this row, and $\norm{\bv A_i}_2$ denote its $\ell_2$ norm. We let $\nnz(\bv A)$ denote the total number of non-zero elements $\bv A$. For a vector $\bv{x}$, we let $\norm{\bv{x}}_2$ denote its Euclidean norm. For a matrix $\bv{A}$, we let $\|\bv A\|_{\infty}$ denote the largest magnitude of an entry, $\|\bv A\|_2 = \max_{\bv{x}} \frac{\norm{\bv{Ax}}_2}{\norm{\bv{x}}_2}$ denote the spectral norm, $\|\bv A\|_F = (\sum_{i,j} \bv A_{ij}^2)^{1/2}$ denote the Frobenius norm, and $\|\bv A\|_{1 \rightarrow 2}$ denote the maximum Euclidean norm of a column. For $\bv{A} \in \R^{n \times n}$ and $S \subseteq [n]$ we let $\bv{A}_S$ denote the principal submatrix corresponding to $S$. We let $\mathbb{E}_2$ denote the $L_2$ norm of a random variable, $\mathbb{E}_2[X] = (\mathbb{E}[X^2])^{1/2}$, where $\E[\cdot]$ denotes expectation.

We use the following basic facts and identities on eigenvalues throughout our proofs. 

\begin{fact}[Eigenvalue of Matrix Product]\label{def: equality of eigenvalues}
For any two matrices $\bv A \in \C^{n\times m}, \bv B \in \C^{m\times n}$, the non-zero eigenvalues of $\bv{AB}$ are identical to those of $\bv B\bv A$.
\end{fact}

\begin{fact}[Girshgorin's circle theorem \cite{gershgorin1931uber}]
\label{thm:girshgorin}
Let $\bv A \in \C^{\nbyn}$ with entries $\bv A_{ij}$. For $i \in [n]$, let $\bv R_i$ be the sum of absolute values of non-diagonal entries in the $i$\textsuperscript{th} row. Let $D(\bv A_{ii}, \bv R_i)$ be the closed disc centered at $\bv A_{ii}$ with radius $\bv R_i$. Then every eigenvalue of $\bv A$ lies within one of the discs $D(\bv A_{ii}, \bv R_i)$.
\end{fact}

\begin{restatable}[Weyl's Inequality \cite{weyl1912asymptotic}]{fact}{eigenvalueperturb}
\label{thm:eigenvalue_perturbation_theorem}
For any two Hermitian matrices $\bv A, \bv B \in \mathbb{C}^{n\times n}$ with $\bv A - \bv B = \bv E$, 
\begin{align}
    \label{eq:eigenvalue_perturbation_theorem}
    \max_i |\lambda_i(\bv A) - \lambda_i(\bv B)| \leq \|\bv E\|_2.\notag
\end{align}
\end{restatable}
\noindent Weyl's inequality ensures that a small Hermitian perturbation of a Hermitian matrix will not significantly change its eigenvalues. The bound can be extended to the case when the perturbation is not Hermitian, with a loss of an $O(\log n)$ factor; to the best of our knowledge this loss is necessary:

\begin{fact}[Non-Hermitian perturbation bound \cite{bhatia2013matrix}]\label{fact:weyl_general}
Let $\bv A \in \C^{n\times n}$ be Hermitian and $\bv B \in \C^{n\times n}$ be any matrix whose eigenvalues are $\lambda_1(\bv B), \ldots, \lambda_n(\bv B)$ such that $Re(\lambda_1(\bv B))\geq \ldots \ge Re(\lambda_n(\bv B))$ (where $Re(\lambda_i(\bv B))$ denotes the real part of $\lambda_i(\bv B)$). Let $\bv A- \bv B = \bv E$. For some universal constant $C$,
\begin{align*}
    \max_i|\lambda_i(\bv A) - \lambda_i(\bv B)| \leq C \log n\|\bv E\|_2.
\end{align*}
\end{fact}

Beyond the above facts, we use several theorems to obtain eigenvalue concentration bounds. We first state a theorem from~\cite{tropp2008norms}, which bounds the spectral norm of a principal submatrix sampled uniformly at random from a bounded entry matrix. We build on this to prove the full eigenspectrum concentration result of Theorem \ref{thm:main_bound}.
\begin{restatable}[Random principal submatrix spectral norm bound \cite{rudelson2007sampling,tropp2008norms}]{theorem}{rps}
\label{thm: tropp2008}
Let $\bv A \in \C^{n \times n}$ be  Hermitian, decomposed into diagonal and off-diagonal parts: $\bv A = \bv D+\bv H$. Let $\bv S \in \R^{n \times n}$ be a diagonal sampling matrix with the $j^{th}$ diagonal entry set to $1$ independently with probability $s/n$ and  $0$ otherwise.
 Then, for some universal constant $C$,
\begin{align}
    \mathbb{E}_2\|\bv{S}\bv{AS}\|_2 \leq C\left[\log n \cdot \mathbb{E}_2\|\bv{S}\bv{HS}\|_{\infty}+\sqrt{\frac{s \log n}{n}} \cdot \mathbb{E}_2\|\bv{HS}\|_{1 \rightarrow 2} + \frac{s}{n} \cdot \|\bv H\|_2\right] + \mathbb{E}_2\|\bv{S}\bv{DS}\|_2.\notag
\end{align}
\end{restatable}
For Theorems \ref{thm:nnz_main_bound} and \ref{thm:l2_main_bound}, we need an extension of Theorem \ref{thm: tropp2008} to the setting where rows are sampled non-uniformly. We will use two bounds here. The first is a decoupling and recoupling result for matrix norms. One can prove this lemma following an analogous result in~\cite{tropp2008norms} for sampling rows/columns uniformly. The proof is almost identical so we omit it.
\begin{restatable}[Decoupling and recoupling]{lemma}{coupling}
\label{lemma:coupling} Let $\bv{H}$ be a Hermitian matrix with zero diagonal. Let $\delta_j$ be a sequence of independent random variables such that $\delta_j=\frac{1}{\sqrt{p_j}}$ with probability $p_j$ and $0$ otherwise. Let $\bv S$ be a square diagonal sampling matrix with $j^{th}$ diagonal entry set to $\delta_j$. Then:
$$\mathbb{E}_2 \|\bv{SHS} \|_2 \leq  2\mathbb{E}_2\|\bv{SH\hat{S}} \|_2 \hspace{1em}\text{and}\hspace{1em}\mathbb{E}_2 \|\bv{SH\hat{S}} \|_{\infty} \leq 4\mathbb{E}_2 \|\bv{SHS} \|_{\infty},$$
where $\bv{\hat{S}}$ is an independent diagonal sampling matrix drawn from the same distribution as $\bv{S}$.
\end{restatable}

\noindent The second  theorem bounds the spectral norm of a non-uniform random column sample of a matrix. We give a proof in Appendix \ref{app:gen_spectral_bounds}, again following a theorem in~\cite{tropp2007} for uniform sampling.
\begin{restatable}[Non-uniform column sampling -- spectral norm bound]{theorem}{rudelson}
\label{thm: rudelson}
Let $\bv{A}$ be an $m \times n$ matrix with rank $r$.
Let $\delta_j$ be a sequence of independent random variables such that $\delta_j=\frac{1}{\sqrt{p_j}}$ with probability $p_j$ and $0$ otherwise. Let $\bv S$ be a square diagonal sampling matrix with $j^{th}$ diagonal entry set to $\delta_j$. 
$$\E_2 \|\bv{AS} \|_2 \leq 5\sqrt{\log r} \cdot\E_2 \|\bv{AS} \|_{1 \rightarrow 2}+ \|\bv{A}\|_2$$
\end{restatable}

\noindent We use a standard Matrix Bernstein inequality to bound the spectral norm of random submatrices.

\begin{restatable}[Matrix Bernstein \cite{tropp2015introduction}]{theorem}{matrixbernstein}
\label{thm:matrix bernstein}
Consider a finite sequence $\{\bv S_k\}$ of random matrices in $\R^{d \times d}$. Assume that for all $k$,
$
    \mathbb{E}[\bv S_k] = \bv 0\quad\text{and}\quad\|\bv S_k\|_2 \leq L.
$
Let $\bv Z = \sum_k \bv S_k$ and let 
$\bv V_1,\bv V_2$ be semidefinite upper-bounds for the matrix valued variances $\bv{Var}_1(\bv Z)$ and $\bv{Var}_2(\bv Z)$:
\begin{align}
    \bv V_1 &\succeq \bv{Var}_1(\bv Z) \eqdef \mathbb{E}\left(\bv {ZZ}^T\right) = \sum_k \mathbb{E}\left(\bv S_k \bv S_k^T\right), \quad \text{and}\notag\\
    \bv V_2 &\succeq \bv{Var}_2(\bv Z) \eqdef \mathbb{E}\left(\bv {Z}^T\bv Z\right) = \sum_k \mathbb{E}\left(\bv S_k^T\bv S_k \right).\notag
\end{align}
Then, letting $v=\max(\|\bv V_1\|_2, \|\bv V_2\|_2)$, for any $t \geq 0$,
\begin{align}
    \Pr(\|\bv Z\|_2\geq t) &\leq 2d \cdot \exp\left(\frac{-t^2/2}{v+Lt/3}\right).\notag
\end{align}
\end{restatable}

\noindent For real valued random variables, we use the standard Bernstein inequality.

\begin{theorem}[Bernstein inequality \cite{bernstein1927extension}]
\label{thm:bernstein}
Let $\{z_j\}$ for $j \in [n]$ be independent random variables with zero mean such that $\lvert z_j \rvert \leq M$ for all $j$. Then for all positive $t$,
\begin{align*}
    \Pr\left(\left\lvert \sum_{j=1}^n z_j \right\rvert \geq t\right) \leq \exp\left(\frac{-t^2/2}{\sum_{i=1}^n \E[z_i^2] + Mt/3}\right).
\end{align*}
\end{theorem}


\section{Sublinear Time Eigenvalue Estimation using Uniform Sampling}\label{sec:eigenvalue estimation}

We now prove our main eigenvalue estimation result -- Theorem \ref{thm:main_bound}. We give the pseudocode for our principal submatrix based estimation procedure in Algorithm \ref{alg:eigenvalue estimate}. We will show that any positive or negative eigenvalue of $\bv{A}$ with magnitude $\ge \epsilon n$ will appear as an approximate eigenvalue in $\bv{A}_S$ with good probability. Thus, in step 5 of Algorithm \ref{alg:eigenvalue estimate}, the positive and negative eigenvvalues of $\bv{A}_S$ are used to estimate the outlying largest and smallest eigenvalues of $\bv{A}$. All other interior eigenvalues of $\bv{A}$ are estimated to be $0$, which will immediately give our $\pm \epsilon n$ approximation bound when the original eigenvalue has magnitude $\le \epsilon n$.

\begin{algorithm}
\caption{Eigenvalue estimator using uniform sampling}
\label{alg:eigenvalue estimate}
\begin{algorithmic}[1]
\STATE {\bfseries Input:} Symmetric $\bv A \in \mathbb{R}^{n\times n}$ with $\|\bv A \|_{\infty} \leq 1$, Accuracy $\epsilon \in (0,1)$, failure prob. $\delta \in (0,1)$.
\STATE Fix $s = \frac{c \log(1/(\epsilon \delta)) \cdot \log^3 n}{\epsilon^3 {\delta}}$ where $c$ is a sufficiently large constant.
\STATE Add each index $i \in [n]$ to the sample set $S$ independently with probability $\frac{s}{n}$. Let the principal submatrix of $\bv A$ corresponding $S$ be $\bv A_S$.
\STATE Compute the eigenvalues of $\bv A_S$: $\lambda_1(\bv{A}_S) \ge \ldots \ge \lambda_{|S|}(\bv{A}_S)$.
\STATE For all $i \in [|S|]$ with $\lambda_i(\bv{A}_S) \ge 0$, let $\tilde \lambda_i(\bv{A}) = \frac{n}{s} \cdot \lambda_i(\bv{A}_S)$. For all $i \in [|S|]$ with $\lambda_i(\bv{A}_S) < 0$, let $\tilde \lambda_{n-(|S|-i)}(\bv{A}) = \frac{n}{s} \cdot \lambda_i(\bv{A}_S)$. For all remaining $i \in [n]$, let $\tilde \lambda_i(\bv{A}) = 0$. 
\STATE {\bfseries Return:} Eigenvalue estimates $\tilde \lambda_1(\bv{A}) \ge \ldots \ge \tilde \lambda_n(\bv{A})$.
\end{algorithmic}
\end{algorithm}
 
\medskip

\noindent \textbf{Running time}. Observe that the expected number of indices chosen by Algorithm \ref{alg:eigenvalue estimate} is $s = \frac{c \log(1/(\epsilon \delta)) \cdot \log^3 n}{\epsilon^3 {\delta}}$. A standard concentration bound can be used to show that with high probability $(1-1/\poly(n))$, the number of sampled entries is $O(s)$. Thus, the algorithm reads a total of $O(s^2)$ entries of $\bv{A}$ and runs in $O(s^{\omega})$ time -- the time to compute a full eigendecomposition of $\bv{A}_S$. 

\subsection{Outer and Middle Eigenvalue Bounds}\label{sec:accuracy bounds}

Recall that we will split $\bv{A}$ into two symmetric matrices (Definition~\ref{def:split}): $\bv A_o=\bv V_o  \bv \Lambda_o \bv V_o^T$ which contains its large magnitude (outlying) eigendirections with eigenvalue magnitudes $\ge \epsilon \sqrt{\delta} n$ and $\bv A_m=\bv V_m \bv \Lambda_m \bv V_m^T$ which contains its small magnitude (middle) eigendirections.

We first show that the eigenvectors in $\bv{V}_o$ are \emph{incoherent}. I.e., that their (eigenvalue weighted) squared row norms are bounded. This ensures that the outlying eigenspace of $\bv{A}$ is well-approximated via uniform sampling.
\begin{restatable}[Incoherence of outlying eigenvectors]{lemma}{rownormbound}
\label{lemma:row_norm}
Let $\bv A \in \R^{n \times n}$ be symmetric with $\|\bv A\|_{\infty} \leq 1$. Let $\bv{V}_o$ be as in Definition \ref{def:split}. Let $\bv V_{o,i}$ denote the $i$\textsuperscript{th} row of $\bv V_o$. Then, 
\begin{align*}
\norm{\bv \Lambda_o^{1/2}\bv{V}_{o,i}}_2^2 \leq \frac{1}{\epsilon \sqrt{\delta}} \hspace{2em}\text{ and }\hspace{2em}\|\bv V_{o,i}\|^2_2 \leq \frac{1}{\epsilon^2\delta n}.
\end{align*}
\end{restatable}


\begin{proof}
Observe that $\bv A \bv V_o =\bv V_o \bv \Lambda_o $. Let  $[\bv A \bv V_o]_i$ denote the $i$\textsuperscript{th} row of the $\bv A \bv V_o$. Then we have
\begin{equation}\label{Eq: row_norm1}
    \|[\bv A \bv V_o]_i\|_2^2 = \|[\bv V_o \bv \Lambda_o]_i\|_2^2 = \sum_{j=1}^r \lambda_j^2 \cdot \bv V_{o,i, j}^2,
\end{equation}
where $r=\rank(\bv A_o)$, $\bv V_{o,i,j}$ is the $(i,j)$\textsuperscript{th} element of $\bv V_o$  and $\lambda_j=\bv \Lambda_o(j,j)$. $\|\bv A\|_\infty \leq 1$ by assumption and since $\bv{V}_o$ has orthonormal columns, its spectral norm is bounded by $1$, thus we have
\begin{equation*}\label{Eq: row_norm2}
    \|[\bv A \bv V_o]_i\|_2^2 =\|[\bv A]_i \bv V_o\|_2^2\leq \|[\bv A]_i\|_2^2 \cdot \|\bv V_o\|^2_{2} \leq n.
\end{equation*}
Therefore, by \eqref{Eq: row_norm1}, we have:
\begin{equation}\label{Eq: eig_bound }
    \sum_{j=1}^r \lambda_j^2 \cdot  \bv V_{o,i, j}^2 \leq n.
\end{equation}
Since by definition of $\bv{\Lambda}_o$, $\lvert \lambda_j \rvert \geq \epsilon \sqrt{\delta} n$ for all $j$, we finally have 
$$\norm{\bv \Lambda_o^{1/2}\bv{V}_{o,i}}_2^2 =\sum_{j=1}^r \lambda_j\cdot \bv{V}_{o,i,j}^2 \leq \frac{n}{\epsilon\sqrt{\delta} n} = \frac{1}{\epsilon \sqrt{\delta}}$$ and
\begin{align*}
    \|\bv V_{o,i}\|_2^2 = \sum_{j=1}^r \bv V_{o,i, j}^2 
    &\leq \frac{n}{\epsilon^2\delta n^2} = \frac{1}{\epsilon^2\delta n}.
\end{align*}
\end{proof}

Let $\bv{\bar S} \in \R^{n \times |S|}$ be the scaled sampling matrix satisfying $\bv{\bar S}^T \bv{A} \bv{\bar S} = \frac{n}{s} \cdot \bv{A}_{S}$.
We next apply Lemma \ref{lemma:row_norm} in conjunction with a matrix Bernstein bound to show that $\bv \Lambda_o^{1/2}\bv V_o^T \bar{\bv S} \bar{\bv S}^T \bv V_o \bv \Lambda_o^{1/2}$ concentrates around its expectation, $\bv \Lambda_o$. Since by Fact \ref{def: equality of eigenvalues}, this matrix has identical eigenvalues to $\frac{n}{s} \cdot \bv{A}_{o,S} = \bar{\bv S}^T\bv V_o \bv \Lambda_o \bv V_o^T  \bar{\bv S}$, this allows us to argue that the eigenvalues of $\frac{n}{s} \cdot \bv{A}_{o,S}$ approximate those of $\bv \Lambda_o$.

\begin{lemma}[Concentration of outlying eigenvalues]\label{lemma: orthonormality}
Let $S \subseteq [n]$ be sampled as in Algorithm \ref{alg:eigenvalue estimate} for $s \geq \frac{c \log(1/(\epsilon \delta))}{\epsilon^3\sqrt{\delta}}$ where $c$ is a sufficiently large constant. Let $\bv{\bar S} \in \R^{n \times |S|}$ be the scaled sampling matrix satisfying $\bv{\bar S}^T \bv{A} \bv{\bar S} = \frac{n}{s} \cdot \bv{A}_{S}$. Letting $\bv \Lambda_o, \bv{V}_o$ be as in Definition \ref{def:split}, with probability at least $1-\delta$, $$\|\bv \Lambda_o^{1/2}\bv V_o^T \bar{\bv S} \bar{\bv S}^T \bv V_o\bv \Lambda_o^{1/2} - \bv \Lambda_o \|_2 \leq \epsilon n.$$ 
\end{lemma}
\begin{proof}
Define $\bv{E} = \bv \Lambda_o^{1/2}\bv V_o^T \bar{\bv S} \bar{\bv S}^T \bv V_o \bv \Lambda_o^{1/2}  - \bv \Lambda_o$.
For all $i \in [n]$, let $\bv{V}_{o,i}$ be the $i^{th}$ row of $\bv{V}_o$ and define the matrix valued random variable
\begin{align}
    \label{eq:rv}
    \bv Y_i = 
    \begin{cases}
    \frac{n}{s}\bv \Lambda_o^{1/2} \bv{V}_{o,i}\bv{V}_{o,i}^T\bv \Lambda_o^{1/2} , & \text{with probability } s/n\\
    0 & \text{otherwise.}
    \end{cases}
\end{align}

Define $\bv Q_i = \bv Y_i - \mathbb{E}\left[\bv Y_i\right]$. Observe that $\bv Q_1, \ldots, \bv Q_n$ are independent random variables and that $ \sum_{i=1}^n \bv Q_i=\bv \Lambda_o^{1/2} \bv V_o^T \bar{\bv S} \bar{\bv S}^T \bv V_o \bv \Lambda_o^{1/2}  - \bv \Lambda_o = \bv{E}$.  Further, observe that $\norm{\bv{Q}_i}_2 \le \max\left (1, \frac{n}{s} -1 \right)  \cdot \norm{ \bv \Lambda_o^{1/2}\bv{V}_{o,i} \bv{V}_{o,i}^T \bv \Lambda_o^{1/2}}_2 \leq \max\left (1, \frac{n}{s} -1 \right)  \cdot  \norm{\bv \Lambda_o^{1/2}\bv{V}_{o,i}}_2^2$. Now, $\norm{\bv \Lambda_o^{1/2}\bv{V}_{o,i}}_2^2 \leq \frac{1}{\epsilon \sqrt{\delta}}$ by Lemma~\ref{lemma:row_norm}. Thus, $\norm{\bv{Q}_i}_2 \le \frac{n}{\epsilon \sqrt{\delta} s}$.
The variance $\bv{Var}(\bv E) \eqdef\E(\bv{EE}^T)= \E(\bv{E}^T\bv{E})  =\sum_{i=1}^n \mathbb{E}[\bv Q_i^2]$ can be bounded as:
\begin{align}
\sum_{i=1}^n \mathbb{E}[\bv Q_i^2]
&= \sum_{i=1}^n \left[ \frac{s}{n} \cdot \left(\frac{n}{s}-1\right)^2  + \left (1-\frac{s}{n}\right ) \right] \cdot ( \bv \Lambda_o^{1/2}\bv{V}_{o,i} \bv{V}_{o,i}^T \bv \Lambda_o \bv{V}_{o,i} \bv{V}_{o,i}^T \bv \Lambda_o^{1/2})\nonumber \\
&\preceq \sum_{i=1}^n \frac{n}{s} \cdot \norm{\bv \Lambda_o^{1/2} \bv{V}_{o,i}}_2^2 \cdot  (\bv \Lambda_o^{1/2}\bv{V}_{o,i}  \bv{V}_{o,i}^T\bv \Lambda_o^{1/2}).\label{var_ineq}
\end{align}
Again by Lemma \ref{lemma:row_norm},  $\norm{\bv \Lambda_o^{1/2} \bv{V}_{o,i}}_2^2 \le \frac{1}{\epsilon \sqrt{\delta}}$. Plugging back into \eqref{var_ineq} we can bound,
\begin{align*}
\sum_{i=1}^n \mathbb{E}[\bv Q_i^2] &\preceq \sum_{i=1}^n \frac{n}{s} \cdot \frac{1}{\epsilon \sqrt{\delta}} \cdot  (\bv \Lambda_o^{1/2}\bv{V}_{o,i} \bv{V}_{o,i}^T \bv \Lambda_o^{1/2}) = \frac{n}{s \epsilon \sqrt{\delta}} \bv{\Lambda}_o \preceq \frac{n^2}{s\epsilon \sqrt{\delta}} \cdot \bv{I}.
\end{align*}
Since $\bv{Q}_i^2$ is PSD, this establishes that $\norm{\bv{Var}(\bv E)}_2 \le  \frac{n^2}{s \epsilon \sqrt{\delta}} $.
We then apply Theorem~\ref{thm:matrix bernstein} (the matrix Bernstein inequality) with $L = \frac{n}{s\epsilon \sqrt{\delta}}$, $v= \frac{n^2}{s\epsilon \sqrt{\delta}}$, and $d \le \frac{1}{\epsilon^2 \delta}$ since there are at most $\frac{\norm{\bv A}_F^2}{\delta \epsilon^2 n^2} \le \frac{1}{\epsilon^2 \delta}$ outlying eigenvalues with magnitude $\ge \sqrt{\delta} \epsilon n$ in $\bs{\Lambda}_o$. This gives:
\begin{align}
    \Pr\left(\left\|\bv E\right\|_2 \geq \epsilon n \right) &\leq \frac{2}{\epsilon^2 \delta} \cdot \exp\left(\frac{-\epsilon^2n^2/2}{v+L\epsilon n/3}\right)\notag\\
    &\leq \frac{2}{\epsilon^2 \delta} \cdot \exp \left(\frac{-\epsilon^2n^2 /2}{\frac{n^2}{s\epsilon\sqrt{\delta}} + \frac{\epsilon n^2}{3s\epsilon \sqrt{\delta}}}\right)\notag\\
    &\leq \frac{2}{\epsilon^2 \delta} \cdot \exp \left( \frac{-s\epsilon^3\sqrt{\delta}}{4}\right).\notag
\end{align}
Thus, if we set $s \geq \frac{c\log(1/(\epsilon \delta))}{\epsilon^3 \sqrt{\delta}}$ for large enough $c$, then the probability is bounded above by $\delta$, completing the proof.
\end{proof}

We cannot prove an analogous leverage score bound to Lemma \ref{lemma:row_norm} for the interior eigenvectors of $\bv A$ appearing in $\bv{V}_m$. Thus we cannot apply a matrix Bernstein bound as in Lemma \ref{lemma: orthonormality}. However, we can use Theorem \ref{thm: tropp2008} to show that the spectral norm of the random principal submatrix $\bv{A}_{m,S}$ is not too large, and thus that the eigenvalues of $\bv{A}_S = \bv{A}_{o,S}+\bv{A}_{m,S}$ are close to those of $\bv{A}_{o,S}$.

\begin{restatable}[Spectral norm bound -- sampled middle eigenvalues]{lemma}{mideigbound}
\label{middle}
Let $\bv A \in \mathbb{R}^{n\times n}$ be symmetric with $\|\bv A\|_\infty \leq 1$. Let $\bv A_m$ be as in Definition \ref{def:split}. Let $S$ be sampled as in Algorithm \ref{alg:eigenvalue estimate}. If $s\geq \frac{c\log n}{\epsilon^2\delta}$ for some sufficiently large constant $c$, then with probability at least $1-\delta$, $\norm{\bv A_{m,S}}_2 \leq \epsilon s$. 
\end{restatable}


\begin{proof}
Let $\bv A_{m}=\bv D_m+\bv H_m$ where $\bv D_m$ is the matrix of diagonal elements and $\bv H_m$ the matrix of off-diagonal elements. Let $\bv S \in \R^{n \times |S|}$ be the binary sampling matrix with $\bv A_{m,S}=\bv S^T\bv A_m\bv S$. From Theorem~\ref{thm: tropp2008}, we have for some constant $C$,
\begin{equation}\label{Eq: tropp_bound}
    \mathbb{E}_2[\|\bv A_{m,S}\|_2] \leq C\bigg[\log n \cdot \mathbb{E}_2[\|\bv S^T\bv H_m\bv S\|_{\infty}]+\sqrt{ \frac{s \log n}{n}} \mathbb{E}_2[\|\bv H_m \bv S\|_{1\rightarrow 2}] +\frac{s}{n}\|\bv H_m\|_2 \bigg]+\mathbb{E}_2[\|\bv S^T\bv D_m \bv S \|].
\end{equation}
Considering the various terms in \eqref{Eq: tropp_bound}, we have $\|\bv S^T\bv H_m\bv S\|_{\infty} \leq \|\bv A_m\|_{\infty}$ and $\norm{\bv{S}^T \bv{D}_m \bv{S}}_2 = \norm{\bv{S}^T \bv{D}_m \bv{S}}_\infty \le \norm{\bv{A}_m}_\infty$. We also have $$\|\bv H_m \|_2 \leq \|\bv A_m \|_2+\| \bv D_m \|_2 \leq \| \bv A_m\|_2 +\|\bv A_m \|_{\infty} \le \epsilon \delta^{1/2} n + \|\bv A_m \|_{\infty}$$ and $$\|\bv H_m \bv S \|_{1\rightarrow 2} \leq \|\bv A_m \bv S \|_{1\rightarrow 2} \leq \|\bv A_m \|_{1\rightarrow 2} \le \sqrt{n}.$$ 
The final bound follows since $\bv{A}_m = \bv{V}_m \bv{V}_m^T\bv{A}  $, where $\bv{V}_m \bv{V}_m^T$ is an orthogonal projection matrix. Thus, $\|\bv A_m \|_{1\rightarrow 2} \le \|\bv A \|_{1\rightarrow 2} \le \sqrt{n}$ by our assumption that $\norm{\bv{A}}_\infty \le 1$.
Plugging all these bounds into \eqref{Eq: tropp_bound} we have, for some constant $C$,
\begin{align}\label{before_infty}
\mathbb{E}_2[\|\bv A_{m,S}\|_2] \leq C\bigg[\log n \cdot \|\bv A_{m}
\|_{\infty}+\sqrt{\log n\cdot s} +s \cdot \epsilon \delta^{1/2} \bigg].
\end{align}
It remains to bound $\norm{\bv{A}_m}_\infty$. We have $\bv A=\bv A_m+\bv A_o$ and thus by triangle inequality, 
\begin{align}\label{eq:inf_tri}
\norm{\bv{A}_m}_\infty \le \norm{\bv{A}}_\infty + \norm{\bv{A}_o}_\infty = 1 + \norm{\bv{A}_o}_\infty.
\end{align}
Writing $\bv A_o=\bv V_o \bv \Lambda_o \bv V_o^T$ (see Definition \ref{def:split}), and letting  $\bv V_{o,i}$ denote the $i$\textsuperscript{th} row of $\bv V_o$, the $(i,j)$\textsuperscript{th} element of $\bv A_o$ has magnitude$$|\bv A_{o,i,j}| = |\bv V_{o,i}\bv \Lambda_o \bv V^T_{o,j}| \leq \|\bv V_{o,i} \|_2 \cdot \| \bv \Lambda_o \bv V^T_{o,j} \|_2 ,$$by Cauchy-Schwarz. From Lemma \ref{lemma:row_norm}, we have $\|\bv V_{o,i}\|_2\leq \frac{1}{\epsilon \delta^{1/2} \sqrt{n}}$. Also, from \eqref{Eq: row_norm1}, $\|\bv \Lambda_o \bv V^T_{o,j} \|_2 = \|[\bv A\bv V_o]_j \|_2 \leq \sqrt{n}   $. Overall, for all $i,j$ we have $\bv A_{o,i,j} \leq \frac{1}{\epsilon \delta^{1/2} \sqrt{n}} \cdot \sqrt{n} =\frac{1}{\epsilon \delta^{1/2}}$, giving $\|\bv A_o \|_{\infty} \leq \frac{1}{\epsilon \delta^{1/2}}$. Plugging back into \eqref{eq:inf_tri} and in turn \eqref{before_infty}, we have for some constant $C$,
$$\mathbb{E}_2[\|\bv A_{m,S}\|_2] \leq C\bigg[\frac{\log n}{\epsilon\delta^{1/2}}+\sqrt{s\log n}+ s\epsilon \delta^{1/2} \bigg].$$
Setting $s \geq \frac{c\log n}{\epsilon^2\delta}$ for sufficiently large $c$, all terms in the right hand side of the above equation are bounded by $\epsilon \sqrt{\delta} s$ and so $$\mathbb{E}_2[\|\bv A_{m,S}\|_2] \leq 3\epsilon\sqrt{\delta} s$$ 
Thus, by Markov's inequality, with probability at least $1-\delta$, we have $\|\bv A_{m,S}\|_2 \leq 3\epsilon s$. We can adjust $\epsilon$ by a constant to obtain the required bound.
\end{proof}

\subsection{Main Accuracy Bounds}\label{sec:main accuracy bounds, uniform}

We now restate our main result, and give its proof via Lemmas \ref{lemma: orthonormality} and \ref{middle}.
\eigvalApprox*

\begin{proof}
Let $\bv{S} \in \R^{n \times |S|}$ be the binary sampling matrix with a single one in each column such that $\bv{S}^T \bv{A} \bv{S} = \bv{A}_S$. Let $\bar{\bv S} = \sqrt{n/s} \cdot \bv{S}$
Following Definition \ref{def:split}, we write $\bv A = \bv A_o + \bv A_m$.
By Fact \ref{def: equality of eigenvalues} we have that the nonzero eigenvalues of $\frac{n}{s} \cdot \bv{ A}_{o,S} = \bar{\bv S}^T \bv V_o \bv \Lambda_o \bv V_o^T \bar{\bv S}$ are identical to those of $\bv{\Lambda}_o^{1/2} \bv V_o^T \bar{\bv S} \bar{\bv S}^T \bv V_o \bv{\Lambda}_o^{1/2}$ where $\bv{\Lambda}_o^{1/2}$ is the square root matrix of $\bv{\Lambda}_o$ such that $\bv{\Lambda}_o^{1/2}\bv{\Lambda}_o^{1/2}=\bv{\Lambda}_o$. 

Note that $\bv{\Lambda}_o$ is Hermitian. However $\bv{\Lambda}_o^{1/2}$ may be complex, and hence $\bv{\Lambda}_o^{1/2} \bv V_o^T \bar{\bv S} \bar{\bv S}^T \bv V_o \bv{\Lambda}_o^{1/2}$ is \emph{not necessarily Hermitian}, although it does have real eigenvalues.
Thus, we can apply the perturbation bound of Fact \ref{fact:weyl_general} to $\bv{\Lambda}_o$ and $\bv{\Lambda}_o^{1/2} \bv V_o^T \bar{\bv S} \bar{\bv S}^T \bv V_o \bv{\Lambda}_o^{1/2}$ to claim for all $i \in [n]$, and some constant $C$, $$\lvert \lambda_i(\bv{\Lambda}_o^{1/2} \bv V_o^T \bar{\bv S} \bar{\bv S}^T \bv V_o \bv{\Lambda}_o^{1/2})-\lambda_i(\bv{\Lambda}_o) \rvert \leq C \log n \|\bv{\Lambda}_o^{1/2} \bv V_o^T \bar{\bv S} \bar{\bv S}^T \bv V_o \bv{\Lambda}_o^{1/2}-\bv{\Lambda}_o \|_2.$$
 By Lemma \ref{lemma: orthonormality} applied with error $\frac{\epsilon}{2C\log n}$, with probability at least $1-\delta$, for any $s \geq \frac{c \log(1/(\epsilon \delta)) \cdot \log^3 n}{\epsilon^3 \sqrt{\delta}}$ (for a large enough constant $c$) we have $\|\bv{\Lambda}_o^{1/2} \bv V_o^T \bar{\bv S} \bar{\bv S}^T \bv V_o \bv{\Lambda}_o^{1/2}-\bv{\Lambda}_o \|_2 \leq \frac{\epsilon n}{2C\log n}$. Thus, for all $i$,
\begin{align}
   \left| \lambda_i(\bv{\Lambda}_o^{1/2} \bv V_o^T \bar{\bv S} \bar{\bv S}^T \bv V_o \bv{\Lambda}_o^{1/2} ) - \lambda_i(\bv \Lambda_o)\right| &<\frac{\epsilon n}{2}.
  \label{eq:apply_perturbation}
\end{align}
We note that the conceptual part of the proof is essentially complete: the nonzero eigenvalues of $\frac{n}{s} \cdot \bv{A}_{o,S}$ are identical to those of $\bv{\Lambda}_o^{1/2} \bv V_o^T \bar{\bv S} \bar{\bv S}^T \bv V_o \bv{\Lambda}_o^{1/2}$, which we have shown well approximate those of $\bv{\Lambda}_o$ and in turn $\bv{A}_o$. i.e., the non-zero eigenvalues of $\frac{n}{s} \cdot \bv{A}_{o,S}$ approximate all outlying eigenvalues of $\bv{A}$. It remains to  carefully argue how these approximations should be `lined up' given the presence of zero eigenvalues in the spectrum of these matrices.  We also must account for the impact of the interior eigenvalues in $\bv{A}_{m,S}$, which is limited by the spectral norm bound of Lemma \ref{middle}.

\medskip

\noindent\textbf{Eigenvalue alignment and effect of interior eigenvalues.} 
First recall that $\bv A_S=\bv A_{o,S}+\bv A_{m,S}$. By Lemma \ref{middle} applied with error $\epsilon/2$, we have $\|\bv A_{m,S} \|_2 \leq \epsilon/2 \cdot s $ with probability at least $1-\delta$ when $s \geq \frac{c\log n}{\epsilon^2 \delta}$. By Weyl's inequality (Fact  \ref{thm:eigenvalue_perturbation_theorem}), for all $i \in [|S|]$ we thus have
\begin{align}\label{eq:eig_middle}
     \left \lvert \frac{n}{s}\lambda_i(\bv A_S) - \frac{n}{s}\lambda_i(\bv A_{o,S}) \right \rvert &\leq \frac{n}{s} \cdot \frac{\epsilon s}{2} = \frac{\epsilon n}{2}.
\end{align}
Consider $i \in [|S|]$ with $\lambda_i(\bv A_{o,S}) > 0$. Since the nonzero eigenvalues of $\frac{n}{s} \cdot \bv{A}_{o,S}$ are identical to those of $\bv{\Lambda}_o^{1/2} \bv V_o^T \bar{\bv S} \bar{\bv S}^T \bv V_o \bv{\Lambda}_o^{1/2}$,
$\frac{n}{s} \cdot \lambda_i(\bv A_{o,S}) = \lambda_i(\bv{\Lambda}_o^{1/2} \bv V_o^T \bar{\bv S} \bar{\bv S}^T \bv V_o \bv{\Lambda}_o^{1/2} )$, and so by 
\eqref{eq:apply_perturbation},
\begin{align}\label{eq:eig_top1}
\left| \frac{n}{s} \cdot  \lambda_i(\bv A_{o,S}) - \lambda_i(\bv \Lambda_o)\right| &< \frac{\epsilon n}{2}.
\end{align}
Analogously, consider $i \in [|S|]$ such that $\lambda_i(\bv A_{o,S}) < 0$. We have $\frac{n}{s}\cdot\lambda_i(\bv A_{o,S}) = \lambda_{r-(|S|-i)}(\bv{\Lambda}_o^{1/2} \bv V_o^T \bar{\bv S} \bar{\bv S}^T \bv V_o \bv{\Lambda}_o^{1/2})$, where $r$ is the dimension of $\bv \Lambda_o$ -- i.e., the number of outlying eigenvalues in $\bv{A}$. Again by \eqref{eq:apply_perturbation} we have
\begin{align}\label{eq:eig_bot1}
\left| \frac{n}{s} \cdot  \lambda_i(\bv A_{o,S}) - \lambda_{r-(|S|-i)}(\bv \Lambda_o)\right| &< \frac{\epsilon n}{2}.
\end{align}
Now the nonzero eigenvalues of $\bv{A}_o$ are identical to those of $\bv{\Lambda}_o$. Consider $i \in [|S|]$ such that  $\lambda_i(\bv A_{S}) \ge \epsilon s$. In this case, by \eqref{eq:eig_middle}, \eqref{eq:eig_top1}, and the triangle inequality, we have $\lambda_i(\bv{\Lambda}_o) > 0$ and thus we have $\lambda_i(\bv{\Lambda}_o) = \lambda_i(\bv{A}_o)$. In turn, again applying  \eqref{eq:eig_middle}, \eqref{eq:eig_top1}, and the triangle inequality, we have
$$\left |\frac{n}{s} \lambda_i(\bv A_{S})  -\lambda_i(\bv{A}_o)\right | \le \left |\frac{n}{s} \lambda_i(\bv A_{o,S})  -\lambda_i(\bv{A}_o)\right | + \left |\frac{n}{s} \lambda_i(\bv A_{S})  -\lambda_i(\bv{A}_{o,S})\right | \le  \epsilon n.$$
 Analogously, for $i \in [|S|]$ such that  $\lambda_i(\bv A_{S}) \le -\epsilon s$, we have by \eqref{eq:eig_middle} and \eqref{eq:eig_bot1} that $\lambda_{r-(|S|-i)}(\bv{\Lambda}_o) < 0$. Thus $\lambda_{r-(|S|-i)}(\bv{\Lambda}_o) = \lambda_{n-(r-i)}(\bv{A}_o)$. Again by \eqref{eq:eig_middle}, \eqref{eq:eig_bot1}, and triangle inequality this gives 
 $$\left| \frac{n}{s} \cdot  \lambda_i(\bv A_{S}) - \lambda_{n-(|S|-i)}(\bv A_o)\right| \le \epsilon n.$$

Now, consider all $i \in [n]$ such that $\lambda_i(\bv{A}_o)$ is not well approximated by one of the outlying eigenvalues of $\bv{A}_{S}$ as argued above. By \eqref{eq:eig_middle}, \eqref{eq:eig_top1}, and \eqref{eq:eig_bot1}, all such eigenvalues must have  $|\lambda_i(\bv{A}_o)| \le 2\epsilon n$. Thus, if we approximate them in any way either by the remaining eigenvalues of $\bv{A}_{S}$ with magnitude $\le \epsilon s$, or else by $0$, we will approximate all to error at most $3 \epsilon n$. Thus, if (as in Algorithm \ref{alg:eigenvalue estimate}) for  $i \in [|S|]$ with $\lambda_i(\bv{A}_{S}) \ge 0$, we let $\tilde \lambda_i(\bv{A}) = \frac{n}{s} \cdot \lambda_i(\bv{A}_{S})$ and for $i \in [|S|]$ with $\lambda_i(\bv{A}_S) < 0$, let $\tilde \lambda_{n-(|S|-i)}(\bv{A}) = \frac{n}{s} \cdot \lambda_i(\bv{A}_{S})$, and let $\tilde \lambda_i(\bv{A}) = 0$ for all other $i$, we will have for all $i$,
\begin{align*}
\left |\tilde \lambda_i(\bv{A}) - \lambda_i(\bv{A}_o) \right | \le 3\epsilon n.
\end{align*}
Finally by definition,  for all $i$, $|\lambda_i(\bv{A})-\lambda_i(\bv{A}_o)| \le \epsilon \sqrt{\delta} n \le \epsilon n$ and thus, via triangle inequality,
$
\left |\tilde \lambda_i(\bv{A}) - \lambda_i(\bv{A}) \right | \le 4\epsilon n.
$
This gives our final error bound after adjusting constants on $\epsilon$. 

Recall that we require $s \geq  \frac{c \log(1/(\epsilon \delta)) \cdot \log^3 n}{\epsilon^3 \sqrt{\delta}}$ for the outer eigenvalue bound of \eqref{eq:apply_perturbation} to hold with probability $1-\delta$. We  require $s \geq \frac{c\log n}{\epsilon^2\delta}$ for $\|\bv A_{m,S}\|_2 \leq \epsilon/2\cdot s$ to hold with probability $1-\delta$ by Lemma \ref{middle}. Thus, for both conditions to hold simultaneously with probability $1-2\delta$ by a union bound, if suffices to set  $s = \frac{c \log(1/(\epsilon \delta)) \cdot \log^3 n}{\epsilon^3 {\delta}} \ge \max\left(\frac{c \log(1/(\epsilon \delta)) \cdot \log^3 n}{\epsilon^3 \sqrt{\delta}}, \frac{c\log n}{\epsilon^2\delta}\right)$, where we use that $\log(1/(\epsilon \delta) \le O(\log n)$, as otherwise our algorithm can take $\bv{A}_S$ to be the full matrix $\bv{A}$. Adjusting $\delta$ to $\delta/2$ completes the theorem.
\end{proof}

\noindent \textbf{Remark:} The proof of Lemma~\ref{lemma: orthonormality} and consequently, Theorem~\ref{thm:main_bound} can be modified to give better bounds for the case when the eigenvalues of $\bv A_o$ lie in a bounded range -- between $\epsilon^a\sqrt{\delta} n$ and $\epsilon^b n$ where $0 \leq b \leq a \leq 1$. See Theorem~\ref{cor: refined_bound1} in Appendix \ref{app:refined} for details. For example, if all the top eigenvalues are equal, one can show that $s = \tilde O\left (\frac{\log^2 n}{\epsilon^2}\right )$ suffices to give $\pm \epsilon n$ error, nearly matching the lower bound of \cite{BakshiChepurkoJayaram:2020}. This seems to indicate that improving Theorem~\ref{thm:main_bound} in general requires tackling the case when the outlying eigenvalues in $\bs{\Lambda}_o$ have a wide range. 


\section{Improved Bounds via Sparsity-Based Sampling}\label{sec:sparsity}

We now prove  the  $\pm \epsilon \sqrt{\nnz(\bv A)}$ approximation bound of Theorem \ref{thm:nnz_main_bound}, assuming the ability to sample each row with probability proportional to $\frac{\nnz(\bv{A}_i)}{\nnz(\bv{A})}$.
Pseudocode for our algorithm is given in Algorithm~\ref{alg:nnz eigenvalue estimate}. Unlike in the uniform sampling case (Algorithm \ref{alg:eigenvalue estimate}), we cannot simply sample a principal submatrix of $\bv A$ and compute its eigenvalues. We must carefully zero out entries lying at the intersection of sparse rows and columns to ensure accuracy of our estimates. A similar approach is taken for the  norm-based sampling result of Theorem \ref{thm:l2_main_bound}. We defer that proof to Appendix \ref{sec:l2}.

\subsection{Preliminary Lemmas}

Our first step is to argue that zeroing out entries in sparse rows/columns in step 5 of Algorithm \ref{alg:nnz eigenvalue estimate} does not introduce significant error. We define $\bv{A}' \in \R^{\nbyn}$ to be the extension of $\bv{A}'$ to the original matrix -- i.e., $\bv{A}'_{ij} = 0$ whenever $i = j$ or $\nnz(\bv A_i) \nnz(\bv A_j) < \frac{\epsilon^2 \nnz({\bv A})}{c_2\log^2 n}$. Otherwise $\bv{A}'_{ij} = \bv{A}_{ij}$.
We argue via a strengthening of Girshgorin's theorem that $|\lambda_i(\bv A) - \lambda_i(\bv A')| \leq \epsilon \sqrt{\nnz(\bv A)}$ for all $i$. 

After this step is complete, our proof follows the same general outline as that of Theorem \ref{thm:main_bound} in Section 3. We split $\bv{A}' = \bv{A}'_o + \bv{A}'_m$, arguing that (1) after sampling $\norm{\bv{A}_{m,S}'}_2 \le \epsilon \sqrt{\nnz(\bv A)}$ and (2) that the eigenvalues of $\bv{A}'_{o,S}$ are $\pm \epsilon \sqrt{\nnz(\bv{A})}$ approximations to those of $\bv{A}'_o$. In both cases, we critically use that the rescaling factors introduced in line 4 of Algorithm \ref{alg:nnz eigenvalue estimate} do not introduce too much variance, due to the zeroing out of entries in $\bv{A}'$.

\begin{flushleft}
\begin{minipage}[!ht]{\linewidth}
\begin{algorithm}[H]
\caption{Eigenvalue estimator using sparsity-based sampling}
\label{alg:nnz eigenvalue estimate}
\begin{algorithmic}[1]
\STATE {\bfseries Input:} Symmetric $\bv A \in \mathbb{R}^{n\times n}$ with $\|\bv A \|_{\infty} \leq 1$, Accuracy $\epsilon \in (0,1)$, failure prob. $\delta \in (0,1)$.  $\nnz(\bv{A}_i)$ for all $i \in [n]$ and $\nnz(\bv{A})$.
\STATE Fix $s = \frac{c_1\log^8 n}{\epsilon^8\delta^4}$ where $c_1$ is a sufficiently large constant.
\STATE Add each $i \in [n]$ to sample set $S$ independently, with probability $p_i=\min \left (1,\frac{s\nnz(\bv{A}_i)}{\nnz(\bv{A})}\right )$. Let the principal submatrix of $\bv A$ corresponding to $S$ be $\bv A_S$. 
\STATE Let $\bv{A}_S = \bv{D} \bv{A}_S \bv{D}$ where $\bv{D} \in \R^{|S| \times |S|}$ is diagonal with $\bv{D}_{i,i} = \frac{1}{\sqrt{p_j}}$ if the $i^{th}$ element of $S$ is $j$. 
\STATE Construct $\bv{A}'_S \in \R^{|S| \times |S|}$ from $\bv{A}_S$  as follows:
\begin{align*}
    \bv [\bv{A}'_S]_{i,j} &=   
    \begin{cases}
    0 & \text{if $i=j$ or }\nnz(\bv A_i) \nnz(\bv A_j) < \frac{\epsilon^2 \nnz({\bv A})}{c_2\log^2 n}\text{ for sufficient large constant $c_2$}\\
     [\bv A_S]_{i,j} & \text{otherwise}.
    \end{cases}
\end{align*}
\STATE Compute the eigenvalues of $\bv A'_S$: $\lambda_1(\bv{A}'_S) \ge \ldots \ge \lambda_{|S|}(\bv{A}'_S)$.
\STATE For all $i \in [|S|]$ with $\lambda_i(\bv{A}'_S) \ge 0$, let $\tilde \lambda_i(\bv{A}) =  \lambda_i(\bv{A}'_S)$. For all $i \in [|S|]$ with $\lambda_i(\bv{A}'_S) < 0$, let $\tilde \lambda_{n-(|S|-i)}(\bv{A}) =  \lambda_i(\bv{A}'_S)$. For all remaining $i \in [n]$, let $\tilde \lambda_i(\bv{A}) = 0$. 
\STATE {\bfseries Return:} Eigenvalue estimates $\tilde \lambda_1(\bv{A}) \ge \ldots \ge \tilde \lambda_n(\bv{A})$.
\end{algorithmic}
\end{algorithm}
\vspace{-0.5em}
\noindent \textbf{Remark:} Throughout, we will assume that $\bv{A}$ does not have any rows/columns that are all $0$, as such rows will never be sampled and will have no effect on the output of Algorithm \ref{alg:nnz eigenvalue estimate}. Additionally, we will assume that $\nnz(\bv{A}) \ge \frac{c_1\log^{8} n}{\epsilon^8\delta^4}$, as otherwise, $\bv{A}$ has at most $s = \frac{c_1\log^{8} n}{\epsilon^8\delta^4}$ non-zero rows.
Thus, rather than running Algorithm \ref{alg:nnz eigenvalue estimate}, we can directly compute the eigenvalues of $\bv A$.
\end{minipage}
\end{flushleft}
\vspace{0.5em}

\begin{lemma}\label{lem:nnz-zeroed}
    Let $\bv A \in \R^{n \times n}$ be symmetric with $\|\bv A\|_\infty \leq 1$ and $\nnz(\bv{A}) \ge 2/\epsilon^2$.  Let $\bv{A}' \in \R^{n \times n}$ have $\bv A'_{ij} = 0$ if $i = j$ or $\nnz(\bv A_i) \cdot \nnz(\bv A_j) < \frac{\epsilon^2 \nnz({\bv A})}{c_2\log^2 n}$ for a sufficiently large constant $c_2$ and $\bv A'_{ij} =  \bv A_{ij}$ otherwise. Then, for all $i \in [n]$, $$|\lambda_i(\bv A) - \lambda_i(\bv A')| \le \epsilon \sqrt{\nnz(\bv{A})}.$$
\end{lemma}

\begin{proof}
We consider the matrix $\bv{A}''$, which is defined identically to $\bv{A}'$ except we only set $\bv{A}''_{ij} = 0$ if  $\nnz(\bv A_i) \cdot \nnz(\bv A_j) < \frac{\epsilon^2 \nnz({\bv A})}{c_2\log^2 n}$. I.e., we do not have the condition requiring setting the diagonal to $0$. We will show that $|\lambda_i(\bv A) - \lambda_i(\bv A'')| \le \epsilon/2 \cdot \sqrt{\nnz(\bv{A})}$. By Weyl's inequality, and the assumption that $\nnz(\bv{A}) \ge 2/\epsilon^2$, we then have $|\lambda_i(\bv A) - \lambda_i(\bv A')| \le \epsilon/2 \cdot \sqrt{\nnz(\bv{A})} + 1 \le \epsilon \cdot \sqrt{\nnz(\bv{A})}$ as required.

Let $\mathcal{I}_k \subset [n]$ be the set of rows/columns with $\nnz(\bv A_i) \in \left[\frac{\nnz(\bv A)}{2^k}, \frac{\nnz(\bv A)}{2^{k-1}}\right)$ and $\bv A_{kl} = \bv A(\mathcal{I}_k, \mathcal{I}_l)$ be the submatrix of $\bv{A}$ formed with rows in $\mathcal{I}_k$ and columns in $\mathcal{I}_l$. Define $\bv{A}''_{kl}$ in the same way and observe that $\bv{A}''_{kl} = \bv{A}_{kl}$ whenever $2^{k+l} \le \frac{c_2 \nnz(\bv A)\log^2 n}{\epsilon^2}$.

When $2^{k+l} > \frac{c_2\nnz(\bv A) \log^2 n}{\epsilon^2}$, we may zero out some entries of $\bv{A}_{kl}$ to produce $\bv{A}_{kl}''$. Let $\bv{\widehat A}_{kl}$ be equal to $\bv{A}_{kl}$ on this set of zeroed out entries, and $0$ everywhere else.
Observe that  $(\bv{\widehat A}_{kl} \bv{\widehat A}_{kl}^T)_{m,:} = (\bv{\widehat A}_{kl})_{m,:}\bv{\widehat A}_{kl}^T$. Next observe that $(\bv{\widehat A}_{kl})_{m,:}$ has at most $\nnz(\bv{A}_{m}) \le \frac{\nnz(\bv A)}{2^{k-1}}$ non-zero entries. Similarly, each row of $\bv{\widehat A}_{kl}^T$ has at most $\frac{\nnz(\bv A)}{2^{l-1}}$ non-zero elements. Thus, for all $m \in |\mathcal{I}_k|$, using that $\norm{\bv{A}}_\infty \le 1$,
\begin{align*}
    \|(\bv{\widehat A}_{kl} \bv{\widehat A}_{kl}^T)_{m,:}\|_1 \leq \frac{\nnz(\bv A)^2}{2^{k+l-2}} = \frac{4\nnz(\bv A)^2}{2^{k+l}}.
\end{align*}
Applying Girshgorin's circle theorem (Theorem \ref{thm:girshgorin}) we thus have:
\begin{align}\label{eq:girsh}
    \|\bv{\widehat A}_{kl}\|_2^2 = \|\bv{\widehat A}_{kl}\bv{\widehat A}_{kl}^T\|_2 \leq \max_m \|(\bv{\widehat A}_{kl} \bv{\widehat A}_{kl}^T)_{m,:}\|_1 \leq  \frac{4\nnz(\bv A)^2}{2^{k+l}}.
\end{align}
Let $\bar{\bv A}_{kl} \in \mathbb{R}^{n \times n}$ be a symmetric matrix such that $\bar{\bv A}_{kl}(\mathcal{I}_k, \mathcal{I}_l)=\bv{\widehat A}_{kl}$, $\bar{\bv A}_{kl}(\mathcal{I}_l, \mathcal{I}_k)=\bv{\widehat A}_{lk}$, and $\bar{\bv A}_{kl}$ is zero everywhere else. By triangle inequality and the bound of \eqref{eq:girsh},
\begin{align*}
    \|\bar{\bv A}_{kl}\|_2 
    \leq \|\bv{\widehat A}_{kl}\|_2 + \|\bv{\widehat A}_{lk}\|_2
    \leq \frac{4\nnz(\bv A)}{2^{(k+l)/2}}.
\end{align*}
Observe that, since we assume all rows have at least one non-zero entry, $\nnz(\bv A_i) \geq 1$ and $\nnz(\bv A) \leq n^2$.  Therefore, $k,l$ can range from $1$ to $\log(n^2) = 2 \log n$. By triangle inequality, 
\begin{align*}
\|\bv A - \bv A''\|_2 &\le \sum_{(k,l): 2^{k+l} > \frac{c_2\nnz(\bv A)\log^2 n}{\epsilon^2}} \norm{\bv{\bar A}_{kl}}\\
&\le \sum_{k=1}^{2 \log n} \frac{4\epsilon \sqrt{\nnz(\bv{A})}}{\sqrt{c_2} \cdot \log n} \cdot \sum_{i=1}^{2 \log n} \frac{1}{2^{i-1}}\\
&\le \frac{16 \epsilon \sqrt{\nnz(\bv{A})}}{\sqrt{c_2}}.
\end{align*}
Finally, setting $c_2$ large enough and using Weyls' inequality (Fact \ref{fact:weyl_general}) we have the required bound:
\begin{align*}
    |\lambda_i(\bv A) - \lambda_i(\bv A'')| &\leq \epsilon/2 \sqrt{\nnz(\bv A)}.
\end{align*}
\end{proof}

We next give a bound on the coherence of the outlying eigenvectors of $\bv{A}'$. This bound is analogous to Lemma \ref{lemma:row_norm}, but is more refined, taking into account the sparsity of each row.

\begin{lemma}[Incoherence of outlying eigenvectors in terms of sparsity]
\label{lemma:general_row_norm}
Let $\bv A, \bv{A}' \in \mathbb{R}^{n\times n}$ be as in Lemma \ref{lem:nnz-zeroed}. Let $\bv A'_o= \bv V'_o\bv \Lambda'_o \bv V_o^{'T}$ where $\bv \Lambda'_o$ is diagonal, with the eigenvalues of $\bv A'$ with magnitude $\ge \epsilon \sqrt{\delta} \sqrt{\nnz(\bv{A})}$ on its diagonal, and $\bv{V}'_o$ has columns equal to the corresponding eigenvectors. Let $\bv V'_{o,i}$ denote the $i$\textsuperscript{th} row of $\bv V'_o$. Then, \begin{align*}
\|\bv{\Lambda}_o^{'1/2}\bv{V}'_{o,i} \|_2^2 \leq \frac{\nnz(\bv{A}_i)}{\epsilon \sqrt{\delta} \sqrt{\nnz(\bv{A})}}\hspace{1em}and\hspace{1em}\|\bv V'_{o,i}\|^2_2 \leq \frac{\nnz(\bv{A}_i)}{\epsilon^2\delta \nnz(\bv{A})}.
\end{align*}
\end{lemma}
\begin{proof}
The proof is nearly identical to that of Lemma \ref{lemma:row_norm}.
Observe that $\bv A' \bv V'_o =\bv V'_o \bv \Lambda'_o $. Letting $[\bv A' \bv V'_o]_i$ denote the $i$\textsuperscript{th} row of the $\bv A' \bv V'_o$, we have
\begin{equation}\label{Eq: row_norm1_nnz}
    \|[\bv A' \bv V'_o]_i\|_2^2 = \|[\bv V'_o \bv \Lambda'_o]_i\|_2^2 = \sum_{j=1}^r \lambda_j^2 \cdot \bv V_{o,i, j}^{'2},
\end{equation}
where $r=\rank(\bv A'_o)$, $\bv V'_{o,i,j}$ is the $(i,j)$\textsuperscript{th} element of $\bv V'_o$  and $\lambda_j=\bv \Lambda'_o(j,j)$. Since $\bv{V}'_o$ has orthonormal columns, we thus have $ \|[\bv A' \bv V'_o]_i\|_2^2 \le  \|\bv A'_i\|_2^2 \le \norm{\bv A_i}_2^2 \le \nnz(\bv{A}_i)$.
Therefore, by \eqref{Eq: row_norm1_nnz},
\begin{equation}\label{Eq: eig_bound_nnz }
    \sum_{j=1}^r \lambda_j^2 \cdot  \bv V_{o,i, j}^{'2} \leq \nnz(\bv{A}_i).
\end{equation}
Since by definition $\lvert \lambda_j \rvert \geq \epsilon \sqrt{\delta} \sqrt{\nnz(\bv{A})}$ for all $j$, we can concluse that $
   \|\bv{\Lambda}_o^{'1/2}\bv{V}'_{o,i} \|_2^2 =\sum_{j=1}^r \lambda_j \cdot  \bv V_{o,i, j}^{'2} \leq \frac{\nnz(\bv{A}_i)}{\epsilon \sqrt{\delta} \sqrt{\nnz(\bv{A})}}$
and
$
    \|\bv V'_{o,i}\|_2^2 = \sum_{j=1}^r \bv V_{o,i, j}^{'2} 
    \leq \frac{\nnz(\bv{A}_i)}{\epsilon^2\delta \nnz(\bv{A})}$, which completes the lemma.
\end{proof}

\subsection{Outer and Middle Eigenvalue Bounds}\label{sec:accuracy bounds_nnz}

Using Lemma \ref{lemma:general_row_norm}, we next argue that the eigenvalues of $\bv{A}_{o,S}'$ will approximate those of $\bv{A}'$, and in turn those of $\bv{A}$. The proof is very similar to Lemma \ref{lemma: orthonormality} in the uniform sampling case.

\begin{lemma}[Concentration of outlying eigenvalues with sparsity-based sampling]\label{lemma:nnz_large}
Let $\bv{A},\bv{A}' \in\mathbb{R}^{n\times n}$ be as in Lemmas \ref{lem:nnz-zeroed} and \ref{lemma:general_row_norm}.
Let $\bv A' = \bv A'_m + \bv A'_o$, where $\bv A'_m = \bv V'_m\bv{\Lambda}'_m\bv {\bv V'}_m^{T}$, and 
$\bv A'_o = \bv V'_o\bv{\Lambda}'_o\bv {\bv V'}_o^{T}$ are projections onto the eigenspaces with magnitude $< \epsilon\sqrt{\delta}\sqrt{\nnz(\bv A)}$ and $\ge \epsilon\sqrt{\delta}\sqrt{\nnz(\bv A)}$ respectively (analogous to Definition \ref{def:split})
As in Algorithm \ref{alg:nnz eigenvalue estimate}, for all $i \in [n]$ let $p_i=\min\left (1,\frac{s \nnz(\bv{A}_i)}{\nnz(\bv{A})}\right )$ and let $\bar{\bv S}$ be a scaled diagonal sampling matrix such that the $\bar{\bv S}_{ii}=\frac{1}{\sqrt{p_i}}$ with probability $p_i$ and $\bar{\bv S}_{ii}=0$ otherwise. If $s \geq \frac{c \log(1/(\epsilon \delta))}{\epsilon^3 \sqrt{\delta}}$ for a large enough constant $c$, then with probability at least $1-\delta$, $\|\bv \Lambda_o^{'1/2}\bv V_o^{'T} \bar{\bv S} \bar{\bv S}^T \bv V'_o\bv \Lambda_o^{'1/2} - \bv \Lambda'_o \|_2 \leq \epsilon \sqrt{\nnz(\bv{A})}$.   
\end{lemma}
\begin{proof}
Define $\bv{E} = \bv \Lambda_o^{'1/2}\bv V_o^{'T} \bar{\bv S} \bar{\bv S}^T \bv V'_o\bv \Lambda_o^{'1/2} - \bv \Lambda'_o$.
For all $i \in [n]$, let $\bv{V}_{o,i}$ be the $i^{th}$ row of $\bv{V}'_o$ and define the matrix valued random variable
\begin{align}
    \label{eq:rv_nnz}
    \bv Y_i = 
    \begin{cases}
    \frac{1}{p_i}\bv \Lambda_o^{'1/2}\bv V'_{o,i} \bv V_{o,i}^{'T}\bv \Lambda_o^{'1/2} , & \text{with probability } p_i\\
    0 & \text{otherwise.}
    \end{cases}
\end{align}
Define $\bv Q_i = \bv Y_i - \mathbb{E}\left[\bv Y_i\right]$. We can  observe that $\bv Q_1, \bv Q_2, \ldots, \bv Q_n$ are independent random variables and that $ \sum_{i=1}^n \bv Q_i=\bv \Lambda_o^{'1/2}\bv V_o^{'T} \bar{\bv S}\bar{\bv S}^T \bv V'_o\bv \Lambda_o^{'1/2} - \bv \Lambda'_o = \bv{E}$. Let $P=\{i \in [n]: p_i < 1\}$. Then, observe that $\sum_{i \in [n]\setminus P} \bv Q_i= 0$. So, $\bv{E}=\sum_{i \in P} \bv{Q}_i$. Then, similar to the proof of Lemma~\ref{lemma: orthonormality}, we need to bound $\|\bv{Q}_i \|_2$ for all $i \in P$ and  $\bv{Var}(\bv E) \eqdef\E(\bv{EE}^T)= \E(\bv{E}^T\bv{E})  =\sum_{i\in P} \mathbb{E}[\bv Q_i^2]$ using the improved row norm bounds of Lemma~\ref{lem:nnz-zeroed}. In particular, we have 
 \begin{align}
 \sum_{i\in P}\mathbb{E}[\bv Q_i^2] &= \sum_{i\in P} \left[ p_i \cdot \left(\frac{1}{p_i}-1\right)^2  + \left (1-p_i\right ) \right] \cdot ( \bv \Lambda_o^{1/2}\bv V_{o,i} \bv V_{o,i}^T \bv \Lambda_o \bv V_{o,i} \bv V_{o,i}^T \bv \Lambda_o^{1/2})\nonumber \\
 &\preceq \sum_{i\in P} \frac{1}{p_{i}} \cdot \norm{\bv \Lambda_o^{1/2} \bv V_{o,i}}_2^2 \cdot  (\bv \Lambda_o^{1/2}\bv V_{o,i}  \bv V_{o,i}^T\bv \Lambda_o^{1/2}).\label{var_ineq_nnz}
 \end{align}
 By Lemma \ref{lem:nnz-zeroed},  $\norm{\bv \Lambda_o^{1/2} \bv V_{o,i}}_2^2 \le \frac{\nnz(\bv{A}_i)}{\epsilon \sqrt{\delta} \sqrt{\nnz(\bv{A})}}$. Plugging back into \eqref{var_ineq_nnz},
 \begin{align*}
 \sum_{i\in P} \mathbb{E}[\bv Q_i^2] &\preceq\sum_{i\in P} \frac{1}{p_{i}} \cdot \frac{\nnz(\bv A_i)}{\epsilon \sqrt{\delta} \sqrt{\nnz(\bv A)}} \cdot  (\bv \Lambda_o^{1/2}\bv V_{o,i} \bv V_{o,i}^T \bv \Lambda_o^{1/2}) \\
 &\preceq \frac{\sqrt{\nnz(\bv A)}}{s \epsilon \sqrt{\delta}}(\sum_{i\in P} \Lambda_o^{1/2}\bv V_{o,i} \bv V_{o,i}^T \bv \Lambda_o^{1/2}) \\
 &= \frac{\sqrt{\nnz(\bv A)}}{s \epsilon \sqrt{\delta}}\bv{\Lambda}_o \preceq  \frac{\nnz(\bv A)}{s \epsilon \sqrt{\delta}} \cdot \bv{I}.
 \end{align*}
 Since $\bv{Q}_i^2$ is PSD this establishes that $v \leq \|\textbf{Var(E)} \|_2 \leq \frac{\nnz(\bv A)}{s \epsilon \sqrt{\delta}}$. Since there are at most $\frac{\nnz(\bv A)}{\delta\epsilon^2\nnz(\bv A)} = \frac{1}{\epsilon^2 \delta}$ eigenvalues with absolute value $\geq \epsilon\sqrt{\delta}\sqrt{\nnz(\bv A)}$, we can apply the matrix Bernstein inequality exactly as in the proof of Lemma \ref{lemma: orthonormality} with $d = \frac{1}{\epsilon^2 \delta}$ to show that when $s \geq \frac{c \log(1/(\epsilon \delta))}{\epsilon^3 \sqrt{\delta}}$ for large enough $c$, with probability at least $1-\delta$, $\left\|\bv E\right\|_2 \leq \epsilon \sqrt{\nnz(\bv{A})} $.
\end{proof}

We next bound the spectral norm of $\bv{A}'_{m,S}$. This is the most challenging part of the proof -- the rows of this matrix are sampled non-uniformly and scaled proportional to their inverse sampling probabilities, so we cannot apply existing bounds on the spectral norms of uniformly sampled random submatrices \cite{rudelson2007sampling}. We extend these bounds to the non-uniform case, critically using that entries which would be scaled up significantly after sampling (i.e. those lying in sparse rows/columns), have already been set to $0$ in $\bv{A}'_{m,S}$, and thus do not contribute to the spectral norm.

\begin{lemma}[Concentration of middle eigenvalues with sparsity-based sampling]
\label{lem:nnz-middle}
Let $\bv{A},\bv{A}' \in\mathbb{R}^{n\times n}$ be as in Lemmas \ref{lem:nnz-zeroed} and \ref{lemma:general_row_norm}.
Let $\bv A' = \bv A'_m + \bv A'_o$, where $\bv A'_m = \bv V'_m\bv{\Lambda}'_m\bv {\bv V'}_m^{T}$, and 
$\bv A'_o = \bv V'_o\bv{\Lambda}'_o\bv {\bv V'}_o^{T}$ are projections onto the eigenspaces with magnitude $< \epsilon\sqrt{\delta}\sqrt{\nnz(\bv A)}$ and $\ge \epsilon\sqrt{\delta}\sqrt{\nnz(\bv A)}$ respectively (analogous to Definition \ref{def:split}).
As in Algorithm \ref{alg:nnz eigenvalue estimate}, for all $i \in [n]$ let $p_i=\min\left (1,\frac{s \nnz(\bv{A}_i)}{\nnz(\bv{A})}\right )$ and let $\bar{\bv S}$ be a scaled diagonal sampling matrix such that the $\bar{\bv S}_{ii}=\frac{1}{\sqrt{p_i}}$ with probability $p_i$ and $\bar{\bv S}_{ii}=0$ otherwise. If $s \geq \frac{c\log^8 n}{\epsilon^8\delta^4}$ for a large enough constant $c$, then with probability at least $1-\delta$, $$\norm{\bar{\bv S}\bv A'_{m}\bar{\bv S}}_2 \leq \epsilon \sqrt{\nnz(\bv{A})}.$$
\end{lemma}
\begin{proof}
The initial part of the proof follows the outline of proof of the spectral norm bound for uniformly random submatrices (Theorem \ref{thm: tropp2008}) of~\cite{tropp2008norms}.
From Lemma \ref{lemma:general_row_norm}, we have $\|{\bv V'}_{o,i}\|_2 \leq \frac{\sqrt{\nnz({\bv A}_i)}}{\epsilon\sqrt{\delta} \sqrt{\nnz({\bv A})}}$. 
Also, following the proof of Lemma~\ref{lemma:general_row_norm},  we have $\|{\bv \Lambda'}_o {\bv V'}^T_{o,j} \|_2 = \|[{\bv A'}{\bv V'}_o]_j \|_2 \leq \sqrt{\nnz({\bv A}_j)}$.
Thus, for all $i,j \in [n]$, using Cauchy Schwarz's inequality, we have 
\begin{align}
\label{eq:nnz-aoij}
    |{\bv A'}_{o,i,j}|=|{\bv V'}_{o,i} {\bv \Lambda'}_o {\bv V'}_{o,j}^T| \leq \|{\bv V'}_{o,i}\|_2 \cdot \|{\bv \Lambda'}_o {\bv V'}_{o,j}^T\|_2 \leq \frac{\sqrt{\nnz({\bv A}_i)}}{\epsilon \sqrt{\delta} \sqrt{\nnz({\bv A})}} \cdot \sqrt{\nnz({\bv A}_j)}.
\end{align} 
Let ${\bv A'}_m = \bv H_m + \bv D_m$ where $\bv{H}_m$ and $\bv{D}_m$ contain the off-diagonal and diagonal elements of $\bv{A}'_m$ respectively. Note that while $\bv{A}'$ is zero on the diagonal, $\bv{A}'_m$ may not be. We have: 
\begin{align*}
    \E_2 \|\bar{\bv S}{\bv A'}_m \bar{\bv S}\|_2 \leq \E_2\|\bar{\bv S}\bv H_m \bar{\bv S}\|_2 + \E_2\|\bar{\bv S}\bv D_m \bar{\bv S}\|_2.
\end{align*}
Using Lemma \ref{lemma:coupling} (decoupling) on $\E_2\|\bv S\bv H_m \bar{\bv S}\|_2$, we get 
\begin{align}\label{eq:decoupled}
    \E_2 \|\bar{\bv S}{\bv A'}_m \bar{\bv S}\|_2 \leq 2\E_2\|\bar{\bv S}\bv H_m \hat{\bv S}\|_2 + \E_2\|\bar{\bv S}\bv D_m \bar{\bv S}\|_2,
\end{align}
where $\hat{\bv S}$ is an independent copy of $\bar{\bv S}$.  Upper bounding the rank of $\bv H_m$ as $n$ and applying Theorem \ref{thm: rudelson} twice to $\E_2\|\bar{\bv S}\bv H_m \hat{\bv S}\|_2$, once for each operator, we get 
\begin{align}
    \E_2\|\bar{\bv S}\bv H_m \hat{\bv S}\|_2 &\leq 5 \sqrt{\log n}\E_2\|\bar{\bv S}\bv H_m \hat{\bv S}\|_{1\to 2} + \E_2\|\hat{\bv S}\bv H_m\|_2\notag\\
    &\leq 5\sqrt{\log n}\E_2\|\bar{\bv S}\bv H_m \hat{\bv S}\|_{1\to 2} + 5\sqrt{\log n}\E_2\|\bv H_m\hat{\bv S}\|_{1\to 2} + \|\bv H_m\|_2\label{eq:nnz shs unsolved}.
\end{align}
Plugging \eqref{eq:nnz shs unsolved} into \eqref{eq:decoupled}, we have:
\begin{align}\label{eq:decoupled2}
    \E_2 \|\bar{\bv S}{\bv A'}_m \bar{\bv S}\|_2 \le 10\sqrt{\log n}\left (\E_2\|\bar{\bv S}\bv H_m \hat{\bv S}\|_{1\to 2} + \E_2\|\bv H_m\hat{\bv S}\|_{1\to 2}\right ) + 2\|\bv H_m\|_2 + \E_2\|\bar{\bv S}\bv D_m \bar{\bv S}\|_2
\end{align}

We now proceed to bound each of the terms on the right hand side of \eqref{eq:decoupled2}. We start with $\E_2\|\bar{\bv S}\bv D_m \bar{\bv S}\|_2$. First, observe that $\E_2\|\bar{\bv S}\bv D_m \bar{\bv S}\|_2 \leq \max_i \frac{1}{p_i}\lvert (\bv{D}_m)_{ii} \rvert$. We consider two cases.

\smallskip

\noindent \textbf{Case 1:} $p_i<1$. Then, $p_i=\frac{s \nnz(\bv{A}_i)}{\nnz(\bv{A})}$ and $\lvert (\bv{D}_m)_{ii} \rvert= \lvert ({\bv{A}'}_m)_{ii} \rvert = \lvert (\bv{A}'_o)_{ii} \rvert $ (since $\bv{A}'_{ii}=0$). Then  by \eqref{eq:nnz-aoij}, we have $\frac{1}{p_i}\lvert (\bv{D}_m)_{ii} \rvert \leq \frac{\sqrt{\nnz(\bv{A})}}{s\epsilon \sqrt{\delta}}  $.

 \smallskip

\noindent\textbf{Case 2:} $p_i=1$. Then we have $\frac{1}{p_i}\lvert (\bv{D}_m)_{ii} \rvert=\lvert (\bv{D}_m)_{ii} \rvert \leq \max_j \lvert (\bv{D}_m)_{jj} \rvert \leq \|\bv{A}'_m \|_2 \leq \epsilon \sqrt{\delta}\sqrt{\nnz(\bv{A})}$. \\
From the two cases above, for $s \geq \frac{1}{\epsilon^2 \delta}$, we have:
\begin{align}
    \label{eq:nnz-sds-bound}
    \E_2\|\bar{\bv S}\bv D_m \bar{\bv S}\|_2 \leq \epsilon\sqrt{\delta}\sqrt{\nnz({\bv A})}.
\end{align} 
We can bound $\norm{\bv{H}_m}_2$ similarly. Since $\bv H_m = {\bv A'}_m - \bv D_m$ and $\|{\bv A'}_m\|_2 \leq \epsilon\sqrt{\delta}\sqrt{\nnz({\bv A})}$,
\begin{align}
    \|\bv H_m\|_2 &\leq \|{\bv A'}_m\|_2 + \|\bv D_m\|_2\notag\\
    &\leq \epsilon\sqrt{\delta}\sqrt{\nnz({\bv A})} + \epsilon\sqrt{\delta}\sqrt{\nnz({\bv A})}\notag\\
    &= 2\epsilon\sqrt{\delta}\sqrt{\nnz({\bv A})}\label{eq:nnz-hm-bound}
\end{align}
where the second step follows from the fact that $\norm{\bv{D}_m}_2 \le \max_i \lvert (\bv{D}_m)_{ii} \rvert  \le \norm{\bv{A}'_m}_2$.

We next bound the term $\E_2\|\bv H_m\hat{\bv S}\|_{1\to 2}$.
 Observe that $\E_2\|\bv H_m\hat{\bv S}\|_{1\to 2} \leq \frac{\max_i \|{\bv A'}_{m,i}\|_2}{\sqrt{p_i}}$, where $\bv{A'}_{m,i}$ is the $i$\textsuperscript{th} column/row of $\bv{A}'_m$. We again consider the two cases when $p_i = 1$ and $p_i<1$:
 
 \smallskip
 
 \noindent\textbf{Case 1:} $p_i = 1$. Then $\|{\bv A'}_{m,i}\|_2 \leq \|{\bv A'}_m\|_2 \leq \epsilon\sqrt{\delta}\sqrt{\nnz({\bv A})}$.
 
  \smallskip

 \noindent\textbf{Case 2:} $p_i < 1$. Then $\|{\bv A'}_{m,i}\|_2 \leq \|{\bv A'}_i\|_2 \leq \sqrt{\nnz({\bv A}_i)}$. Thus, setting $s \geq \frac{1}{\epsilon^2\delta}$ we have:
 \begin{align*}
     \frac{\|{\bv A'}_{m,i}\|_2}{\sqrt{p_i}}
    &\leq \sqrt{\frac{\nnz({\bv A})}{s\nnz({\bv A}_i)}} \cdot \|{\bv A'}_i\|_2\\
    &\leq \sqrt{\frac{\nnz({\bv A})}{s}} \leq \epsilon\sqrt{\delta}\sqrt{\nnz({\bv A})}.
\end{align*}
Thus, from the two cases above, for all $i \in [n]$, adjusting $\epsilon$ by a $\frac{1}{\sqrt{\log n}}$ factor, we have
for $s \geq \frac{\log n}{\epsilon^2\delta}$:
\begin{align}
\label{eq:nnz-hms-bound}
    \E_2\|\bv H_m\hat{\bv S}\|_{1 \to 2} &\leq \frac{\epsilon\sqrt{\delta}\sqrt{\nnz({\bv A})}}{\sqrt{\log n}}.
\end{align}

Overall, plugging \eqref{eq:nnz-sds-bound}, \eqref{eq:nnz-hm-bound}, and \eqref{eq:nnz-hms-bound} back into \eqref{eq:decoupled2}, we have :
\begin{align}\label{eq:decoupled3}
    \E_2 \|\bar{\bv S}{\bv A'}_m \bar{\bv S}\|_2 \le 10\sqrt{\log n} \cdot \E_2\|\bar{\bv S}\bv H_m \hat{\bv S}\|_{1\to 2} + 15 \epsilon \sqrt{\delta}\sqrt{\nnz({\bv A})}.
\end{align}

It remains to bound $\E_2\|\bar{\bv S}\bv H_m \hat{\bv S}\|_{1\to 2}$, which is the most complex part of the proof.  Since $\hat{\bv S}$ is an independent copy of $\bar{\bv S}$, we denote the norm of the $i$\textsuperscript{th} column of $\bar{\bv S}\bv H_m\hat{\bv S}$ as $\frac{\|(\bar{\bv S}\bv H_m)_{:,i}\|_2}{\sqrt{p_i}}$. Then $\E_2\|\bar{\bv S}\bv H_m\hat{\bv S}\|_{1\to 2} \leq \E_2 \left (\max_{i: i\in [n]} \frac{\|(\bar{\bv S}\bv H_m)_{:,i}\|_2}{\sqrt{p_i}}\right )$. 
We will argue that $\max_{i: i\in [n]} \frac{\|(\bar{\bv S}\bv H_m)_{:,i}\|_2}{\sqrt{p_i}}$ is bounded by $\epsilon\sqrt{\delta} \sqrt{\nnz(\bv{A})}$ with probability $1-1/\poly(n)$. Since our sampling probabilities are all at least $1/n^2$ and since $\norm{\bv{H}_m}_F \le \norm{\bv{A}}_F \le n$, this value is also deterministically bounded by $n^2$. Thus, our high probability bound implies the needed bound on $ \E_2 \left (\max_{i: i\in [n]} \frac{\|(\bar{\bv S}\bv H_m)_{:,i}\|_2}{\sqrt{p_i}}\right )$.

We begin by observing that since ${\bv A'}_m = \bv H_m + \bv D_m$, $\|(\bar{\bv S}{\bv A'}_m)_{:,i}\|_2 \geq \|(\bar{\bv S}\bv H_m)_{:,i}\|_2$, and so to bound $\max_{i: i\in [n]} \frac{\|(\bar{\bv S}\bv H_m)_{:,i}\|_2}{\sqrt{p_i}}$, it suffices to bound $\frac{\|(\bar{\bv S}\bv A'_m)_{:,i}\|_2}{\sqrt{p_i}}$ for all $i\in [n]$. Towards this end, for a fixed $i$ and any $j \in [n]$, define 
\begin{align*}
    z_j &=
    \begin{cases}
    \frac{1}{p_j}|{\bv A'}_{m,i,j}|^2 & \text{with probability $p_j$}\\
    0 & \text{otherwise}.
    \end{cases}
\end{align*}
Then $\sum_{j=1}^n z_j = \|(\bar{\bv S}{\bv A'}_{m})_{:,i}\|_2^2$ and $\mathbb{E}[\sum_{j=1}^n z_j] = \|\bv{A}'_{m,i} \|_2^2 \leq \|\bv{A}'_{i} \|_2^2 \leq \nnz(\bv{A}_i)$. Since $\sum_{j=1}^n z_j = \|(\bar{\bv S}{\bv A'}_{m})_{:,i}\|_2^2$ is a sum of independent random variables, we can bound this quantity by applying Bernstein's inequality.  To do this, we must bound $|z_j|$ for all $j \in [n]$ and $\bv{Var}\left(\sum_{j=1}^n z_j\right)$.  We will again consider the cases of $p_i <1$ and $p_i = 1$ separately.

\smallskip

\noindent \textbf{Case 1:} $p_i<1$. Then, we have $p_i = s\nnz({\bv A}_i) / \nnz({\bv A})$. If ${\bv A'}_{i,j} \neq 0$ then
\begin{align*}
    |z_j| &\leq \frac{1}{p_j} |{\bv A'}_{m,i,j}|^2 \leq \max\left(1, \frac{\nnz({\bv A})}{s\nnz({\bv A}_j)}\right) |{\bv A'}_{m,i,j}|^2\\
    &\leq |{\bv A'}_{m,i,j}|^2 + \frac{2\nnz({\bv A})}{s\nnz({\bv A}_j)}\left(|{\bv A'}_{i,j}|^2 + |{\bv A'}_{o,i,j}|^2\right)\\
    &\leq |{\bv A'}_{m,i,j}|^2 + \frac{2\nnz({\bv A})}{s\nnz({\bv A}_j)}\left(|{\bv A'}_{i,j}|^2 + \frac{\nnz({\bv A}_i)\nnz({\bv A}_j)}{\epsilon^2\delta\nnz({\bv A})}\right)\\
    &\leq |{\bv A'}_{m,i,j}|^2 + \frac{2\nnz({\bv A})}{s\nnz({\bv A}_j)}|{\bv A'}_{i,j}|^2 + \frac{2\nnz({\bv A}_i)}{\epsilon^2\delta s},
\end{align*}
where the fourth inequality uses \eqref{eq:nnz-aoij}.
By the thresholding procedure which defines $\bv{A}'$, if $\bv A'_{ij} \neq 0$,
\begin{align}
    \nnz({\bv A}_i) \cdot \nnz({\bv A}_j) 
    \geq \frac{\epsilon^2 \nnz({\bv A})}{c_2\log^2 n}
    \Rightarrow
    \nnz({\bv A}_j) 
    \geq \frac{\epsilon^2 \nnz({\bv A})}{c_2 \log^2 n \nnz({\bv A}_i) },\label{eq:nnz-aj}
\end{align}
and thus for $p_i < 1$ and ${\bv A'}_{ij} \neq 0$ we have
\begin{align*}
    |z_j| &\leq |{\bv A'}_{m,i,j}|^2 + \frac{2c_2 \log^2 n \nnz({\bv A}_i)}{s\epsilon^2} + \frac{2\nnz({\bv A}_i)}{\epsilon^2\delta s}.
\end{align*}
If ${\bv A'}_{i,j} = 0$ then we simply have
\begin{align*}
    |z_j| &\leq |{\bv A'}_{m,ij}|^2 + \frac{\nnz({\bv A}_i)}{s\epsilon^2\delta}.
\end{align*}
Overall for all $j \in [n]$,
\begin{align}
    |z_j| &\leq |{\bv A'}_{m,i,j}|^2 + \frac{2\nnz({\bv A}_i)}{s\epsilon^2\delta} + \frac{2c_2\log^2 n \nnz({\bv A}_i)}{s\epsilon^2}, \label{eq:partial abs z}
\end{align}
and since $|{\bv A'}_{m,i,j}|^2 \leq  \sum_{j=1}^n |{\bv A'}_{m,i,j}|^2 = \|{\bv A'}_{m,i}\|_2^2 \leq \|{\bv A'}_i\|_2^2 \leq \nnz({\bv A}_i)$, 
\begin{align}
    |z_j| &\leq \nnz({\bv A}_i) + \frac{2\nnz({\bv A}_i)}{s\epsilon^2\delta} + \frac{2 c_2 \log^2 n \nnz({\bv A}_i)}{s\epsilon^2}. \label{eq:abs z}
\end{align}
For $s \geq c\left (\frac{\log^2 n}{\epsilon^2} + \frac{1}{\epsilon^2 \delta}\right )$ and large enough $c$, we thus have $|z_j| \leq 2\nnz({\bv A})$. 

\noindent We next bound the variance by:
\begin{align*}
    \bv{Var}\left(\sum_{j=1}^n z_j\right) &\leq \sum_{j=1}^n \E [z_j^2] \leq \sum_{j=1}^n p_j\frac{1}{p_j^2} |{\bv A'}_{m,i,j}|^4\\
    &= \sum_{j=1}^n \max \left(1, \frac{\nnz({\bv A})}{s\nnz({\bv A}_j)}\right)|{\bv A'}_{m,i,j}|^4 \\
    &\leq \sum_{j=1}^n |{\bv A'}_{m,i,j}|^4 + \sum_{j=1}^n \frac{12\nnz({\bv A})}{s\nnz({\bv A}_j)} \left(|{\bv A'}_{i,j}|^4 + |{\bv A'}_{o,i,j}|^4\right)\\
    &\leq \|{\bv A'}_{m,i}\|_2^4 + \sum_{j=1}^n \frac{12\nnz({\bv A})}{s\nnz({\bv A}_j)} \left(|\bv A_{i,j}'|^4 + \frac{\nnz(\bv A_i)^2\nnz(\bv A_j)^2}{\epsilon^4\delta^2\nnz(\bv A)^2}\right),
\end{align*}
where the last inequality uses \eqref{eq:nnz-aoij}. Now since $\bv A_{ii}' = 0$ for all $i$ and $\|\bv A'\|_\infty \leq 1$ we have
\begin{align}
    \bv{Var}\left(\sum_{j=1}^n z_j\right) &\leq \|{\bv A'}_{m,i}\|_2^4 + \sum_{j:{\bv A'}_{i,j}\neq 0}\frac{12\nnz({\bv A})}{s\nnz({\bv A}_j)} + \sum_{j=1}^n \frac{12\nnz({\bv A}_i)^2\nnz({\bv A}_j)}{s\epsilon^4\delta^2\nnz({\bv A})}. \label{eq:unsolved var sum zj 1}
\end{align}
\noindent Combining \eqref{eq:nnz-aj} with the second term to the right of \eqref{eq:unsolved var sum zj 1} we have
\begin{align*}
    \bv{Var}\left(\sum_{j=1}^n z_j\right) &\leq \|{\bv A'}_{m,i}\|_2^4 + \sum_{j:{\bv A'}_{i,j}\neq 0}\frac{12c_2\log^2 n \cdot \nnz({\bv A_i})}{s\epsilon^2} + \sum_{j=1}^n \frac{12\nnz({\bv A}_i)^2\nnz({\bv A}_j)}{s\epsilon^4\delta^2\nnz({\bv A})}, 
\end{align*}
and since $|\{j:{\bv A'}_{i,j}\neq 0 \}| = \nnz({\bv A}_i)$, we have
\begin{align}
    \bv{Var}\left(\sum_{j=1}^n z_j\right) &\leq \|{\bv A'}_{m,i}\|_2^4 + \frac{12c_2\log^2 n \cdot \nnz({\bv A_i})^2}{s\epsilon^2} + \sum_{j=1}^n \frac{12\nnz({\bv A}_i)^2\nnz({\bv A}_j)}{s\epsilon^4\delta^2\nnz({\bv A})}. \label{eq:unsolved var sum zj 2}
\end{align}
\noindent Finally since $\sum_{j=1}^n \nnz(\bv A_j) = \nnz(\bv A)$ and $\|{\bv A'}_{m,i}\|_2^4 \leq \norm{\bv{A'}_i}_2^4 \le \nnz(\bv A_i)^2$ we have
\begin{align}
    \bv{Var}\left(\sum_{j=1}^n z_j\right) &\leq \nnz({\bv A}_i)^2 + \frac{12c_2 \log^2 n \cdot \nnz({\bv A}_i)^2}{s\epsilon^2} + \frac{12\nnz({\bv A}_i)^2}{s\epsilon^4\delta^2}.\label{eq:var sum zj}
\end{align}

\noindent For $s \geq c\left (\frac{\log^2 n}{\epsilon^2} + \frac{1}{\epsilon^4\delta^2}\right )$ for large enough $c$, we have $\bv{Var}\left(\sum_{j=1}^n z_j\right) \leq 2\nnz({\bv A}_i)^2$. 

\noindent Therefore, using \eqref{eq:abs z} and \eqref{eq:var sum zj} with $s \geq c\left (\frac{\log^2 n}{\epsilon^2} + \frac{1}{\epsilon^4\delta^2}\right )$, we can 
apply Bernstein inequality (Theorem \ref{thm:bernstein}) (for some constant $c$) to get
\begin{align*}
    \Pr\left(\|(\bar{\bv S}\bv A'_m)_{:,i}\|_2^2 \geq \E\|(\bar{\bv S}\bv A'_m)_{:,i}\|_2^2 + t \right) &\leq \Pr\left(\sum_{j=1}^n z_j \geq \nnz(\bv{A}_i) + t \right)\\
    &\leq \exp\left(\frac{-t^2/2}{c\nnz({\bv A}_i)^2 + ct\nnz({\bv A}_i)/3}\right).
\end{align*}
If we set $t= \log n \cdot \nnz({\bv A}_i)$, for some constant $c'$ we have
\begin{align*}
     \Pr\left(\|(\bar{\bv S}\bv A'_m)_{:,i}\|_2^2 \geq \E\|(\bar{\bv S}\bv A'_m)_{:,i}\|_2^2+\log n \cdot \nnz({\bv A}_i) \right)
     &\leq \exp\left(\frac{-(\log n)^2/2}{c+c(\log n)/3}\right) \leq  \exp(-c'\log n) \leq 1/n^{c'}.
\end{align*}
Since ${\bv A'}_m = \bv H_m + \bv D_m$, we have $\|(\bar{\bv S}{\bv A'}_m)_{:,i}\|_2 \geq \|(\bar{\bv S}\bv H_m)_{:,i}\|_2$. Then with probability at least $1-1/n^{c'} \ge 1 -\delta$, for any row $i$ with $p_i <1$, we have 
\begin{align*}
    \frac{1}{p_i} \cdot \|(\bar{\bv S}\bv H_m)_{:,i}\|_2^2 \leq \frac{\nnz({\bv A})}{s\nnz({\bv A}_i)}\cdot c(\log n)\nnz({\bv A}_i) \leq \frac{\epsilon^2 \delta\nnz({\bv A})}{\log n},
\end{align*}
for $s \geq c\left (\frac{\log^2 n}{\epsilon^2} + \frac{1}{\epsilon^4\delta^2}\right )$ for large enough $c$. 
Observe that, as in Lemma \ref{lemma: orthonormality} w.l.o.g. we have assumed $1-n^{c'} \ge 1-\delta$, since otherwise, our algorithm would read all $n^2$ entries of the matrix.

\smallskip

\noindent \textbf{Case 2:} $p_i = 1$. Then, we have $\nnz({\bv A}_i) \geq \nnz({\bv A})/s$. As in the $p_i < 1$ case, we have from~\eqref{eq:partial abs z}: 
\begin{align*}
    |z_j| &\leq |{\bv A'}_{m,i,j}|^2 + \frac{2\nnz({\bv A}_i)}{s\epsilon^2\delta} + \frac{2c_2\log^2 n \nnz({\bv A}_i)}{s\epsilon^2}.
\end{align*}
Now, we observe that $|{\bv A'}_{m,i,j}|^2 \leq \sum_{j=1}^n |{\bv A'}_{m,i,j}|^2 \leq \| \bv{A}'_{m,i}\|^2_2 \leq \| \bv{A}'\|^2_2 \leq \epsilon^2\delta \nnz(\bv{A})$, which gives us
\begin{align}
    |z_j| &\leq \epsilon^2 \delta \nnz(\bv{A}) + \frac{2\nnz({\bv A}_i)}{s\epsilon^2\delta} + \frac{2c_2\log^2 n \nnz({\bv A}_i)}{s\epsilon^2}.\label{eq: abs zj p=1}
\end{align}
Thus, for $s \geq c\left (\frac{\log^2 n}{\epsilon^4\delta} + \frac{1}{\epsilon^4\delta^2}\right )$ for a large enough constant $c$ and adjusting for other constants we have $|z_j| \leq 2\epsilon^2\delta\nnz({\bv A})$.
Also observe that the expectation of $\sum z_j$ can be bounded by:
\begin{align*}
    \E\left(\sum_{j=1}^n z_j\right) = \E\|(\bar{\bv S}{\bv A'}_m)_{:,i}\|_2^2 = \|{\bv A'}_{m,i}\|_2^2 \leq \|{\bv A'}_m\|_2^2 \leq \epsilon^2\delta\nnz({\bv A}).
\end{align*}

\noindent Next, the variance of the sum of the random variables $\{z_j\}$ can again be bounded by following the analysis presented in and prior to \eqref{eq:unsolved var sum zj 2} and \eqref{eq:var sum zj} we have
\begin{align}
    \bv{Var}\left(\sum_{j=1}^n z_j\right) &\leq \|{\bv A'}_{m,i,j}\|_2^4 + \frac{12c_2 \log^2 n \cdot \nnz({\bv A}_i)^2}{s\epsilon^2} + \frac{12\nnz({\bv A}_i)^2}{s\epsilon^4\delta^2}\notag\\
    &\leq \epsilon^4\delta^2\nnz({\bv A})^2 + \frac{12c_2 \log^2 n \cdot \nnz({\bv A}_i)^2}{s\epsilon^2} + \frac{12\nnz({\bv A}_i)^2}{s\epsilon^4\delta^2},\label{eq: var z_j p=1}
\end{align}
where we again bound $\|{\bv A'}_{m,i,j}\|_2^4$ using $$|{\bv A'}_{m,i,j}|^2 \leq \sum_{j=1}^n |{\bv A'}_{m,i,j}|^2 \leq \| \bv{A}'_{m,i}\|^2_2 \leq \| \bv{A}'\|^2_2 \leq \epsilon^2\delta \nnz(\bv{A}).$$

\noindent Then for $s \geq c(\frac{\log^2 n}{\epsilon^6\delta^2} + \frac{1}{\epsilon^8\delta^4})$, we have $\bv{Var}\left(\sum_{j=1}^n z_j\right) \leq 2\epsilon^4\delta^2\nnz({\bv A})^2$ for large enough constant $c$.

\noindent Using \eqref{eq: abs zj p=1} and \eqref{eq: var z_j p=1} and noting that $\sum_{j=1}^n\E\left( z_j^2\right) \geq \bv{Var}\left(\sum_{j=1}^n z_j\right) - \E^2\left(\sum_{j=1}^n z_j\right)$ we can apply the Bernstein inequality (Theorem \ref{thm:bernstein}):
\begin{align*}
    \Pr\left(\|(\bar{\bv S}\bv A'_m)_{:,i}\|_2^2 \geq \E\|(\bar{\bv S}\bv A'_m)_{:,i}\|_2^2+t \right) &\leq \Pr\left(\sum_{j=1}^n z_j \geq \epsilon^2\delta\nnz(\bv{A}_i)+t \right) \\ &\leq \exp\left(\frac{-t^2/2}{c\epsilon^4\delta^2\nnz({\bv A})^2+c\epsilon^2\delta\nnz({\bv A})t/3}\right).
\end{align*}
If we set $t=(\log n)\epsilon^2\delta\nnz({\bv A})$, then for some constant $c'$ we have
\begin{align*}
    \Pr\left(\|(\bar{\bv S}\bv A'_m)_{:,i}\|_2^2 \geq \E\|(\bar{\bv S}\bv A'_m)_{:,i}\|_2^2+t \right) &\leq \exp(-c'\log n) \leq 1/n^{c'}.
\end{align*}
This, since $\|(\bar{\bv S}\bv H_m)_{:,i}\|_2^2 \leq \|(\bar{\bv S}\bv A'_m)_{:,i}\|_2^2$, when $p_i = 1$, setting  $s \geq c(\frac{\log^2 n}{\epsilon^6\delta^2} + \frac{1}{\epsilon^8\delta^4})$ for large enough $c$,  we have with probability $\ge 1-1/n^{c'}$ 
$
    \frac{1}{p_i}\|(\bar{\bv S}\bv H_m)_{:,i}\|_2^2 = \|(\bar{\bv S}\bv H_m)_{:,i}\|_2^2 \leq \|(\bar{\bv S}\bv A'_m)_{:,i}\|_2^2 \leq (\log n)\epsilon^2\delta\nnz({\bv A}).
$

We thus have, that with probability $\ge 1-1/n^{c'}$, for both cases when $p_i < 1$ and $p_i = 1$, $\frac{\|(\bar{\bv{S}}\bv H_m)_{:,i}\|_2^2}{p_i} \le (\log n)\epsilon^2\delta\nnz({\bv A})$. Taking a union bound over all $i \in [n]$,  with probability at least  $1-1/n^{c'-1}$, $ \max_i \frac{\|(\bar{\bv S}\bv H_m)_{:,i}\|_2}{\sqrt{p_i}} \leq \sqrt{\log n}\epsilon\sqrt{\delta}\sqrt{\nnz({\bv A})}$ for $s \geq c(\frac{\log^2 n}{\epsilon^6\delta^2} + \frac{1}{\epsilon^8\delta^4})$. As stated before, since $p_i \geq \frac{1}{n^2}$ for all $i \in [n]$, and since $\norm{\bv{H}_m}_F \le \norm{\bv{A}}_F \le n$, we also have $\max_i \frac{\|(\bar{\bv S}\bv H_m)_{:,i}\|_2}{\sqrt{p_i}} \leq n^2$. Thus, 
\begin{align*}
    \E_2 \left (\max_{i: i\in [n]} \frac{\|(\bar{\bv S}\bv H_m)_{:,i}\|_2}{\sqrt{p_i}}\right ) \leq \sqrt{\log n}\epsilon\sqrt{\delta}\sqrt{\nnz({\bv A})}(1-\frac{1}{n^{c'-1}})+\frac{1}{n^{c'-3}} \leq \sqrt{\log n}\epsilon\sqrt{\delta}\sqrt{\nnz({\bv A})}.
\end{align*} 
after adjusting $\epsilon$ by at most some constants.
Overall, we finally get 
$$\E_2\|\bar{\bv S}\bv H_m\hat{\bv S}\|_{1\to 2} \leq \E_2 \left (\max_{i: i\in [n]} \frac{\|(\bar{\bv S}\bv H_m)_{:,i}\|_2}{\sqrt{p_i}}\right ) \leq \epsilon\sqrt{\log n}\sqrt{\delta}\sqrt{\nnz({\bv A})}.$$
Plugging this bound into \eqref{eq:decoupled3}, we have for $s \geq c(\frac{\log^2 n}{\epsilon^6\delta^2} + \frac{1}{\epsilon^8\delta^4})$,
\begin{align*}
    \E_2\|\bar{\bv S}{\bv A'}_m \bar{\bv S}\|_2 &\leq (\log n)\epsilon\sqrt{\delta}\sqrt{\nnz({\bv A})}.
\end{align*}
Finally after adjusting $\epsilon$ by a $\frac{1}{\log n}$ factor, we have for $s \geq c(\frac{\log^{8} n}{\epsilon^6\delta^2} + \frac{\log^8 n}{\epsilon^8\delta^4})$ or $s \geq \frac{c\log^8 n}{\epsilon^8\delta^4}$,
\begin{align*}
    \E_2\|\bar{\bv S}{\bv A'}_m \bar{\bv S}\|_2 &\leq \epsilon\sqrt{\delta}\sqrt{\nnz({\bv A})}.
\end{align*}
The final bound then follows via Markov's inequality on $\|\bar{\bv S}{\bv A'}_m \bar{\bv S}\|_2$.
\end{proof}

\subsection{Main Accuracy Bound}

We are finally ready to prove our main result for sparsity-based sampling, which we restate below.
\nnzEigVal*

\begin{proof}
With Lemmas \ref{lemma:nnz_large} and \ref{lem:nnz-middle} in place, the proof is nearly identical to that of Theorem \ref{thm:main_bound}, with the additional need to apply Lemma \ref{lem:nnz-zeroed} to show that the eigenvalues of $\bv{A}'$ are close to those of $\bv A$.

For all $i \in [n]$ let $p_i=\min\left (1,\frac{s \nnz(\bv{A}_i)}{\nnz(\bv{A})}\right )$ and let $\bar{\bv S}$ be a scaled diagonal sampling matrix such that the $\bar{\bv S}_{ii}=\frac{1}{\sqrt{p_i}}$ with probability $p_i$ and $\bar{\bv S}_{ii}=0$ otherwise. Let $\bv{A}'$ be the matrix constructed from $\bv{A}$ by zeroing out its elements as described in Lemma~\ref{lem:nnz-zeroed}. Then, note that $\bar{\bv S}\bv{A}'\bar{\bv S}=\bv{A}'_{S}$ where $\bv{A}'_S$ is the submatrix constructed as in Algorithm~\ref{alg:nnz eigenvalue estimate}. We first show that the eigenvalues of $\bv{A}'_S$ approximate those of $\bv{A}'$ up to error $\epsilon\sqrt{\nnz(\bv{A})}$. The steps are almost identical to those in the proof of Theorem~\ref{thm:main_bound}. We provide a brief outline of the steps but skip the details. 

We split $\bv{A}'$ as $\bv A' = \bv A'_o + \bv A'_m$ where $\bv A'_o$ and $ \bv A'_m$ contain eigenvalues of $\bv{A}'$ of magnitudes $< \epsilon\sqrt{\delta} \sqrt{\nnz(\bv{A})}$ and$\ge \epsilon\sqrt{\delta} \sqrt{\nnz(\bv{A})}$. This implies $\bv{A}'_S=\bv A'_{o,S}+\bv A'_{m,S}$ where $\bv{A}'_{o,S}=\bar{\bv{S}}\bv A'_o \bar{\bv{S}}$ and $\bv A'_{m,S}=\bar{\bv{S}} \bv A'_m \bar{\bv{S}}$. By Fact \ref{def: equality of eigenvalues} we have that the nonzero eigenvalues of $ \bv{A}'_{o,S} = \bar{\bv S} \bv V'_o \bv \Lambda'_o \bv V_o^{'T} \bar{\bv S}$ are identical to those of $\bv{\Lambda}_o^{'1/2} \bv V_o^{'T} \bar{\bv S} \bar{\bv S} \bv V'_o \bv{\Lambda}_o^{'1/2}$. Thus, applying the perturbation bound of Fact \ref{fact:weyl_general}, we have:
 \begin{align*}
     \left| \lambda_i(\bv{\Lambda}_o^{'1/2} \bv V_o^{'T} \bar{\bv S} \bar{\bv S} \bv V'_o \bv{\Lambda}_o^{'1/2} ) - \lambda_i(\bv \Lambda'_o)\right| \leq  C\log n\|\bv{\Lambda}_o^{'1/2} \bv V_o^{'T} \bar{\bv S} \bar{\bv S} \bv V'_o \bv{\Lambda}_o^{'1/2}  - \bv \Lambda'_o\|_2.
 \end{align*}
 From Lemma~\ref{lemma:nnz_large}, we get $ \|\bv{\Lambda}_o^{'1/2} \bv V_o^{'T} \bar{\bv S} \bar{\bv S} \bv V'_o \bv{\Lambda}_o^{'1/2}  - \bv \Lambda'_o\|_2 \leq \epsilon \sqrt{\nnz(\bv{A})}$ for $s \geq \frac{c\log(1/(\epsilon \delta))}{\epsilon^3 \sqrt{\delta}}$ with probability at least $1-\delta$. Thus, setting the error parameter to $\frac{\epsilon}{\log n}$ in Lemma~\ref{lemma:nnz_large}, for $s \geq \frac{c\log(1/(\epsilon \delta))\log^3 n}{\epsilon^3 \sqrt{\delta}}$, with probability at least $1-\delta$ we have:
\begin{align}
   \left| \lambda_i(\bv{\Lambda}_o^{'1/2} \bv V_o^{'T} \bar{\bv S} \bar{\bv S} \bv V'_o \bv{\Lambda}_o^{'1/2} ) - \lambda_i(\bv \Lambda'_o)\right| &< \epsilon \sqrt{\nnz(\bv{A})}.
   \label{eq:nnz_outer_weyl}
\end{align}
We have thus shown that the non-zero eigenvalues of $ \bv{A}'_{o,S}$ approximate all outlying eigenvalues of $\bv{A}'$. Note that by Lemma~\ref{lem:nnz-middle}, we also have $\|\bv A'_{m,S} \|_2 \leq \epsilon \sqrt{\nnz(\bv{A})}$  with probability at least $1-\delta$ for $s \geq \frac{c\log^8 n}{\epsilon^8\delta^4}$. Then, similarly to the section on eigenvalue alignment of Theorem~\ref{thm:main_bound}, we can argue how these approximations `line up' in the presence of zero eigenvalues in the spectrum of these matrices, concluding that, for all $i \in [n]$,
\begin{align*}
\left |\tilde \lambda_i(\bv{A}) - \lambda_i(\bv{A}') \right | \le \epsilon \sqrt{\nnz(\bv{A})}. 
\end{align*}
 
\noindent Finally, by Lemma~\ref{lem:nnz-zeroed}, we have $\lvert \lambda_i(\bv{A}') -\lambda_i(\bv{A})\rvert \leq \epsilon \sqrt{\nnz(\bv{A})}$ for all $i \in [n]$. Thus, via triangle inequality, $
\left |\tilde \lambda_i(\bv{A}) - \lambda_i(\bv{A}) \right | \le 2\epsilon \sqrt{\nnz(\bv{A})}$, which gives the required bound after adjusting  $\epsilon$ to $\epsilon/2$.

Recall that we require $s \geq  \frac{c \log(1/(\epsilon \delta)) \cdot \log^3 n}{\epsilon^3 \sqrt{\delta}}$ for \eqref{eq:nnz_outer_weyl} to hold with probability $1-\delta$. We also require $s \geq \frac{c\log^8n}{\epsilon^8\delta^4}$ for $\|\bv A_{m,S}'\|_2 \leq \epsilon\sqrt{\nnz(\bv A)}$ to hold with probability $1-\delta$ by Lemma \ref{lem:nnz-middle}. Thus, for both conditions to hold simultaneously with probability $1-2\delta$ by a union bound, it suffices to set  $s = \frac{c\log^8n}{\epsilon^8\delta^4} \ge \max\left(\frac{c \log(1/(\epsilon \delta)) \cdot \log^3 n}{\epsilon^3 \sqrt{\delta}}, \frac{c\log^8n}{\epsilon^8\delta^4}\right)$, where we use that $\log(1/(\epsilon \delta) \le O(\log n)$, as otherwise our algorithm can take $\bv{A}_S$ to be the full matrix $\bv{A}$. Adjusting $\delta$ to $\delta/2$ completes the theorem.
\end{proof}


\section{Empirical Evaluation}\label{sec: empirical evalutation}

We complement our theoretical results by evaluating Algorithms \ref{alg:eigenvalue estimate} (uniform sampling) and Algorithm \ref{alg:nnz eigenvalue estimate} (sparsity-based sampling) in approximating the eigenvalues of several symmetric matrices. We defer an evaluation of Algorithm \ref{alg:l2_eig_est} (norm-based sampling) to later work. Algorithm \ref{alg:eigenvalue estimate} and Algorithm \ref{alg:nnz eigenvalue estimate} perform very well. They seem to have error dependence roughly $1/\epsilon^2$ in practice, as compared to the $1/\epsilon^3$ dependence proven in Theorem \ref{thm:main_bound} and $1/\epsilon^8$ dependence in Theorem \ref{thm:nnz_main_bound}. Closing the gap between the theory and observed results would be very interesting. 

\subsection{Datasets}\label{subsec: dataset}

We test Algorithm~\ref{alg:eigenvalue estimate} (uniform sampler) on three dense matrices. We also compare the relative performance of Algorithm \ref{alg:eigenvalue estimate} and Algorithm \ref{alg:nnz eigenvalue estimate} (sparsity sampler) on three other synthetic and real world matrices.
 
The first two dense matrices, following \cite{cai2021fast}, are created by sampling $5000$ points from a binary image. We then normalize all the points in the range $[0,1]$ in both axes. The original image and resulting set of points are shown in Figure \ref{fig:kong_dataset}. We then compute a similarity matrix for the points using two common similarity functions used in machine learning and computer graphics: $\delta(\bv x,\bv y) = \tanh\left({\frac{\langle \bv x,\bv y\rangle}{2}}\right)$, the hyperbolic tangent; and $\delta(\bv x,\bv y) = {\norm{\bv x-\bv y}_2^2} \cdot \log\left({\norm{\bv x-\bv y}_2^2}\right)$, the thin plane spline. 
These measures lead to symmetric, indefinite, and entrywise bounded similarity matrices. 

Our next dense matrix (called the block matrix) is based on the construction of the hard instance for the lower bound in~\cite{BakshiChepurkoJayaram:2020} which shows that we need $\Omega(1/\epsilon^2) \times \Omega(1/\epsilon^2)$ samples to compute $\epsilon n$ approximations to the eigenvalues of a bounded entry matrix. It is a $5000\times5000$ matrix containing a $2500\times2500$ principal submatrix of all $1$s, with the rest of the entries set to $0$. 
It has $\lambda_1(\bv{A}) = 2500$ and all other eigenvalues equal to $0$.

We now describe the three matrices used to compare Algorithm \ref{alg:eigenvalue estimate} and Algorithm \ref{alg:nnz eigenvalue estimate}. All three are graph adjacency matrices, which are symmetric, indefinite, entrywise bounded and sparse. 
Spectral density estimation for graph structured matrices is an important primitive in network analysis \cite{dong2019network}. The first is a dense \erdos\ graph with $5000$ nodes and connection probability $0.1$. 
The second two are real world graphs, taken from SNAP \cite{snapnets}; namely Facebook \cite{mcauley2012learning} and Arxiv COND-MAT \cite{leskovec2007graph}. The Facebook graph contains 
$4039$ nodes and $88234$ directed edges. We symmetrize the adjacency matrix. Arxiv COND-MAT is a collaboration network between authors of Condensed Matter papers published on arXiv, containing 
$23133$ nodes and $93497$ undirected edges. 
Both these graphs are very sparse -- the number of edges is $\leq 1\%$ of the total edges in a complete graph with same number of nodes.

\begin{figure}[H]
    \centering
    \includegraphics[width=0.3\textwidth]{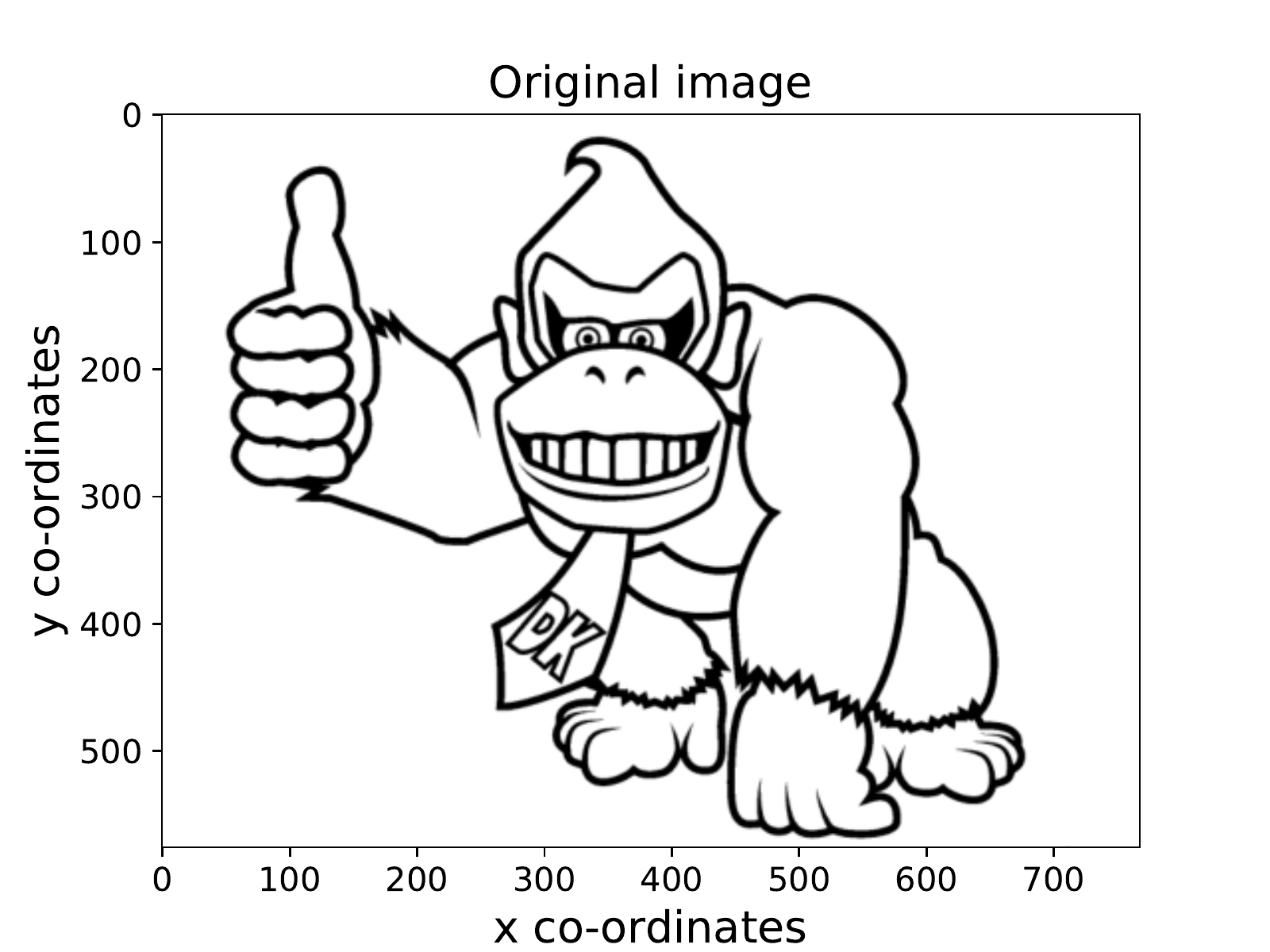}
    \includegraphics[width=0.3\textwidth]{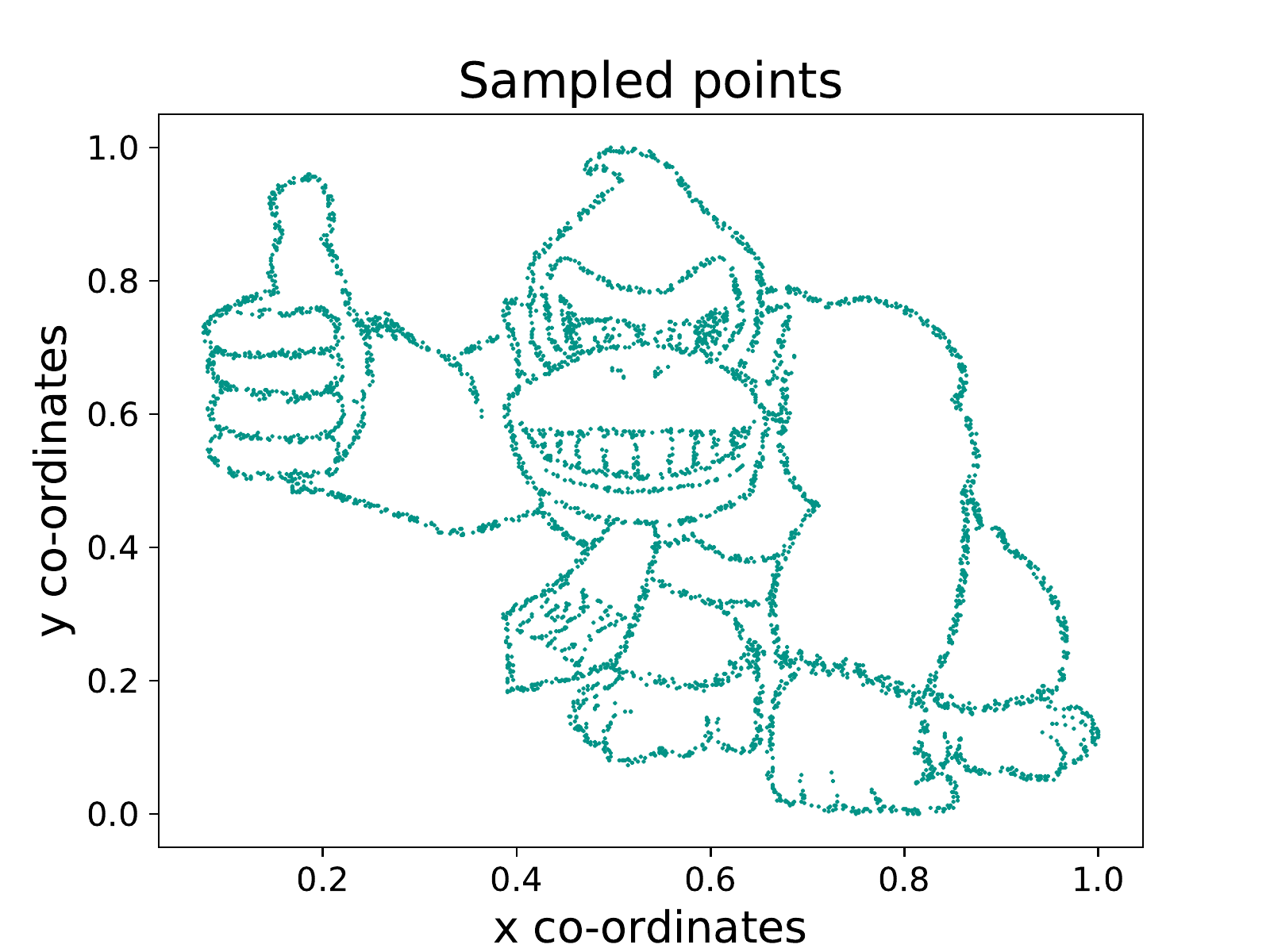}
    \caption{\textbf{Kong dataset}. The image on the left is the original synthetic binary image and the image on the right shows the 5000 sampled points from the outline used as dataset in our experiments.}
    \label{fig:kong_dataset}
    \vspace{-1em}
\end{figure}

\subsection{Implementation Details}\label{subsec:implementation}
Apart from uniform random sampling (Algorithm \ref{alg:eigenvalue estimate}), we also apply the sparsity-based sampling technique in Algorithm \ref{alg:nnz eigenvalue estimate} and a modification to Algorithm \ref{alg:nnz eigenvalue estimate}, where we do not zero out the elements of the sampled submatrix $\bv A_S$ (we call this \emph{simple sparsity sampler}). In practice, to apply Algorithm \ref{alg:nnz eigenvalue estimate}, we zero out element $[\bv A_S]_{i,j}$ (line 5 of Algorithm \ref{alg:nnz eigenvalue estimate}) if $i=j$ or $\nnz(\bv A_i)\nnz(\bv A_j) < \frac{\nnz(\bv A)}{c_2 s}$, where $c_2$ is a constant and $s$ is the size of the sample. We set $c_2 = 0.1$ experimentally as this results in consistent behavior across datasets. 

\subsection{Experimental Setup}\label{subsec:experiments}
We subsample each matrix and compute its eigenvalues using numpy \cite{numpyeigvals}.
We then use our approximation algorithms to estimate the eigenvalues of $\bv{A}$ by scaling the eigenvalues of the sampled submatrix. For $t$ trials, we report the logarithm of the average absolute scaled error, $\log \left(\frac{1}{t}\sum\frac{|\Tilde \lambda_{i,t}(\bv A) - \lambda_i(\bv A)|}{\sqrt{\nnz(\bv A)}}\right)$, where $\Tilde \lambda_{i,t}(\bv A)$ is the estimated eigenvalue in the $t^{th}$ trial, $\lambda_i(\bv A)$ is the true eigenvalue and $\nnz(\bv A)$ is the number of non-zero elements in $\bv A$. Recall that $\sqrt{\nnz(\bv A)} \ge \norm{\bv A}_F$ is an upper bound on all eigenvalue magnitudes. 
Also note that for the fully dense matrices, $\sqrt{\nnz(\bv A)} \approx n$. 
 
We repeat our experiments for $t = 50$ trials at different sampling rates and aggregate the results. The resultant errors of estimation for dense matrices are plotted in Figure~\ref{fig:errors1} and for the graph matrices are plotted in Figure~\ref{fig:errors2}. The $x$-axis is the log proportion of the number of random samples chosen from the matrix. If we sample $1\%$ of the rows/columns, then the $\log$ comes to around $-4.5$. 
In these log-log plots, if the sample size has polynomial dependence on $\epsilon$, e.g., $\epsilon n$ or $\epsilon\sqrt{\nnz(\bv A )}$ error is achieved with sample size proportional to $1/\epsilon^p$, we expect to see error falling off linearly, with slope equal to $-1/p$ where $p$ is the exponent on $\epsilon$. 
 
As a baseline we also show the error if we approximate all eigenvalues with $0$ which results in an error of $\frac{\lambda_i}{\sqrt{\nnz(\bv A)}}$. This helps us to observe how the approximation algorithms perform for both large and small order eigenvalues, as opposed to just approximating everything by $0$.

\noindent\textbf{Code.} All codes are written in Python and available at \url{https://github.com/archanray/eigenvalue_estimation}. 

\subsection{Summary of Results}\label{subsec:results}

Our results are plotted in Figures~\ref{fig:errors1} and \ref{fig:errors2}. We observe relatively small error in approximating all eigenvalues, with the error decreasing as the number of samples increases. 
What is more interesting is that the relationship between sample size and error $\epsilon n$ seems to be generally on the order of $1/\epsilon^2$, our expected lower bound for approximating eigenvalues by randomly sampling a principal submatrix. This can be seen by observing the slope of approximately $-1/2$ on the log-log error plots. In some cases, we do better in approximating small eigenvalues of $\bv{A}$ -- if the eigenvalue lies well within the range of middle eigenvalues, i.e. $\{-\epsilon n, \epsilon n\}$), we may achieve a very good absolute error estimate simply by approximating it to $0$. 

As expected, on the graph adjacency matrices (in Figure~\ref{fig:errors2}), sparsity-based sampling techniques generally achieve better error than uniform sampling. For the \erdos\ graph, we expect the node degrees (and hence row sparsities) to be similar. Thus the sampling probability for each row will be roughly uniform, which leads to similar performance of sparsity-based techniques and uniform sampling. For the real world graphs, which have power law degree distributions, sparsity-based sampling techniques has a significant effect. As a result Algorithm \ref{alg:nnz eigenvalue estimate}, and the simple sparsity sampler variant significantly outperform uniform sampling. 

Algorithm \ref{alg:nnz eigenvalue estimate}  almost always dominates simple sparsity sampler. In some cases simple sparsity sampler performs better or equivalent to Algorithm \ref{alg:nnz eigenvalue estimate}. This may happen because for two reasons: 1) if Algorithm \ref{alg:nnz eigenvalue estimate} zeroes out almost all of the sampled submatrix $\bv A_S$ for small samples, the algorithm will underestimate the corresponding eigenvalue, and 2) the cut-off threshold for the term $\nnz(\bv A_i)\nnz(\bv A_j)$ may be too high leading to no difference between simple sparsity sampler and Algorithm \ref{alg:nnz eigenvalue estimate}.

We also observe that approximating all eigenvalues with 0 results in very good approximation for small eigenvalues of the \erdos\ graph. We believe this is because the smaller eigenvalues are significantly less than the largest eigenvalue (of the order of $3500$). We see similar trends of approximating eigenvalues with zero for the real world graphs too. But since eigenvalues at the extreme spectrum are of a larger order, we see reasonably good approximation for the sampling algorithms. Algorithm \ref{alg:nnz eigenvalue estimate} outperforms approximation by $0$ in all of these cases.

In the dense matrices uniform sampling almost always outperforms approximation by $0$ when estimating any reasonably large eigenvalues. Additionally, note that the block matrix is rank-$1$ with true eigenvalues $\{2500, 0, \ldots, 0\}$. Any sampled principal submatrix will also have rank at most $1$. Thus, outside the top eigenvalue, the submatrix will have all zero eigenvalues. So, in theory, our algorithm should give perfect error for all eigenvalues outside the top -- we see that this is nearly the case. The very small and sporadic error in the plots for these eigenvalues arises due to numerical roundoff in the eigensolver. The only non-trivial approximation for this matrix is for the top eigenvalue. This approximation seems to have error dependency around $1/\epsilon^2$, as expected.  

\section{Conclusion}\label{sec:conclusion}

We present efficient algorithms for estimating all eigenvalues of a symmetric matrix with bounded entries up to additive error $\epsilon n$, by reading just a $\poly(\log n,1/\epsilon) \times \poly(\log n,1/\epsilon) $ random principal  submatrix. We give improved error bounds of $\epsilon \sqrt{\nnz(\bv A)}$ and $\epsilon \norm{\bv A}_F$ when the  rows/columns are sampled with probabilities proportional to their sparsities or squared $\ell_2$ norms, respectively. 

As discussed, our work leaves  several open questions. In particular, it is open if our query complexity for $\pm \epsilon n$ approximation can be improved, possibly to $\tilde O(\log^c n/\epsilon^4)$ total entries using principal submatrix queries or $\tilde O(\log^c/\epsilon^2)$ entries using general queries. The later bound is open even when $\bv{A}$ is PSD, a setting where we know that sampling a $O(1/\epsilon^2) \times O(1/\epsilon^2)$ principal submatrix (with $O(1/\epsilon^4)$ total entries) does suffice. Additionally, it is open if we can achieve sample complexity independent of $n$, by removing all $\log n$ factors, as have been done for the easier problem of testing positive semidefiniteness \cite{BakshiChepurkoJayaram:2020}. See Section \ref{optimal} for more details.

It would also be interesting to extend our results to give improved approximation bounds for other properties of the matrix spectrum, such as various Schatten-$p$ norms and spectral summaries. For many of these problems large gaps in understanding exist -- e.g., for $\pm n^{3/2}$ approximation to the Schatten-$1$ norm, which requires $\Omega(n)$ queries, but for which no $o(n^2)$ query algorithm is known. Applying our techniques to improve sublinear time PSD testing algorithms under an $\ell_2$ rather than $\ell_\infty$ approximation requirement \cite{BakshiChepurkoJayaram:2020} would also be interesting. Finally, it would be interesting to identify additional assumptions on $\bv{A}$ or on the sampling model where stronger approximation guarantees (e.g., relative error) can be achieved in sublinear time.

\subsection*{Acknowledgements}

We thank Ainesh Bakshi, Rajesh Jayaram, Anil Damle, and Christopher Musco for helpful conversations about this work. RB, CM and AR was partially supported by an Adobe Research grant, along with NSF Grants 2046235 and 1763618. PD and GD were partially supported by NSF AF 1814041, NSF FRG 1760353, and DOE-SC0022085. 

\begin{figure}[H]
    \centering
    \begin{subfigure}[t]{\textwidth}
    \centering
    \includegraphics[width=0.32\textwidth]{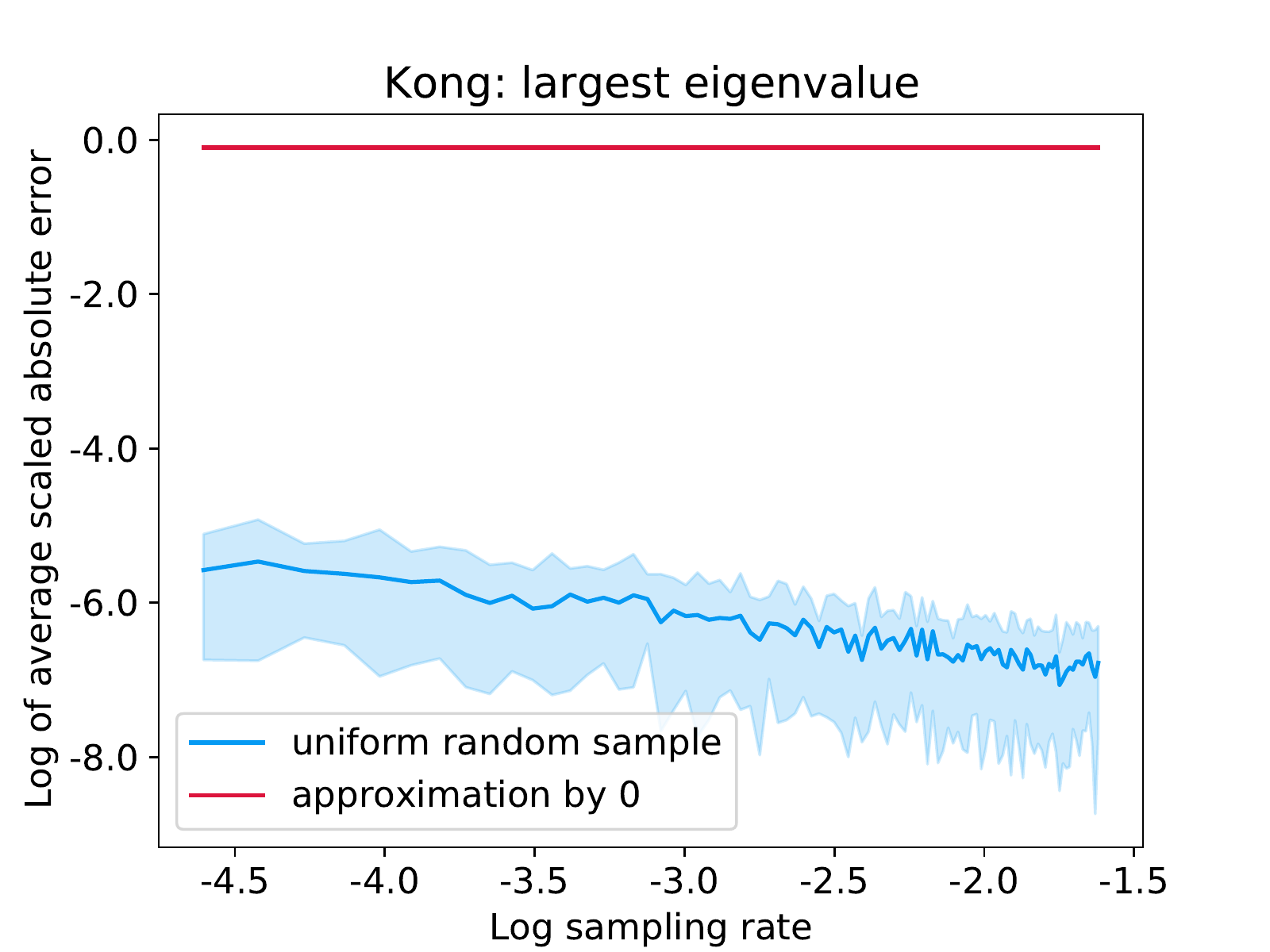}
    \includegraphics[width=0.32\textwidth]{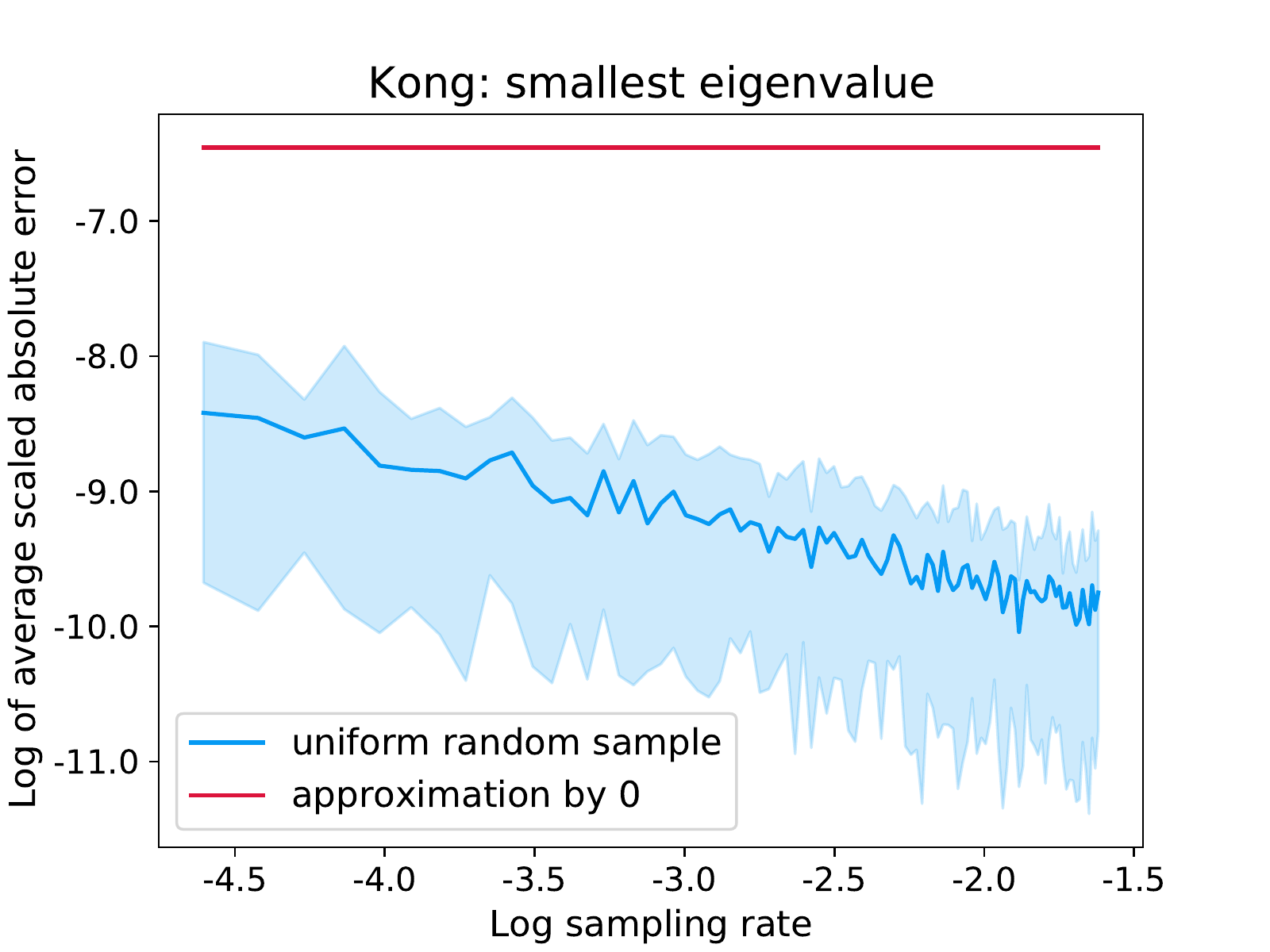}
    \includegraphics[width=0.32\textwidth]{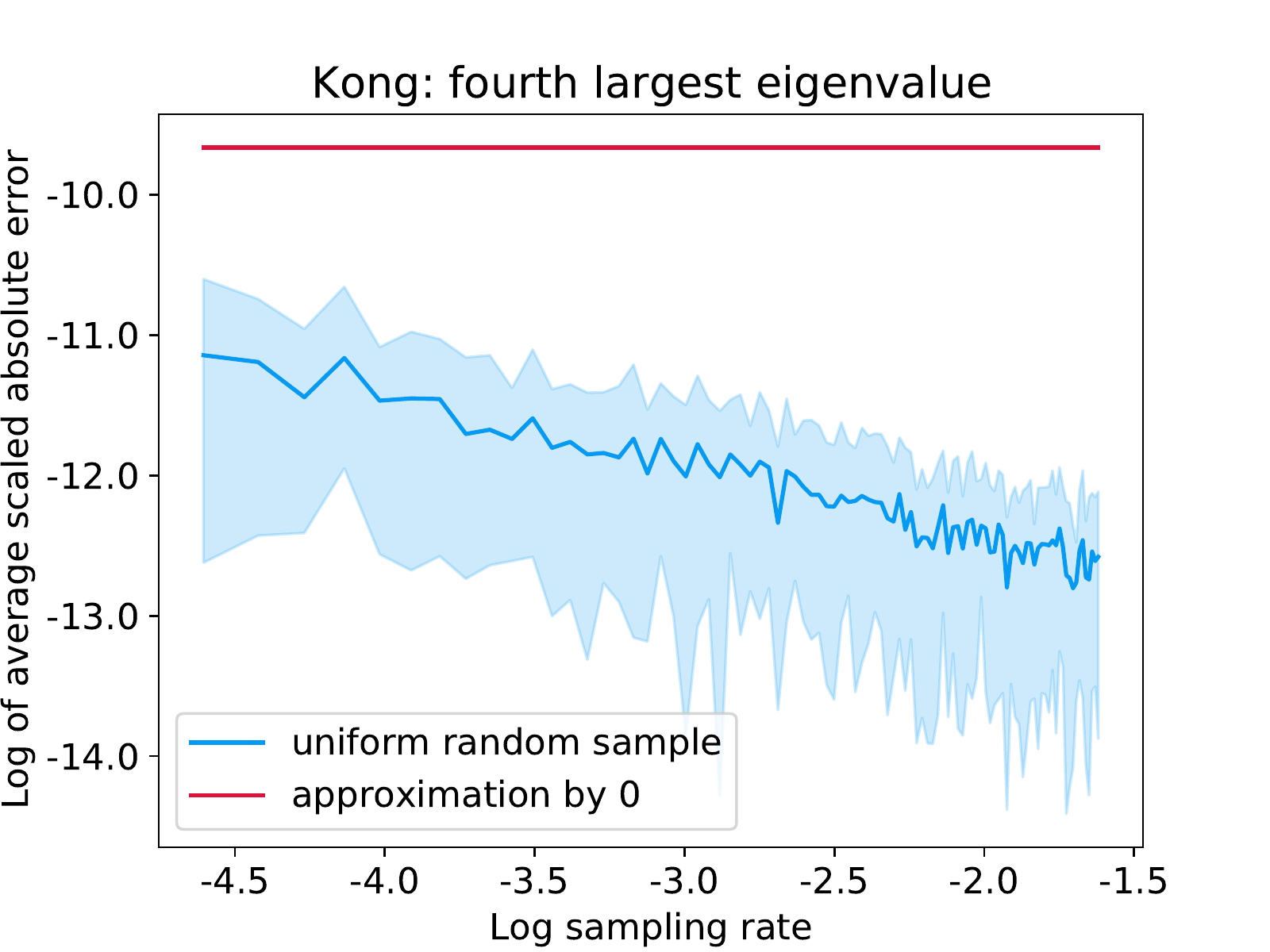}
    \caption{Hyperbolic tangent similarity matrix.} 
    \end{subfigure}
    \vskip\baselineskip
    \vspace{-1.3em}
    \begin{subfigure}[t]{\textwidth}
    \centering
    \includegraphics[width=0.32\textwidth]{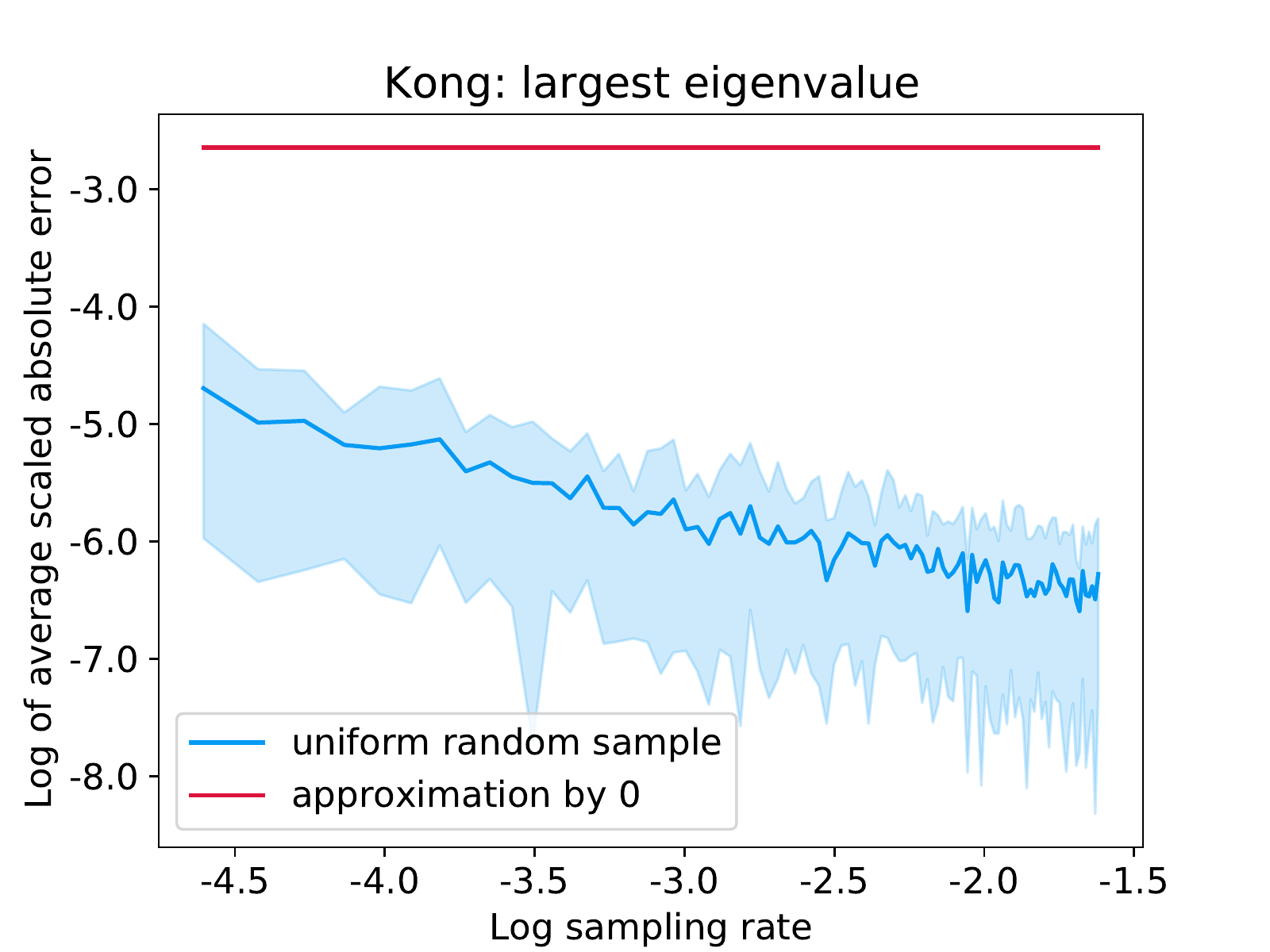}
    \includegraphics[width=0.32\textwidth]{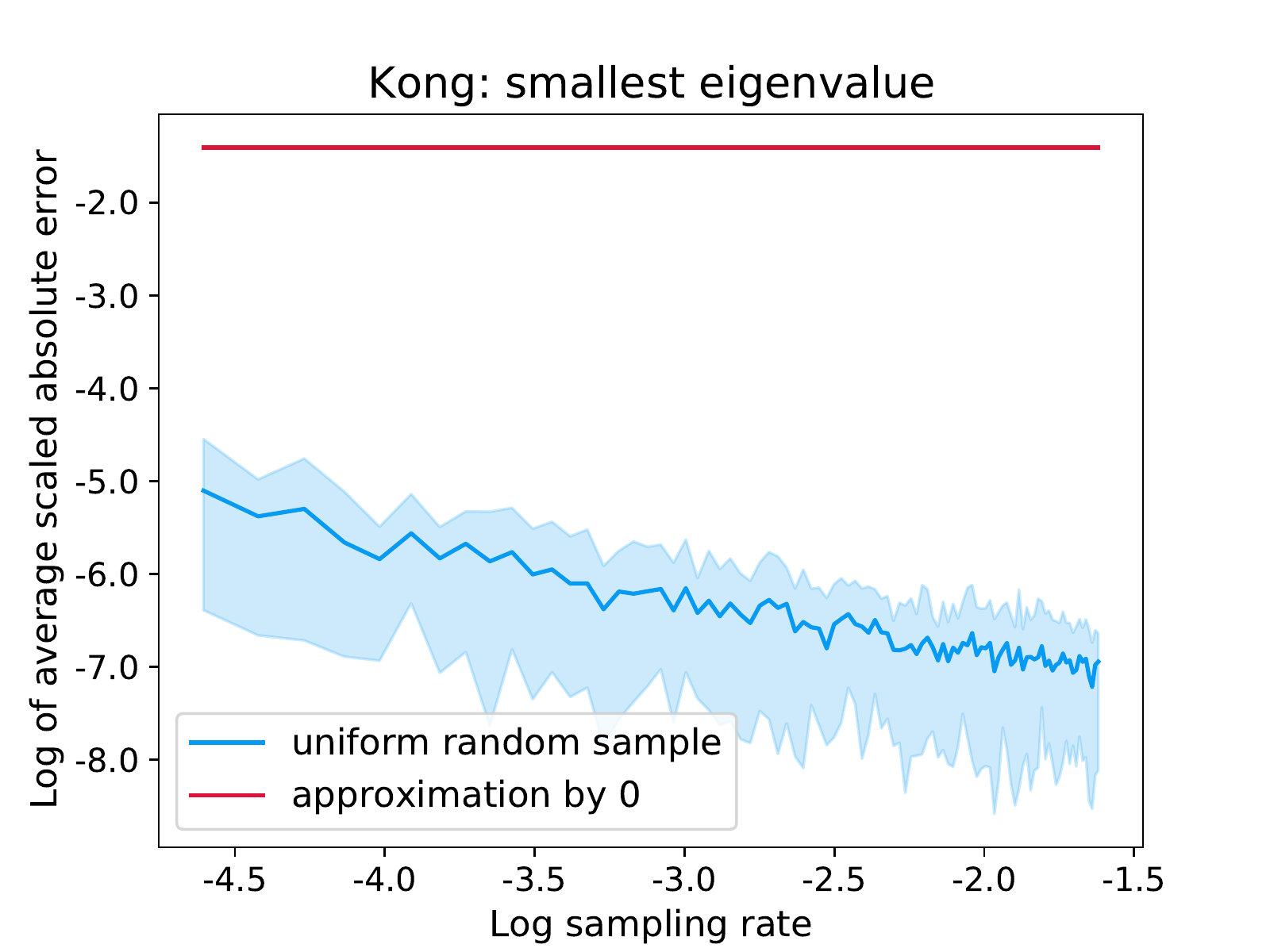}
    \includegraphics[width=0.32\textwidth]{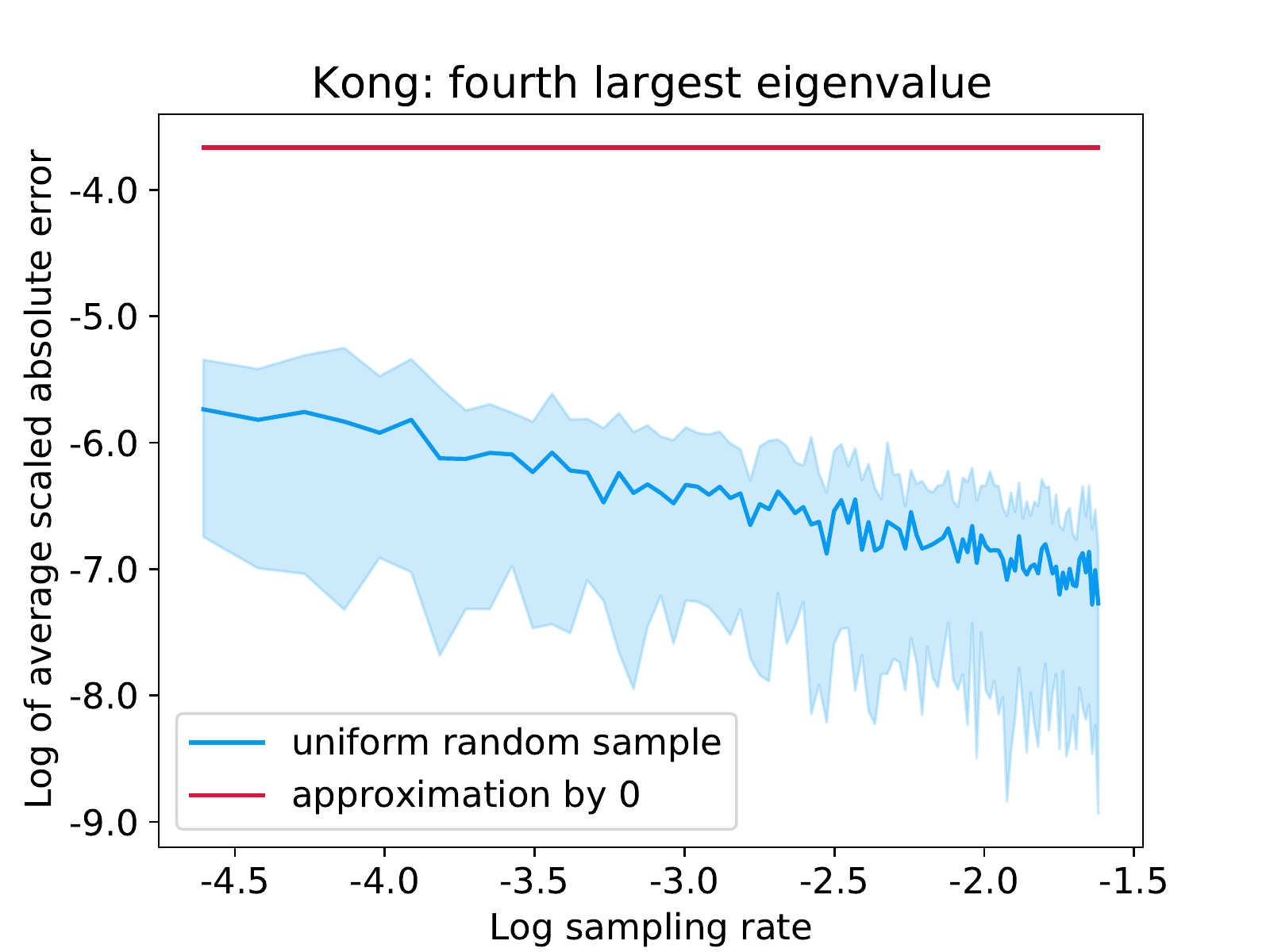}
    \caption{Thin plane spline similarity matrix.}
    \end{subfigure}
    \vskip\baselineskip
    \vspace{-1.3em}
    \begin{subfigure}[t]{\textwidth}
    \centering
    \includegraphics[width=0.32\textwidth]{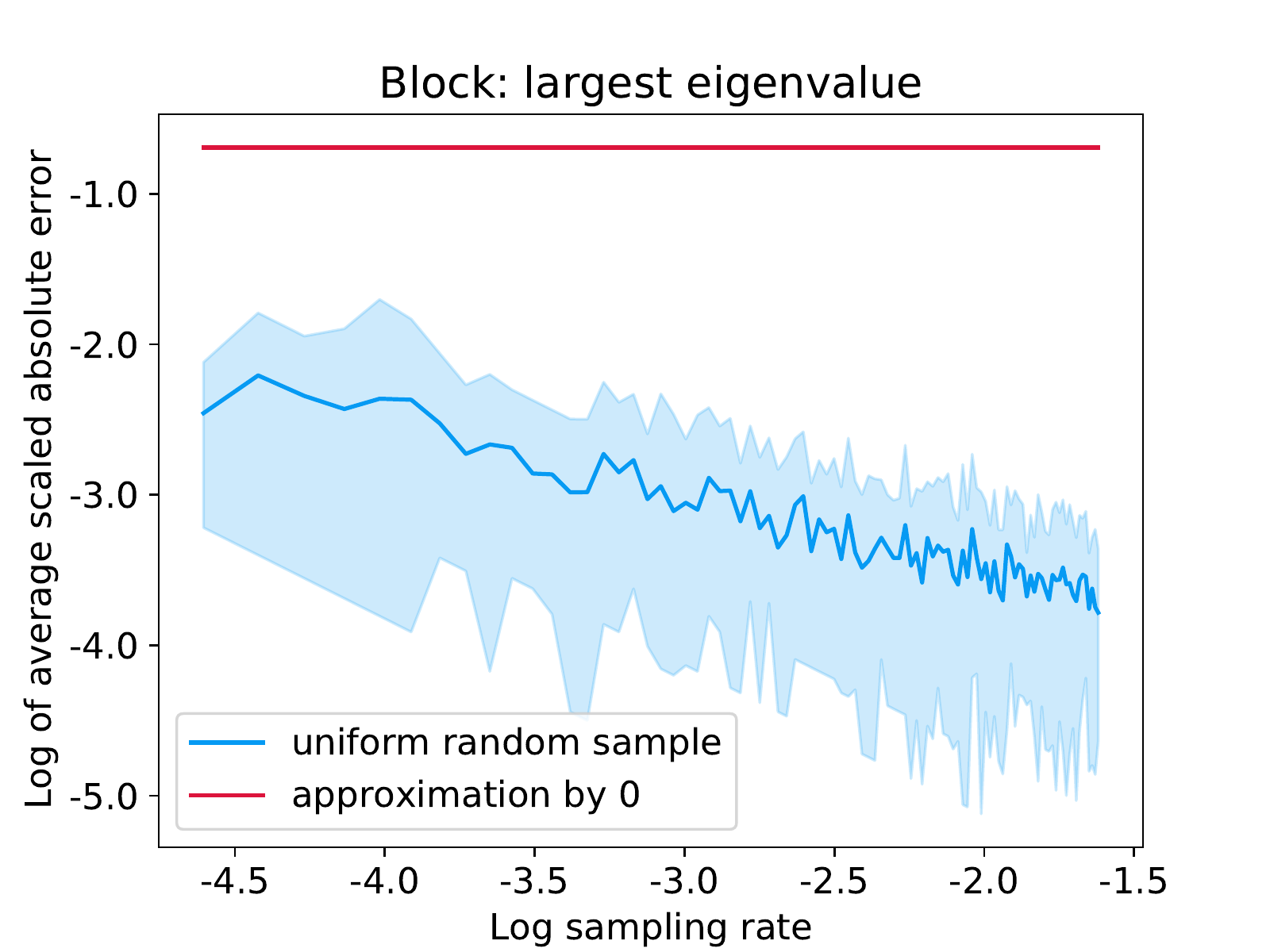}
    \includegraphics[width=0.32\textwidth]{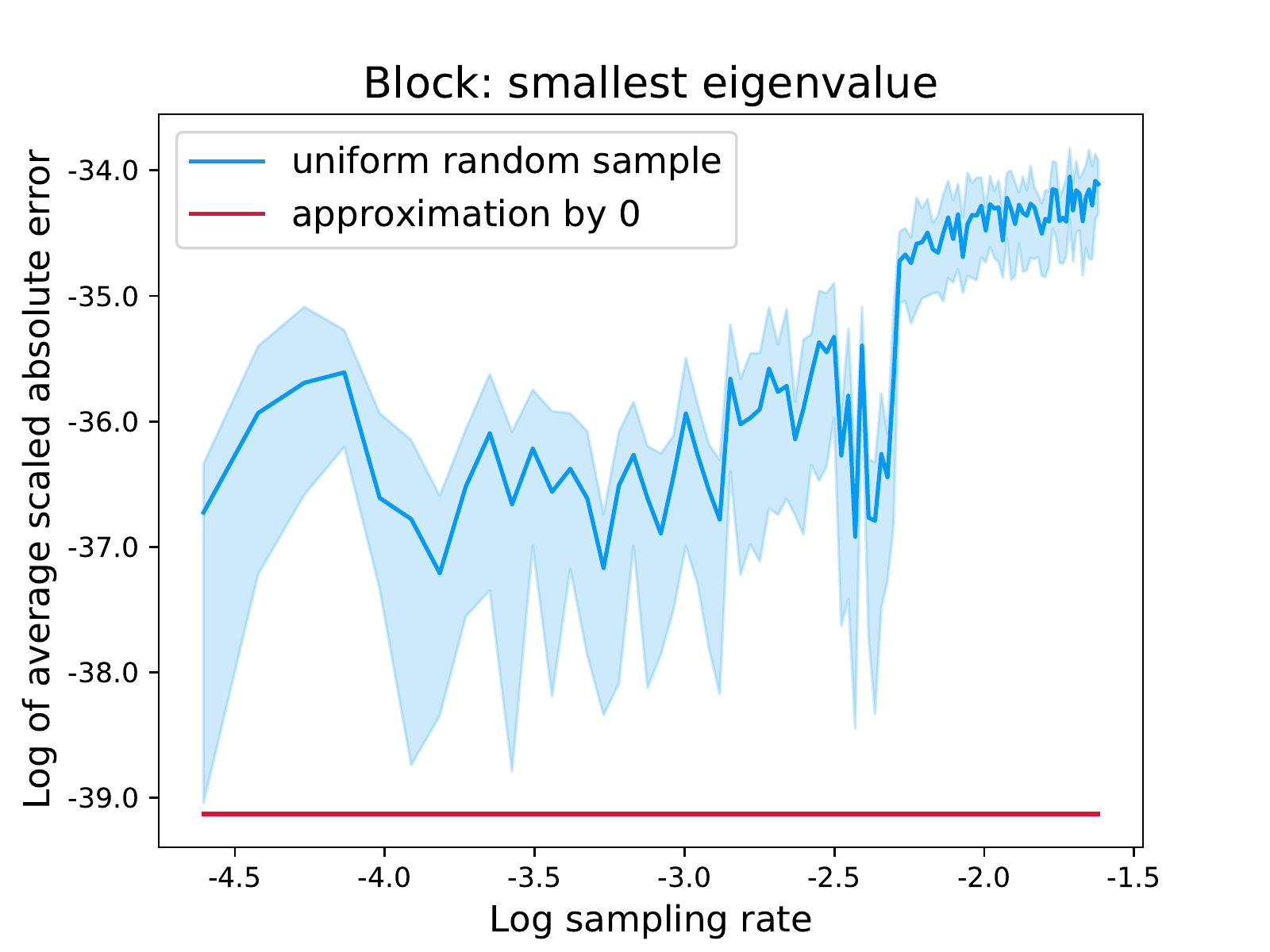}
    \includegraphics[width=0.32\textwidth]{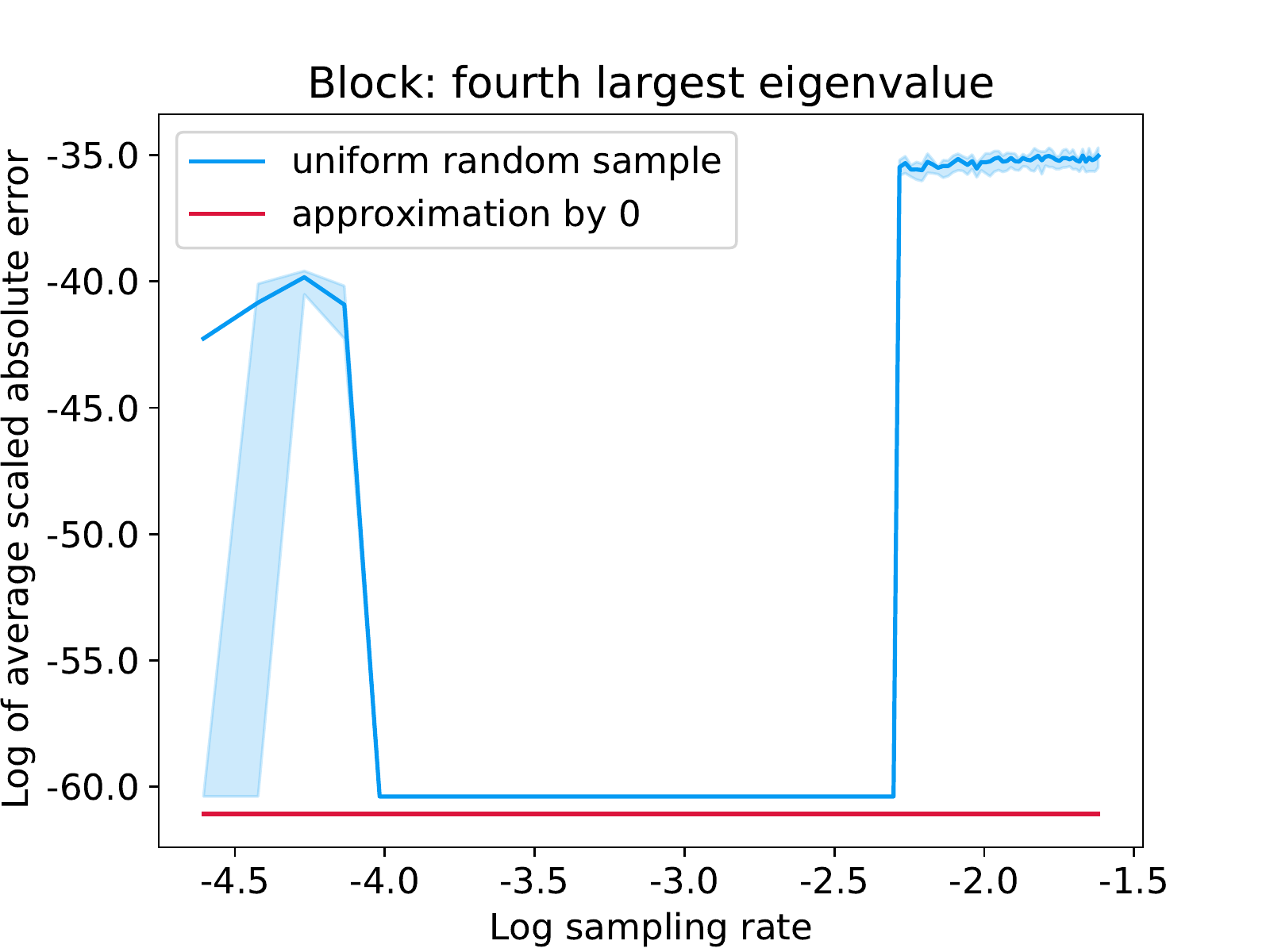}
    \caption{Block matrix.}
    \end{subfigure}
    \vspace{-0.5em}
    \caption{{\textbf{Approximation error of eigenvalues of dense matrices}. Log scale absolute error vs. log sampling rate for Algorithm \ref{alg:eigenvalue estimate} and and approximation by 0, as described in Section \ref{subsec:experiments}, for approximating the largest, smallest and fourth largest of three of the example matrices. 
    The corresponding true eigenvalues for each matrix in-order are: (hyperbolic tangent) $\{4.52\mathrm{e}{+03}, -7.85\mathrm{e}{+00}, 3.18\mathrm{e}{-01}\}$, (thin plane spline) $\{3.54\mathrm{e}{+02}, -1.22\mathrm{e}{+03},  1.28\mathrm{e}{+02}\}$ and (block matrix) $\{ 2.50\mathrm{e}{+03}, -5.08\mathrm{e}{-14},1.49\mathrm{e}{-23}\}$.} }
\label{fig:errors1}
\end{figure}

\begin{figure}[H]
    \centering
    \begin{subfigure}[t]{\textwidth}
    \centering
    \includegraphics[width=0.32\textwidth]{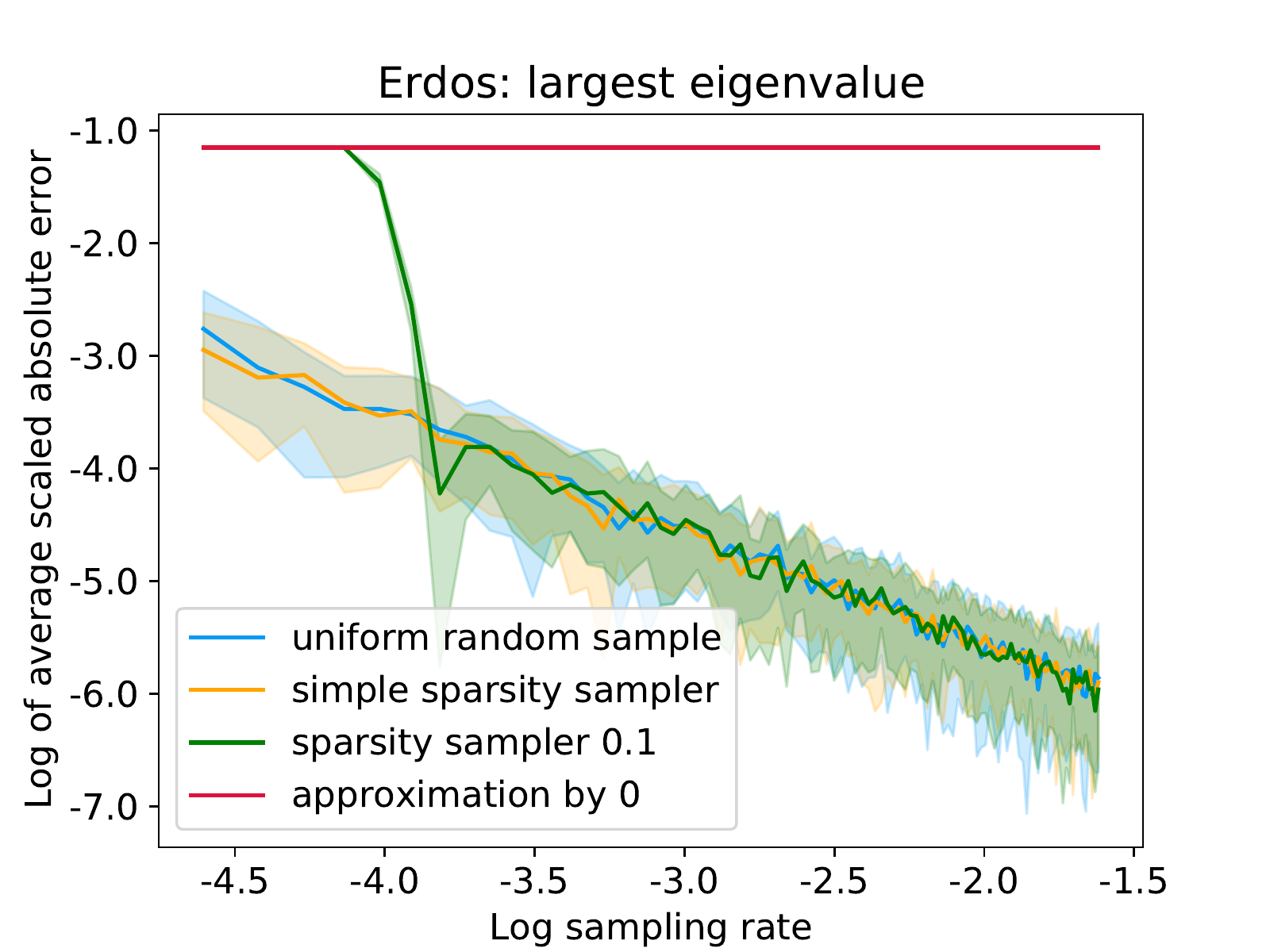}
    \includegraphics[width=0.32\textwidth]{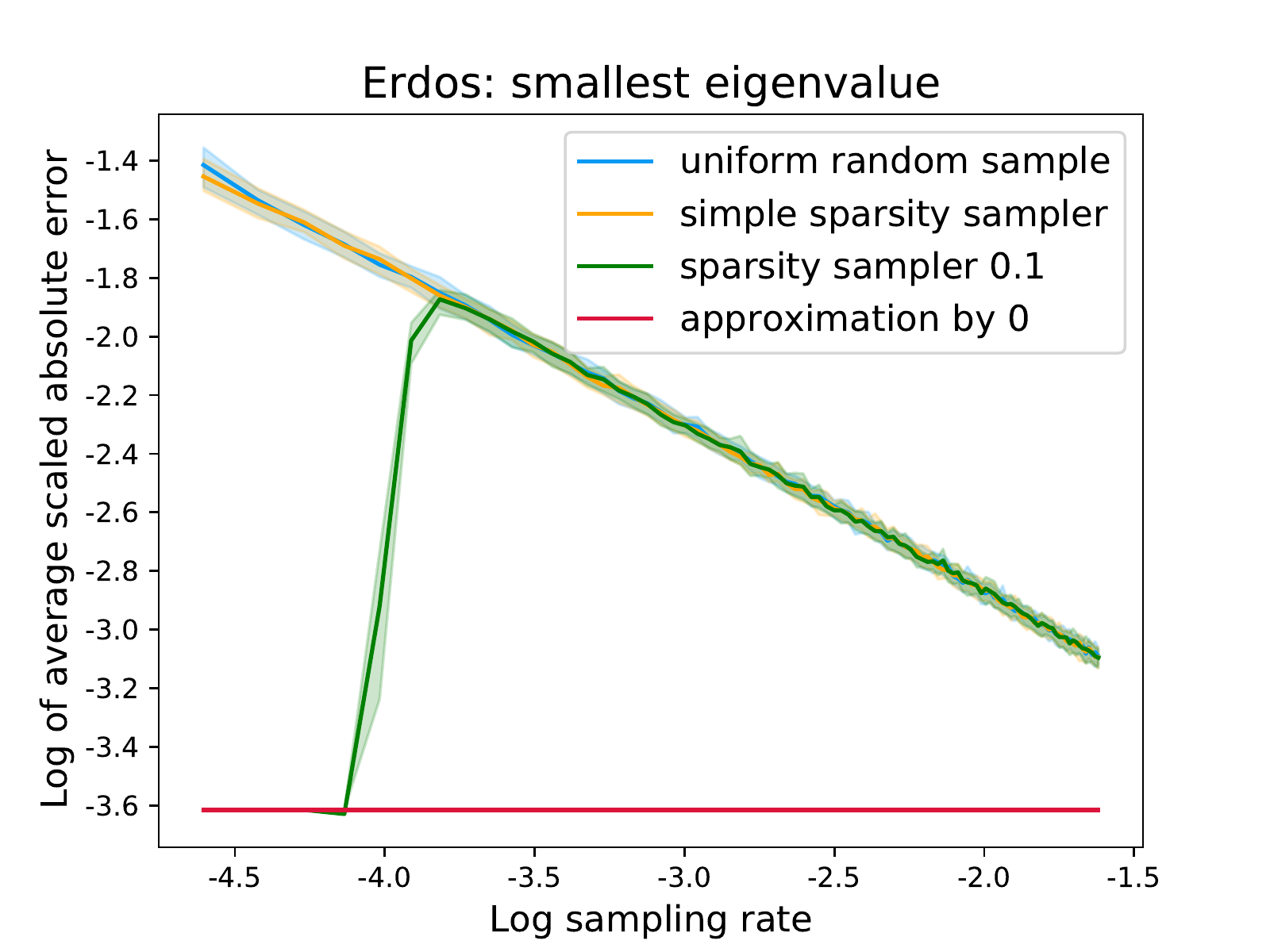}
    \includegraphics[width=0.32\textwidth]{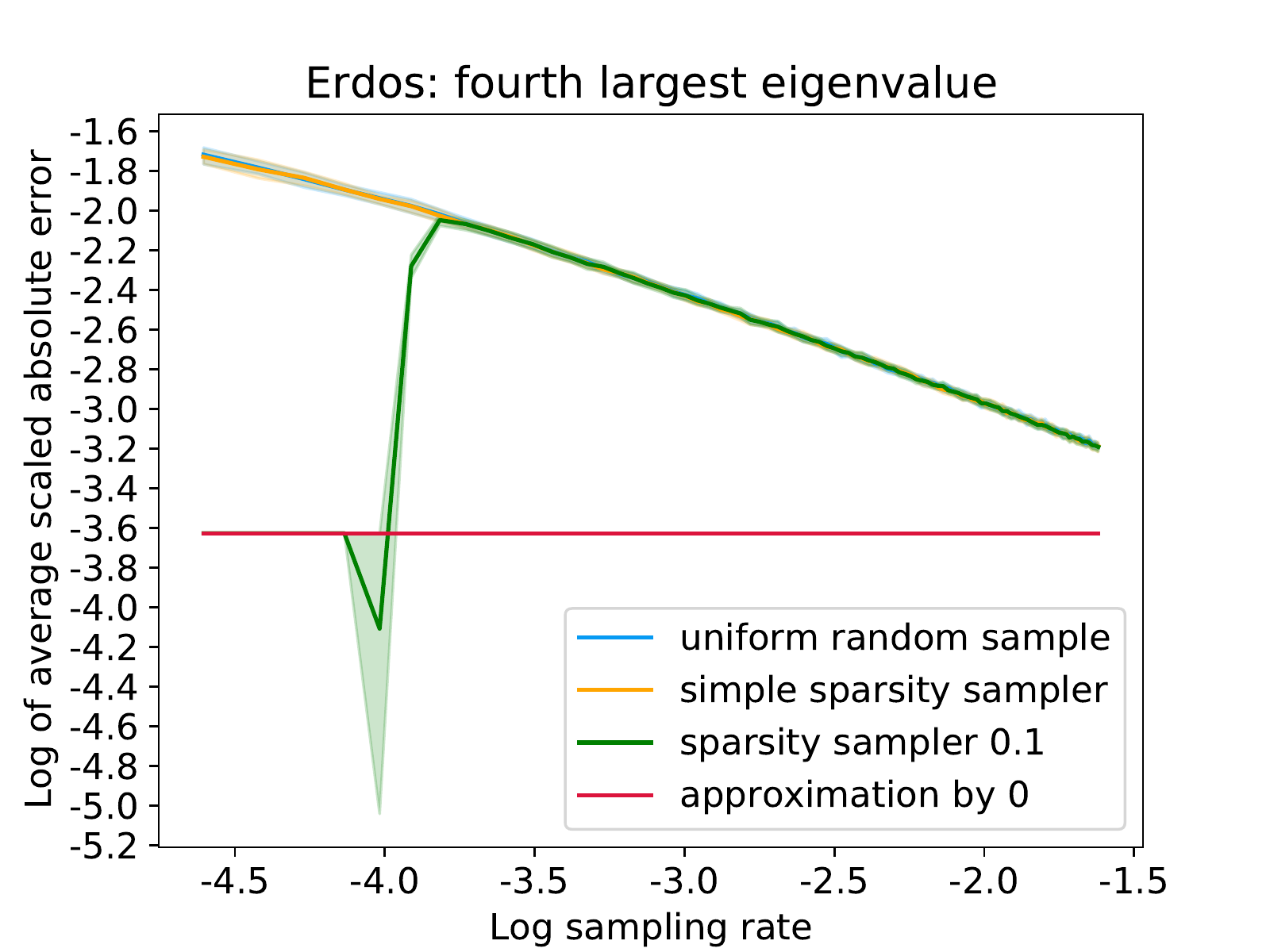}
    \caption{\erdos\ graph adjacency matrix \cite{erdos59a}.}
    \end{subfigure}
    \vskip\baselineskip
    \vspace{-1.3em}
    \begin{subfigure}[t]{\textwidth}
    \centering
    \includegraphics[width=0.32\textwidth]{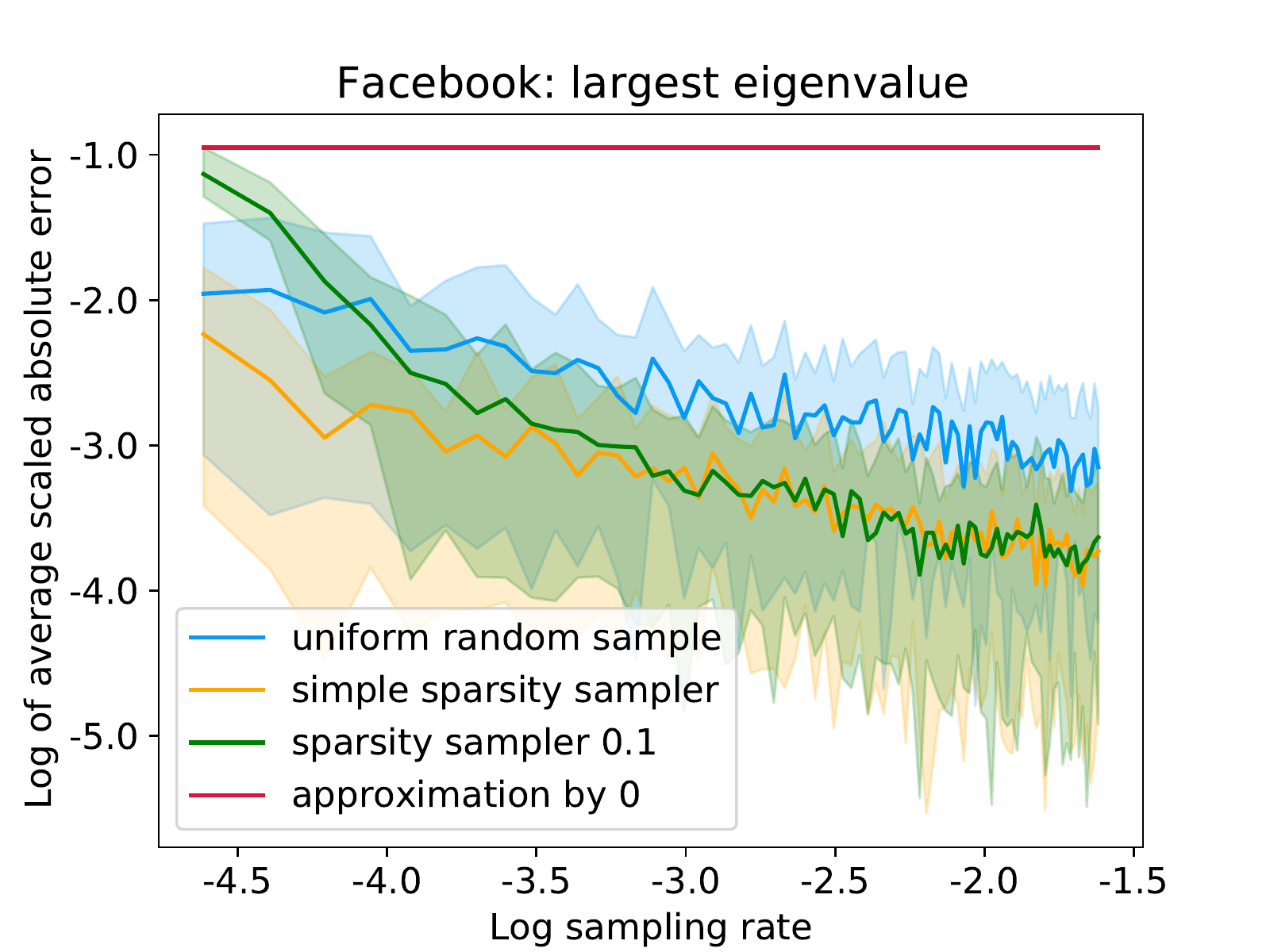}
    \includegraphics[width=0.32\textwidth]{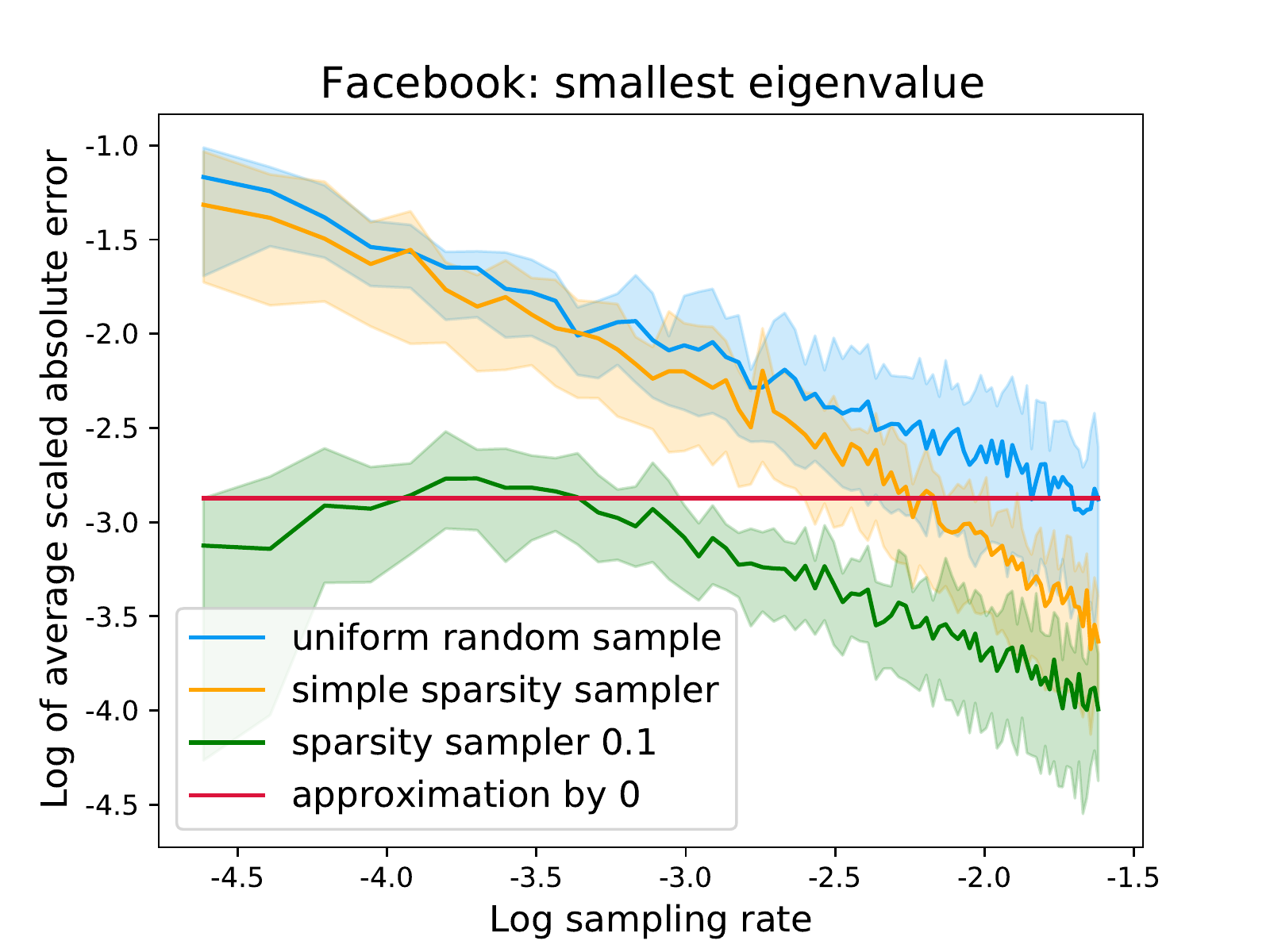}
    \includegraphics[width=0.32\textwidth]{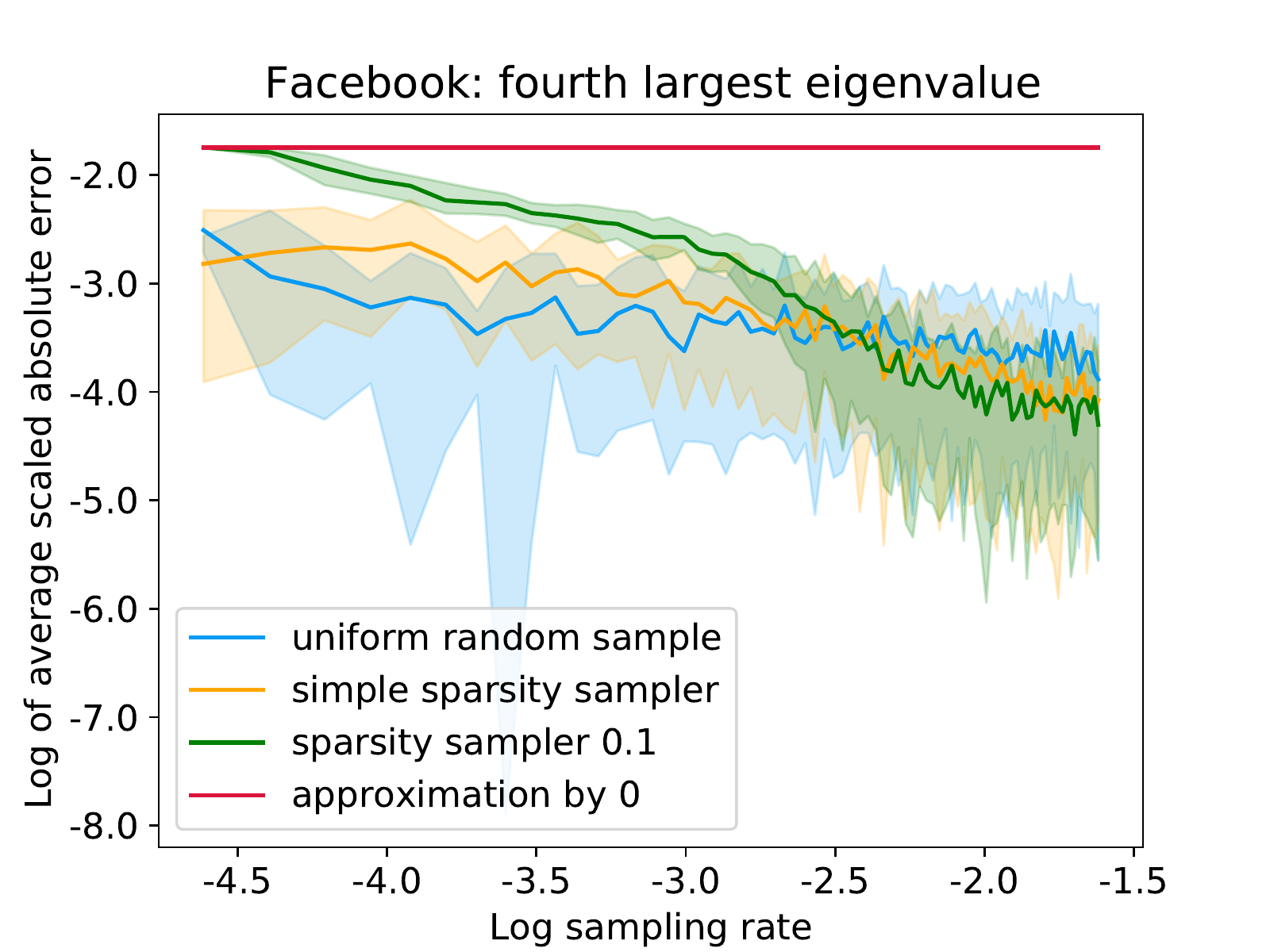}
    \caption{Facebook graph adjacency matrix \cite{mcauley2012learning}.}
    \end{subfigure}
    \vskip\baselineskip
    \vspace{-1.3em}
    \begin{subfigure}[t]{\textwidth}
    \centering
    \includegraphics[width=0.32\textwidth]{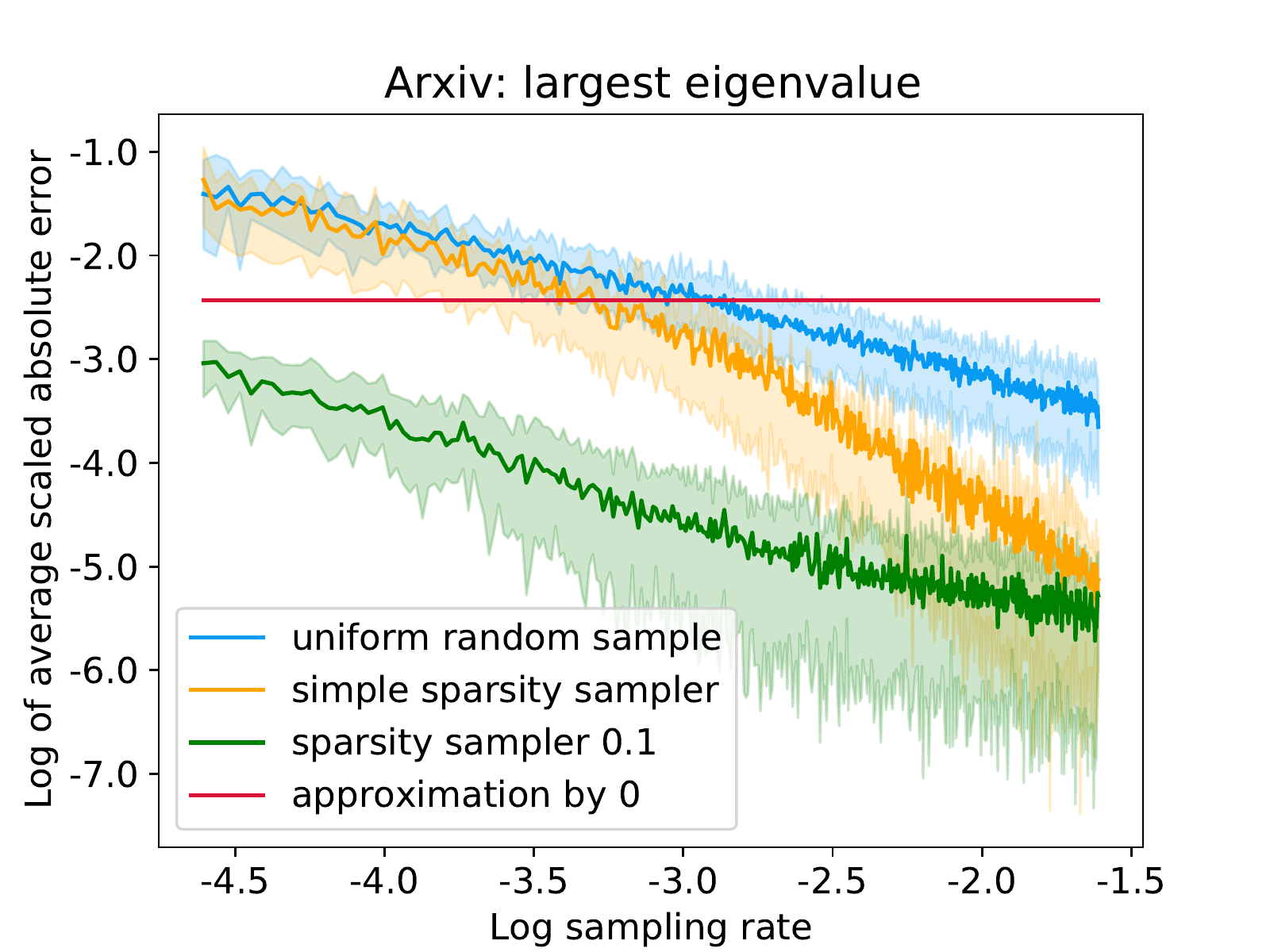}
    \includegraphics[width=0.32\textwidth]{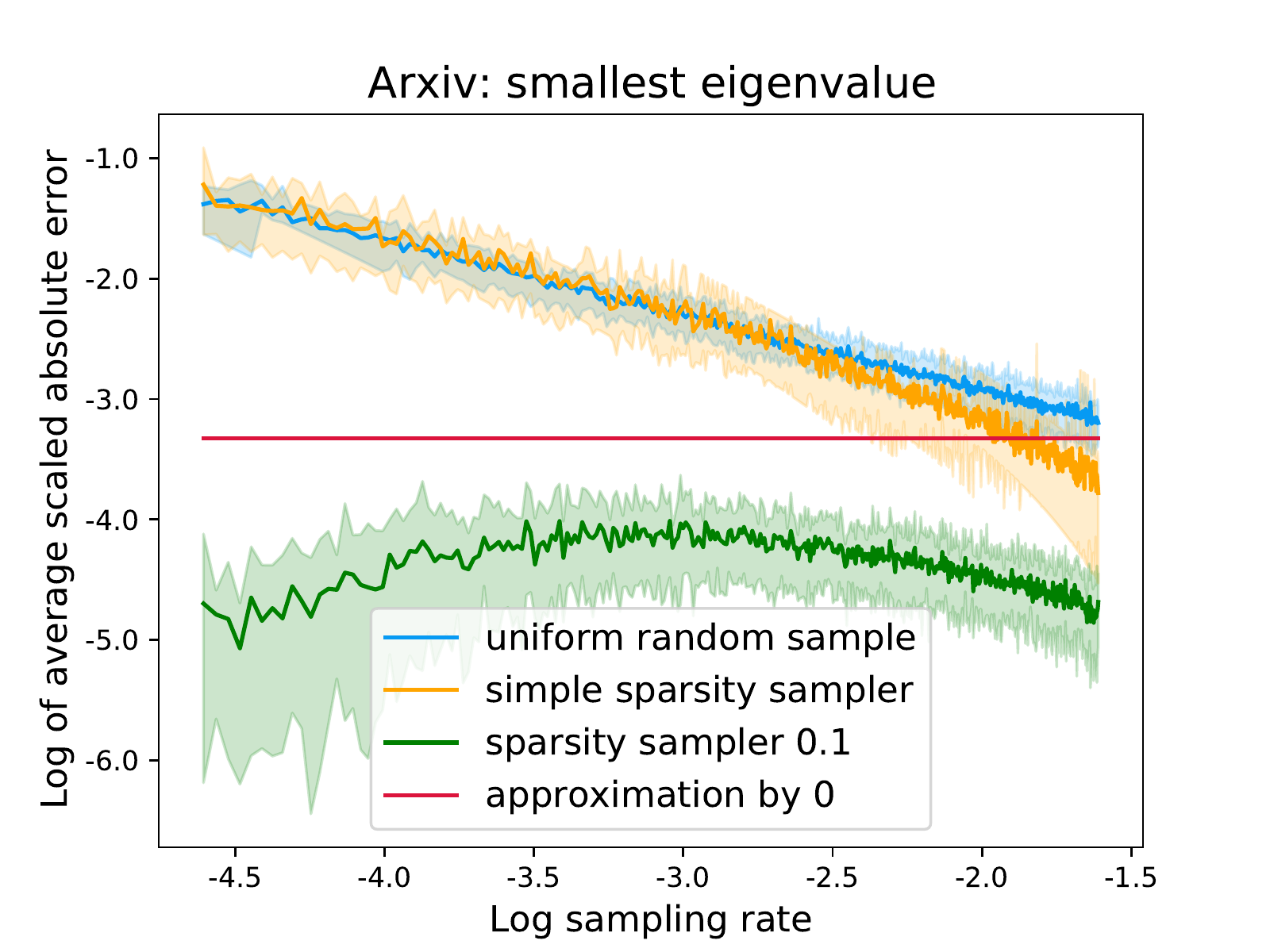}
    \includegraphics[width=0.32\textwidth]{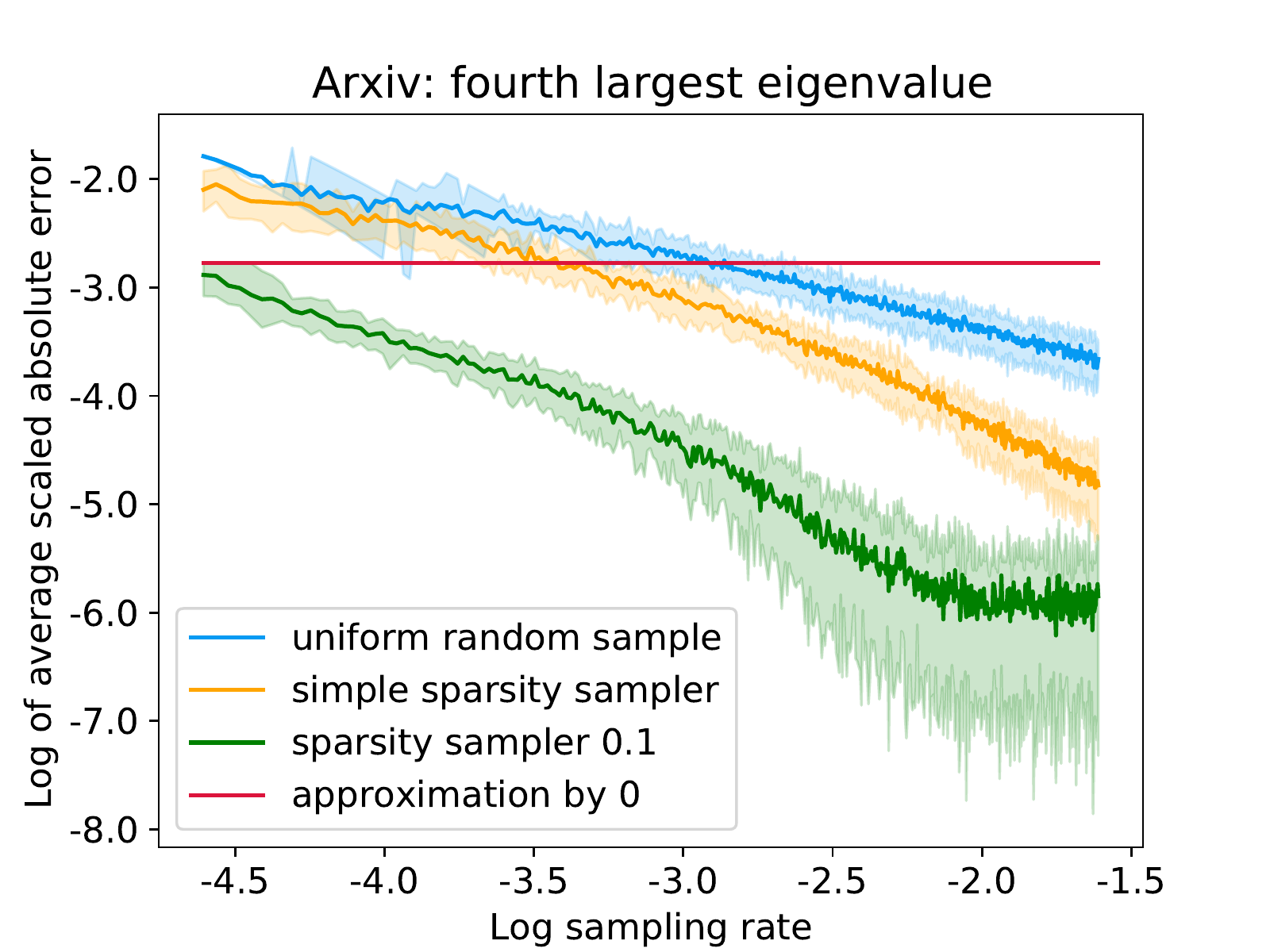}
    \caption{ArXiv collaboration network adjacency matrix \cite{leskovec2007graph}.}
    \end{subfigure}
    \vspace{-0.5em}
    \caption{\textbf{Approximation error of eigenvalues of sparse matrices}. Log scale absolute error vs. log sampling rate for Algorithm \ref{alg:eigenvalue estimate}, Algoithm \ref{alg:nnz eigenvalue estimate}, simple sparsity sampler and approximation by 0, as described in Section \ref{subsec:experiments}, for approximating the largest, smallest, and fourth largest of remaining three example matrices. 
    The corresponding true eigenvalues for each matrix in-order are: 
    (\erdos) $\{500.57, -42.52, 42.02\}$, (Facebook) $\{162.37, -23.75, 73.28\}$ and (arXiv) $\{37.95, -15.58, 26.92\}$.}
    \label{fig:errors2}
\end{figure}

\clearpage

\bibliography{references}

\pagebreak

\appendix


\section{Eigenvalue Approximation for PSD Matrices}\label{app:psd}

Here we give a simple proof that shows if Algorithm \ref{alg:eigenvalue estimate} is used to approximate the eigenvalues of positive semidefinite (PSD) matrices (i.e., with all non-negative eigenvalues) using a $O(1/\epsilon^2) \times O(1/\epsilon^2)$ random submatrix, then the $\ell_2$ norm of the error of eigenvalue approximations is bounded by $\epsilon n$. This much stronger result immediately implies that each eigenvalue of a PSD matrix can be approximated to $\pm \epsilon n$ additive error using just a $O(1/\epsilon^2) \times O(1/\epsilon^2)$ random submatrix. The proof follows from a bound in~\cite{bhatia2013matrix} which bounds the $\ell_2$ norm of the difference vector of eigenvalues of a Hermitian matrix and any other random matrix by the Frobenius norm of the difference of the two matrices. This improves on the bound of Theorem \ref{thm:main_bound} for general entrywise bounded matrices by a $1/\epsilon^2$ factor, and matches the $O(1/\epsilon^4)$ lower bound for principal submatrix queries in~\cite{BakshiChepurkoJayaram:2020}. Note that the hard instance used to prove the lower bound in~\cite{BakshiChepurkoJayaram:2020} can in fact be negated to be PSD, thus showing that our upper bound here is tight.

We first state the result from~\cite{bhatia2013matrix} which we will be using in our proof.

\begin{fact}[$\ell_2$-norm bound on eigenvalues \cite{bhatia2013matrix}]\label{fact:spectral error bound} 
Let $\bv A \in \C^{n\times n}$ be Hermitian and $\bv B \in \C^{n\times n}$ be any matrix whose eigenvalues are $\lambda_1(\bv B), \ldots, \lambda_n(\bv B)$ such that $Re(\lambda_1(\bv B))\geq \ldots \ge Re(\lambda_n(\bv B))$ (where $Re(\lambda_i(\bv B))$ denotes the real part of $\lambda_i(\bv B)$). Let $\bv A- \bv B = \bv E$. Then
\begin{align*}
    \left(\sum_{i=1}^n \left|\lambda_i(\bv A) - \lambda_i(\bv B)\right|^2\right)^{1/2} \leq \sqrt{2}\|\bv E\|_F.
\end{align*}
\end{fact}

Our result is based on the following Lemma, we prove at the end of the section.
\begin{lemma}\label{lem:amm}
Consider a PSD matrix $\bv A = \bv {BB}^T$ with $\norm{\bv A}_\infty \le 1$.  Let $S$ be sampled as in Algorithm \ref{alg:eigenvalue estimate} for $s \geq \frac{1}{\epsilon^2\delta} $. Let $\bv{\bar S} \in \R^{n \times |S|}$ be the scaled sampling matrix satisfying $\bv{\bar S}^T \bv{A} \bv{\bar S} = \frac{n}{s} \cdot \bv{A}_{S}$.
Then with probability at least $1-\delta$,
\begin{align*}
\norm{\bv B^T \bv {\bar S} \bv{\bar S}^T \bv B - \bv B^T \bv B}_F \le \epsilon n.
\end{align*}
\end{lemma}
\noindent From the above Lemma we have:
\begin{restatable}[Spectral norm bound -- PSD matrices]{corollary}{psdBound}
\label{cor:psd bound}
Consider a PSD matrix $\bv A$ with $\norm{\bv A}_\infty \le 1$.  
Let $S$ be a subset of indices formed by including each index in $[n]$ independently with probability $s/n$ as in Algorithm \ref{alg:eigenvalue estimate}. Let $\bv A_S$ be the corresponding principal submatrix of $\bv{A}$, with eigenvalues $\lambda_1(\bv{A}_S) \ge \ldots \ge \lambda_{|S|}(\bv{A}_S)$.

For all $i \in [|S|]$ with $\lambda_i(\bv{A}_S) \ge 0$, let $\tilde \lambda_i(\bv{A}) = \frac{n}{s} \cdot \lambda_i(\bv{A}_S)$. For all other $i \in [n]$, let $\tilde \lambda_i(\bv{A}) = 0$.
Then if $s \ge \frac{2}{\epsilon^2\delta}$, with probability at least $1-\delta$,
\begin{align*}
\left(\sum_{i=1}^n\left| \tilde \lambda_i(\bv A) - \lambda_i(\bv A) \right | ^2\right)^{1/2} &\le \epsilon n,
\end{align*}
which implies that for all $i \in [n]$,
\begin{align*}
\lambda_i(\bv A) - \epsilon n \le \tilde \lambda_i(\bv A) \le \lambda_i(\bv A) + \epsilon n.
\end{align*}
\end{restatable}
\begin{proof}
Let $S$ be sampled as in Algorithm \ref{alg:eigenvalue estimate} and let $\bv{\bar S} \in \R^{n \times |S|}$ be the scaled sampling matrix satisfying $\bv{\bar S}^T \bv{A} \bv{\bar S} = \frac{n}{s} \cdot \bv{A}_{S}$.
Since $\bv A$ is PSD, we can write $\bv A= \bv B \bv B^T$ for some matrix $\bv B \in \R^{n \times \rank(\bv{A})}$. From Lemma~\ref{lem:amm}, for $s \geq \frac{1}{\epsilon^2\delta}$, we have with probability at least $1-\delta$:
\begin{align*}
 \norm{\bv B^T  \bv{\bar S} \bv{\bar S}^T \bv B - \bv B^T \bv B}_F \le \epsilon n
\end{align*}
Using Fact~\ref{fact:spectral error bound}, we have,
\begin{align}
      \left(\sum_{i=1}^{\rank(\bv A)} \left|\lambda_i(\bv B^T  \bv{\bar S}  \bv{\bar S}^T \bv B) - \lambda_i(\bv B^T \bv B)\right|^2\right)^{1/2}
    \le \sqrt{2}\norm{\bv B^T \bv{\bar S}\bv{\bar S}^T \bv B - \bv B^T \bv B}_F  \le \sqrt{2}\epsilon n.\label{eq:spectral error bound psd}
\end{align}
Also from Fact~\ref{def: equality of eigenvalues}, we have $\lambda_i(\bv B^T \bv B)=\lambda_i(\bv{B}\bv{B}^T)=\lambda_i(\bv{A})$ for all $i \le \rank(\bv{A})$. Thus,
\begin{align*}
     \left(\sum_{i=1}^{\rank(\bv A)} \left|\lambda_i(\bv B^T  \bv{\bar S}  \bv{\bar S}^T \bv B) - \lambda_i(\bv{A})\right|^2\right)^{1/2} \leq \sqrt{2}\epsilon n
\end{align*}
Also by Fact~\ref{def: equality of eigenvalues}, all non-zero eigenvalues of $\bv B^T  \bv{\bar S}  \bv{\bar S}^T \bv B$ are equal to those of $\bv{\bar S}^T \bv{BB}^T  \bv{\bar S}=\frac{n}{s} \cdot \bv{A}_S$. All other eigenvalue estimates are set to $0$. Further, for all $i > \rank(\bv A)$, $\lambda_i(\bv A) = 0$. Thus, 
\begin{align*}
     \left(\sum_{i=1}^{n} \left|\tilde \lambda_i(\bv A) - \lambda_i(\bv{A})\right|^2\right)^{1/2} \leq \sqrt{2}\epsilon n.
\end{align*}
Adjusting $\epsilon$ to $\epsilon/\sqrt{2}$ then  gives us the bound.
\end{proof}

We now prove Lemma \ref{lem:amm}, using a standard  approach for sampling based approximate matrix multiplication -- see e.g. \cite{drineas2001fast}.
\begin{proof}[Proof of Lemma \ref{lem:amm}]
For $k = 1,\ldots,n$ let $\bv Y_k = \frac{n}{s} - 1$ with probability $\frac{s}{n}$ and $\bv Y_k = -1$ with probability $1 - \frac{s}{n}$. Thus $\E[\bv Y_k] = 0$ and
\begin{align*}
\norm{\bv B^T  \bv{\bar S}  \bv{\bar S}^T \bv B - \bv B^T \bv B}_F^2 = \sum_{i = 1}^n \sum_{j = 1}^n \left (\sum_{k=1}^n \bv Y_k \cdot \bv B_{ik} \bv B_{jk}\right )^2.
\end{align*}
Fixing $i,j$, each the $\bv Y_k \cdot \bv B_{ik} \bv B_{jk}$ are $0$ mean independent random variables. Thus we have:
\begin{align*}
\E \left [ \norm{\bv B^T  \bv{\bar S}  \bv{\bar S}^T \bv B - \bv B^T \bv B}_F^2 \right ] &= \sum_{i = 1}^n \sum_{j = 1}^n \E \left [ \left (\sum_{k=1}^n \bv Y_k \cdot \bv B_{ik} \bv B_{jk} \right )^2 \right ]\\
&= \sum_{i = 1}^n \sum_{j = 1}^n \Var \left [\sum_{k=1}^n \bv Y_k \cdot \bv B_{ik} \bv B_{jk} \right ]\\
&= \sum_{i = 1}^n \sum_{j = 1}^n \sum_{k=1}^n \Var \left [\bv Y_k \cdot \bv B_{ik} \bv B_{jk} \right ]\\
&\le \sum_{i = 1}^n \sum_{j = 1}^n \sum_{k=1}^n \frac{n}{s} \cdot \bv B_{ik}^2 \bv B_{jk}^2.
\end{align*} 
since $\Var[\bv Y_k] = \left (\frac{n}{s} -1 \right )^2 \cdot \frac{s}{n} + \left (1-\frac{s}{n} \right ) = \frac{n}{s} - 2 + \frac{s}{n} + 1 - \frac{s}{n} = \frac{n}{s}-1$.
Rearranging the sums we have:
\begin{align*}
\E[\norm{\bv B^T  \bv{\bar S}  \bv{\bar S}^T \bv B - \bv B^T \bv B}_F^2] \le \frac{n}{s} \sum_{k=1}^n \sum_{i=1}^n \bv B_{ik}^2 \sum_{j=1}^n \bv B_{jk}^2.
\end{align*}
Observe that  $\sum_{j=1}^n \bv B_{jk}^2 = \bv A_{kk} \le \norm{\bv A}_\infty \le 1$, thus overall we have:
\begin{align*}
\E[\norm{\bv B^T  \bv{\bar S}  \bv{\bar S}^T \bv B - \bv B^T \bv B}_F^2] \le \frac{n^2}{s} \le \epsilon^2 \delta n^2.
\end{align*}
So by Markov's inequality, with probability $\ge 1-\delta$, $\norm{\bv B^T  \bv{\bar S}  \bv{\bar S}^T \bv B - \bv B^T \bv B}_F^2 \le \epsilon^2 n^2$. This completes the theorem after taking a square root.
\end{proof}
\noindent \textbf{Remark}: The proof of Lemma~\ref{lem:amm} can be easily modified to show that the $i$\textsuperscript{th} row of $\bv{A}$ can be sampled with probability proportional to $\frac{|\bv{A}_{ii}|}{\text{tr}(\bv{A})}$ to approximate the eigenvalues of any PSD $\bv{A}$ up to $\pm \epsilon \cdot \text{tr}(\bv{A})$ error ($\text{tr}(\bv{A})$ is the trace of $\bv{A}$). When sampling with probability proportional to $\frac{|\bv{A}_{ii}|}{\text{tr}(\bv{A})}$, we do not require a bounded entry assumption on $\bv{A}$. 

\section{Alternate Bound for Uniform Sampling}\label{app:alt_bound}

In this section we provide an alternate bound for approximating eigenvalues with uniform sampling. The sample complexity is worse by a factor of $1/\epsilon$ for this approach, but better by a factor $\log^2 n$ as compared to Theorem \ref{thm:main_bound}.
We start with an analog to Lemma \ref{lemma: orthonormality},  showing that the outlying eigenspace remains nearly orthogonal after sampling. In particular, we show concentration of the Hermitian matrix $\bv{V}_o^T \bv{\bar S} \bv{\bar S}^T \bv{V}_o$ about its expectation $\bv{V}_o^T \bv{V}_o = \bv I$ rather than the non-Hermitian $\bs{\Lambda}_o^{1/2}\bv{V}_o^T \bv{\bar S} \bv{\bar S}^T\bv{V}_o \bs{\Lambda}_o^{1/2}$ as in Lemma \ref{lemma: orthonormality}. This allows us to use Weyl's inequality in our final analysis, rather than the non-Hermitian eigenvalue perturbation bound of Fact \ref{fact:weyl_general}, saving a $\log^2 n$ factor in the sample complexity.

\begin{lemma}[Near orthonormality -- sampled outlying eigenvalues]\label{lem:uniform-eps4-near-orthonorm}
Let $S$ be sampled as in Algorithm \ref{alg:eigenvalue estimate} for $s \geq \frac{c\log(1/(\epsilon \delta))}{\epsilon^4\delta} $ where $c$ is a sufficiently large constant. Let $\bv{\bar S} \in \R^{n \times |S|}$ be the scaled sampling matrix satisfying $\bv{\bar S}^T \bv{A} \bv{\bar S} = \frac{n}{s} \cdot \bv{A}_{S}$. Then with probability at least $1-\delta$, $\|\bv V_o^T \bar{\bv S} \bar{\bv S}^T \bv V_o - \bv{I} \|_2 \leq \epsilon.$
\end{lemma} 
\begin{proof}
The result is standard in randomized numerical linear algebra -- see e.g., \cite{cohen2015uniform}. For completeness, we give a proof here. Define $\bv{E} = \bv V_o^T \bar{\bv S} \bar{\bv S}^T \bv V_o  - \bv I$.
For all $i \in [n]$, let $\bv V_{o,i}$ be the $i^{th}$ row of $\bv{V}_o$ and define the matrix valued random variable
\begin{align*}
    \bv Y_i = 
    \begin{cases}
    \frac{n}{s}\bv V_{o,i} \bv V_{o,i}^T , & \text{with probability } s/n\\
    0 & \text{otherwise.}
    \end{cases}
\end{align*}
Then, similar to the proof of Lemma~\ref{lemma: orthonormality}, define $\bv Q_i = \bv Y_i - \mathbb{E}\left[\bv Y_i\right]$. Since $\bv Q_1, \bv Q_2, \ldots, \bv Q_n$ are independent random variables and $ \sum_{i=1}^n \bv Q_i=\bv V_o^T \bar{\bv S} \bar{\bv S}^T \bv V_o  - \bv I = \bv{E}$ , we need to bound $\|\bv{Q}_i \|_2$ for all $i \in [n]$ and  $\bv{Var}(\bv E) \eqdef\E(\bv{EE}^T)= \E(\bv{E}^T\bv{E})  =\sum_{i=1}^n \mathbb{E}[\bv Q_i^2]$. Observe $\|\bv Q_i\|_2 \leq \max\left(1, \frac{n}{s}-1\right)\|\bv V_{o,i}\bv V_{o,i}^T\|_2 = \max\left(1, \frac{n}{s}-1\right)\|\bv V_{o,i}\|_2^2 \leq \frac{1}{\epsilon^2\delta s}$, by row norm bounds of Lemma \ref{lemma:row_norm}. Again, using Lemma \ref{lemma:row_norm} we have
\begin{align*}
  \sum_{i=1}^n \mathbb{E}[\bv Q_i^2]  &= \sum_{i=1}^n \frac{s}{n}\cdot \left( \frac{n}{s}-1\right)^2(\bv{V}_{o,i} \bv{V}^T_{o,i})^2+\left(1-\frac{s}{n} \right)(\bv{V}_{o,i} \bv{V}^T_{o,i})^2 \\
  &\preceq \sum_{i=1}^n \frac{n}{s}\|\bv{V}_{o,i} \|_2^2(\bv{V}_{o,i} \bv{V}^T_{o,i}) \\
  &\preceq \sum_{i=1}^n \frac{n}{s}\frac{1}{\epsilon^2 \delta n} (\bv{V}_{o,i} \bv{V}^T_{o,i}) \\
  &\preceq \frac{1}{s\epsilon^2 \delta} \cdot \bv{I}
\end{align*}
where $\bv I$ is the identity matrix of appropriate dimension. By setting $d = \frac{1}{\epsilon^2\delta}$, we can finally bound the probability of the event $\|\bv{E} \|_2 \geq \epsilon n$ using Theorem \ref{thm:matrix bernstein} (the matrix Bernstein inequality) with $\delta$ if $s \geq \frac{c\log(1/(\epsilon\delta))}{\epsilon^4\delta}$. Since these steps follow  Lemma~\ref{lemma: orthonormality} nearly exactly, we omit them here.
\end{proof}

With Lemma \ref{lem:uniform-eps4-near-orthonorm} in place, we can now give our alternate sample complexity bound.

\begin{restatable}[Sublinear Time Eigenvalue Approximation]{theorem}{largeeigbound}
\label{thm:bound on large values}
Let $\bv A \in \mathbb{R}^{n\times n}$ be symmetric with $\|\bv A\|_\infty \leq 1$ and eigenvalues $\lambda_1(\bv{A}) \ge \ldots \ge \lambda_n(\bv{A})$. Let $S \subseteq [n]$ be formed by including each index independently with probability $s/n$ as in Algorithm \ref{alg:eigenvalue estimate}. Let $\bv A_S$ be the corresponding principal submatrix of $\bv{A}$, with eigenvalues $\lambda_1(\bv{A}_S) \ge \ldots \ge \lambda_{|S|}(\bv{A}_S)$.

For all $i \in [|S|]$ with $\lambda_i(\bv{A}_S) \ge 0$, let $\tilde \lambda_i(\bv{A}) = \frac{n}{s} \cdot \lambda_i(\bv{A}_S)$. For all $i \in \{1, \ldots, |S|\}$ with $\lambda_i(\bv{A}_S) < 0$, let $\tilde \lambda_{n-(|S|-i)}(\bv{A}) = \frac{n}{s} \cdot \lambda_i(\bv{A}_S)$. For all other $i \in [n]$, let $\tilde \lambda_i(\bv{A}) = 0$.
If $s \geq \frac{c \log n}{\epsilon^4\delta}$, for a large enough constant $c$,
 then with probability $\ge 1-\delta$, for all $i \in [n]$,
\begin{align*}
    \lambda_i(\bv A) -\epsilon n \leq \tilde \lambda_i(\bv A) \leq \lambda_i(\bv A) +\epsilon n.
\end{align*}
\end{restatable}

\begin{proof}
Let $\bv{S} \in \R^{n \times |S|}$ be the binary sampling matrix with a single one in each column such that $\bv{S}^T \bv{A} \bv{S} = \bv{A}_S$. Let $\bar{\bv S} = \sqrt{n/s} \cdot \bv{S}$.
Following Definition \ref{def:split}, we write $\bv A = \bv A_o + \bv A_m$.
By Fact \ref{def: equality of eigenvalues} we have that the nonzero eigenvalues of $\frac{n}{s} \cdot \bv{ A}_{o,S} = \bar{\bv S}^T \bv V_o \bv \Lambda_o \bv V_o^T \bar{\bv S}$ are identical to those of $\bv{\Lambda}_o \bv V_o^T \bar{\bv S} \bar{\bv S}^T \bv V_o$.

Note that $\bv H = \bv V_o^T \bar{\bv S} \bar{\bv S}^T \bv V_o$ is positive semidefinite. Writing its eigendecomposition $\bv H = \bv U \bv W \bv U^T$ we can define the matrix squareroot $\bv H^{1/2} = \bv U \bv W^{1/2} \bv U^T$ with $\bv{H}^{1/2}\bv{H}^{1/2} = \bv{H}$. By Lemma \ref{lem:uniform-eps4-near-orthonorm} applied with error $\epsilon/6$, with probability at least $1-\delta$, all eigenvalues of $\bv{H}$ lie in the range $[1-\epsilon/6, 1+\epsilon/6]$. In turn, all eigenvalues of $\bv{H}^{1/2}$ also lie in this range.
Again using Fact \ref{def: equality of eigenvalues}, we have that the nonzero eigenvalues of $\bv \Lambda_o \bv H$, and in turn those of $\frac{n}{s} \cdot\bv{A}_{o,S}$, are identical to those of $\bv H^{1/2}\bv \Lambda_o\bv H^{1/2}$. 

Let $\bv E=\bv H^{1/2}-\bv I=\bv U \bv W^{1/2}\bv U^T -\bv U \bv U^T=\bv U (\bv W^{1/2}-\bv I)\bv U^T$. 
Since the diagonal entries of $\bv W^{1/2}$ lie in $[1-\epsilon/6,1+\epsilon/6]$, those of $\bv W^{1/2}-\bv I$ lie in $[-\epsilon/6,\epsilon/6]$. Thus, $\| \bv E\|_2 \leq \epsilon/6$. We can write
\begin{align}
\lambda_i(\bv H^{1/2}\bv \Lambda_o\bv H^{1/2}) = \lambda_i((\bv I + \bv E)\bv \Lambda_o(\bv I + \bv E)) = \lambda_i(\bv \Lambda_o + \bv E\bv \Lambda_o + \bv \Lambda_o\bv E + \bv E\bv \Lambda_o\bv E).
    \notag
\end{align}
We can then bound
\begin{align*}
    \|\bv E\bv \Lambda_o + \bv \Lambda_o\bv E + \bv E\bv \Lambda_o\bv E\|_2 &\leq \|\bv E\bv \Lambda_o\|_2 + \|\bv \Lambda_o\bv E\|_2 + \|\bv E\bv \Lambda_o\bv E\|_2
    \notag\\
    &\leq \|\bv E\|_2\|\bv \Lambda_o\|_2 + \|\bv \Lambda_o\|_2\|\bv E\|_2 + \|\bv E\|_2\|\bv \Lambda_o\|_2\|\bv E\|_2 \notag\\
    &\leq \epsilon n/6 +n\epsilon/6 + \epsilon^2 n/36\notag\\
    &\leq \epsilon/2 \cdot n.
\end{align*}
Applying Weyl's eigenvalue perturbation theorem (Fact \ref{thm:eigenvalue_perturbation_theorem}), we thus have for all $i$,
\begin{align}
   \left| \lambda_i(\bv H^{1/2}\bv \Lambda_o\bv H^{1/2} ) - \lambda_i(\bv \Lambda_o)\right| &< \epsilon/2 \cdot n.
  \label{eq:applying weyls to eps4 bound}
\end{align}
Note that we have shown that the nonzero eigenvalues of $\frac{n}{s} \cdot \bv{A}_{o,S}$ are identical to those of $\bv H^{1/2}\bv \Lambda_o\bv H^{1/2}$, which we have shown well approximate those of $\bv{\Lambda}_o$ and in turn $\bv{A}_o$ \emph{i.e.}, the non-zero eigenvalues of $\frac{n}{s} \cdot \bv{A}_{o,S}$ approximate all outlying eigenvalues of $\bv{A}$. We can also bound the middle eigenvalues using Lemma \ref{middle} as in Theorem \ref{thm:main_bound}. Now the only thing left is to argue how these approximations `line up' in the presence of zero eigenvalues in the spectrum of these matrices. This part of the proof again proceeds similarly to that of Theorem~\ref{thm:main_bound} in Section \ref{sec:main accuracy bounds, uniform}.

Analogous to Theorem \ref{thm:main_bound}, from Lemma \ref{lem:uniform-eps4-near-orthonorm} equation \eqref{eq:applying weyls to eps4 bound} holds with probability $1-\delta$ if $s \geq \frac{c\log(1/(\epsilon\delta))}{\epsilon^4\delta}$. We also require $s \geq \frac{c\log n}{\epsilon^2\delta}$ for $\|\bv A_{m,S}\|_2 \leq \epsilon n$ to hold with probability $1-\delta$ by Lemma \ref{middle}. Thus, for both conditions to hold simultaneously with probability $1-2\delta$ by a union bound, it suffices to set  $s= \frac{c\log n}{\epsilon^4\delta} \geq \max\left(\frac{c\log(1/(\epsilon\delta))}{\epsilon^4\delta}, \frac{c\log n}{\epsilon^2\delta}\right)$, where we use that $\log(1/(\epsilon \delta)) = O(\log n)$, as otherwise our algorithm can take $\bv{A}_S$ to be all of  $\bv{A}$. Adjusting $\delta$ to $\delta/2$ completes the theorem.
\end{proof}


\section{Refined Bounds}\label{app:refined}

In this section, we show how it is possible to get better query complexity or tighter approximation factors by modifying the proof of Theorem~\ref{thm:main_bound} and Lemmas~\ref{lemma: orthonormality} and~\ref{lemma:row_norm} under some assumptions. We give an extension to Theorem~\ref{thm:main_bound} in Theorem \ref{cor: refined_bound1} for the case when the eigenvalues of $\bv A_o$ lie in a bounded range -- between $\epsilon^a\sqrt{\delta} n$ and $\epsilon^b n$ where $0 \leq b \leq a \leq 1$. 

\begin{theorem}\label{cor: refined_bound1}
Let $\bv A \in \mathbb{R}^{n \times n}$ be symmetric with $\|\bv{A}\|_{\infty} \leq 1$ and eigenvalues $\lambda_1(\bv{A}) \ge \ldots \ge \lambda_n(\bv{A})$. Let $\bv{A}_o $ be as in Definition~\ref{def:split} such that for all eigenvalues $\lambda_i(\bv{A}_o)$ we have either $\epsilon^a \sqrt{\delta}n \leq \lvert \lambda_i(\bv{A}_o) \rvert \leq \epsilon^b n$ or $\lambda_i(\bv{A}_o)=0$ where $0 \leq b \leq a \leq 1$. Let $S \subseteq [n]$ be formed by including each index independently with probability $s/n$ as in Algorithm \ref{alg:eigenvalue estimate}. Let $\bv A_S$ be the corresponding principal submatrix of $\bv{A}$, with eigenvalues $\lambda_1(\bv{A}_S) \ge \ldots \ge \lambda_{|S|}(\bv{A}_S)$.

For all $i \in [|S|]$ with $\lambda_i(\bv{A}_S) \ge 0$, let $\tilde \lambda_i(\bv{A}) = \frac{n}{s} \cdot \lambda_i(\bv{A}_S)$. For all $i \in [|S|]$ with $\lambda_i(\bv{A}_S) < 0$, let $\tilde \lambda_{n-(|S|-i)}(\bv{A}) = \frac{n}{s} \cdot \lambda_i(\bv{A}_S)$. For all other $i \in [n]$, let $\tilde \lambda_i(\bv{A}) = 0$. If $s \geq \frac{c \log (1/(\epsilon \delta))   \log^{2+a-b}n}{\epsilon^{2+a-b}\delta}$, for large enough $c$,
 then with probability at least $1-\delta$, for all $i \in [n]$,
\begin{align*}
    \lambda_i(\bv A) -\epsilon n \leq \tilde \lambda_i(\bv A) \leq \lambda_i(\bv A) +\epsilon n.
\end{align*}

\end{theorem}

\begin{proof}
The proof follows by modifying the proofs of Theorem~\ref{thm:main_bound}, Lemmas~\ref{lemma:row_norm} and~\ref{lemma: orthonormality} to account for the tighter intervals. 
First observe that since $\lvert \lambda_i(\bv A_o) \rvert \geq \epsilon^a \sqrt{\delta}n $ for all $i$, we can give a tighter row norm bound for $\bv{V}_o$ from the proof of Lemma~\ref{lemma:row_norm}. In particular, from equation~\eqref{Eq: eig_bound } we get: 
\begin{align*}
\norm{\bv \Lambda_o^{1/2}\bv{V}_{o,i}}_2^2 \leq \frac{1}{\epsilon^a \sqrt{\delta}} \hspace{2em}\text{ and }\hspace{2em}\|\bv V_{o,i} \|_2^2 &\leq \frac{n}{\epsilon^{2a} \delta n^2}=\frac{1}{\epsilon^{2a}\delta n}.
\end{align*}
We can then bound the number of samples we need to take such that for $\bv{\Lambda}_o^{1/2} \bv V_o^T \bar{\bv S} \bar{\bv S}^T \bv V_o \bv{\Lambda}_o^{1/2}$ (as defined in Theorem~\ref{thm:bound on large values}) we have $\|\bv{\Lambda}_o^{1/2} \bv V_o^T \bar{\bv S} \bar{\bv S}^T \bv V_o \bv{\Lambda}_o^{1/2}-\bv{\Lambda}_o \|_2 \leq \epsilon n $ with probability at least $1-\delta$ via a matrix Bernstein bound. 
By appropriately modifying the proof of Lemma~\ref{lemma: orthonormality} to incorporate the stronger row norm bound for $\bv{V}_o$, we can show that sampling with probability $s/n$ for $s \geq \frac{c \log(1/(\epsilon \delta))}{\epsilon^{2+a-b}\delta}$ for large enough $c$ suffices. Specifically, we get $L \leq \frac{n}{\epsilon^a\sqrt{\delta}s}$, $v \leq \frac{n^2}{\epsilon^{a-b}\sqrt{\delta}s}$ and $d\leq \log(1/(\epsilon^{2} \delta))$ for the Bernstein bound in Lemma~\ref{lemma: orthonormality} which enables us to get the tighter bound.
Thus, we have $\|\bv{\Lambda}_o^{1/2} \bv V_o^T \bar{\bv S} \bar{\bv S}^T \bv V_o \bv{\Lambda}_o^{1/2}-\bv{\Lambda}_o \|_2 \leq \epsilon n $ with probability $1-\delta$ for $s \geq \frac{c \log(1/(\epsilon \delta))}{\epsilon^{2+a-b}\sqrt{\delta}}$ following Lemma~\ref{lemma: orthonormality}. We also require $s \geq \frac{c\log n}{\epsilon^{2}\delta}$ for $\|\bv A_{m,S}\|_2 \leq \epsilon n$ to hold with probability $1-\delta$ by Lemma \ref{middle}. Then, following the proof of Theorem~\ref{thm:main_bound}, by
Fact \ref{fact:weyl_general}, for all $i \in [n]$, and some constant $C$, we have: $$\lvert \lambda_i(\bv{\Lambda}_o^{1/2} \bv V_o^T \bar{\bv S} \bar{\bv S}^T \bv V_o \bv{\Lambda}_o^{1/2})-\lambda_i(\bv{\Lambda}_o) \rvert \leq C \log n \|\bv{\Lambda}_o^{1/2} \bv V_o^T \bar{\bv S} \bar{\bv S}^T \bv V_o \bv{\Lambda}_o^{1/2}-\bv{\Lambda}_o \|_2.$$

As in the proof of Theorem \ref{thm:main_bound}, adjusting $\epsilon$ by a $\frac{1}{C\log n}$ factor, we get $\lvert \lambda_i(\bv{\Lambda}_o^{1/2} \bv V_o^T \bar{\bv S} \bar{\bv S}^T \bv V_o \bv{\Lambda}_o^{1/2})-\lambda_i(\bv{\Lambda}_o) \rvert \leq \epsilon n$ with probability $1-\delta$ for $s \geq \frac{c \log(1/(\epsilon \delta)) \log^{2+a-b} n}{\epsilon^{2+a-b}\sqrt{\delta}}$. Then we follow the proof of Theorem~\ref{thm:main_bound} to align the eigenvalues completing the proof.
\end{proof}

\section{Spectral Norm Bounds for Non-Uniform Random Submatrices}\label{app:gen_spectral_bounds}

\rudelson*
\begin{proof}
The proof follows from~\cite{tropp2007}. We begin by first defining the following term
\begin{align*}
    E \coloneqq \E_2 \|\bv{AS} \|_2.
\end{align*}
Now we have
\begin{align*}
    E^2 = \E \|\bv A \bv S \|^2_2 =\E \|\bv A\bv S \bv S \bv A^*\|_2 = \E \left\|\sum_{j=1}^n \delta^2_j \bv A_j \bv A_j^*\right\|_2,
\end{align*}
where $\delta_j$ is the sequence of independent random variables such that $\delta_j = \frac{1}{\sqrt{p_j}}$ with probability $p_j$ and $0$ otherwise, and $\bv A_j$ is the $j$\textsuperscript{th} column of $\bv A$. Then, $\mu_j  = \E[(\delta_j)^2]=1$. Let $\{\delta_j'\}$ be an independent copy of the sequence $\{\delta_j\}$.
Subtracting the mean and applying triangle inequality  we have
\begin{align*}
    E^2 \leq \E \left\| \sum_{j=1}^n (\delta^2_j -\E[(\hat{\delta})^2])\bv A_j \bv A_j^*\right\|_2 + \left\|\sum_{j=1}^n  \bv A_j \bv A_j^*\right\|_2.
\end{align*}
Using Jensen's inequality we have
\begin{align*}
    E^2 &\leq \E \left\| \sum_{j=1}^n (\delta^2_j -(\delta_j')^2)\bv A_j \bv A_j^*\right\|_2  + \left\| \bv A\bv A^*\right\|_2.
\end{align*}
The random variables $(\delta^2_j -(\delta_j')^2)$ are symmetric and independent. Let $\{\epsilon_j\}$ be i.i.d Rademacher random variables for all $j \in [n]$. Then applying the standard symmetrization argument followed by triangle inequality, we have:
\begin{align*}
    E^2 &\leq 2 \E \left\|\sum_{j=1}^n \epsilon_j\delta^2_j \bv A_j\bv A_j^*\right\|_2  + \left\| \bv A\bv A^*\right\|_2.
\end{align*}
Let $\Omega = \{j: \delta_j = \frac{1}{\sqrt{p_j}}\}$. Let $\mathbb{E}$ be the partial expectation with respect to $\{ \epsilon_j\}$, keeping the other random variables fixed. Then, we get:
\begin{align*}
    E^2 &\leq 2 \E_{\Omega}\left[\E_{\epsilon} \left\|\sum_{\Omega} \epsilon_j\delta_j^2 \bv A_j\bv A_j^T\right\|_2 \right]   + \| \bv A \|^2_2.
\end{align*}
Using Rudelson's Lemma 11 of \cite{tropp2007} for any matrix $\bv X$ with columns $\bv x_1, \bv x_2, \cdots, \bv x_n$ and any $q =2\log n$ we have
\begin{align*}
    \left(\E \left\|\sum_{j=1}^n \epsilon_j \bv x_j \bv x_j^*\right\|_2^q\right)^{1/q} &\leq 1.5 \sqrt{q} \|\bv X\|_{1\to 2} \|\bv X\|_2.
\end{align*}
Since $(.)^{1/q}$ is concave for $q \geq 1$, using Jensen's inequality we get:
\begin{align*}
    \E \left\|\sum_{j=1}^n \epsilon_j \bv x_j \bv x_j^*\right\|_2 &\leq 1.5 \sqrt{q} \|\bv X\|_{1\to 2} \|\bv X\|_2
\end{align*}
Applying the above result to the matrix $\bv A \bv S$, we get:
\begin{align*}
    E^2 &\leq 3\sqrt{q}\left[\mathbb{E}(\|\bv A \bv S \|_{1 \rightarrow 2} \|\bv A \bv S \|_2) \right]+\|\bv A \|_2^2.
\end{align*}
Applying Cauchy Schwartz we get:
\begin{align*}
    E^2 &\leq 3\sqrt{q}(\mathbb{E}\|\bv A \bv S \|^2_{1 \rightarrow 2})^{1/2}(\mathbb{E}\|\bv A \bv S \|^2_{2})^{1/2}+\|\bv A \|_2^2.
\end{align*}
The above equation is of the form $E^2 \leq bE+c$. Thus, the values of $E$ fro which the above equation is true is given by $E \leq \frac{b+\sqrt{b^2+4c}}{2} \leq b+\sqrt{c}$. Thus, we get:
\begin{align*}
    \E_2 \|\bv{AS} \|_2 \leq 3\sqrt{q}\E_2 \|\bv{AS} \|_{1 \rightarrow 2}+ \|\bv{A}\|_2.
\end{align*}
This gives us the final bound.
\end{proof}


\section{Improved Bounds via Row-Norm-Based Sampling}\label{sec:l2}
Building on the sparsity-based sampling results presented in Section \ref{sec:sparsity}, we now show how to obtain improved approximation error of $\pm \epsilon \norm{\bv A}_F$ assuming we can sample the rows of $\bv{A}$ with probabilties proportional to their squared $\ell_2$ norms. The ability to sample by norms also allows us to remove the assumption that $\bv A$ has bounded entries -- our results apply to any symmetric matrix.

For technical reasons, we  mix row norm sampling with uniform sampling, forming a random principal submatrix by sampling each index $i \in [n]$ independently with probability $p_i=\min \left (1,\frac{s\norm{\bv A_i}^2_2}{\norm{\bv A}^2_F}+\frac{1}{n^2}\right )$ and rescaling each sampled row/column by $1/\sqrt{p_i}$. As in the sparsity-based sampling setting, we must carefully zero out entries of the sampled submatrix to ensure concentration of the sampled eigenvalues. Pseudocode for the full algorithm is given in  Algorithm~\ref{alg:l2_eig_est}.

\subsection{Preliminary Lemmas}
Our proof closely follows that of Theorem \ref{thm:nnz_main_bound} in Section \ref{sec:sparsity}. We start by defining $\bv{A}' \in \mathbb{R}^{n \times n}$ obtained by zeroing out entries of $\bv{A}$ as described in Algorithm~\ref{alg:l2_eig_est}. We have $\bv{A}'_{ij}=0$ whenever  1) $i=j$ and $\norm{\bv A_i}^2_2 < \frac{\epsilon^2}{4}\norm{\bv A}^2_F$ or 2) $i \neq j$ and $\norm{\bv A_i}_2^2 \cdot \norm{\bv A_j}_2^2 < \frac{\epsilon^2 \norm{\bv A}_F^2 \cdot \lvert \bv A_{ij}\rvert^2}{c_2\log^4 n }$.  Otherwise $\bv{A}'_{ij}=\bv{A}_{ij}$. Similar to the sparsity sampling case, we  argue that the eigenvalues of $\bv{A}'$ are close to $\bv{A}$ i.e., zeroing out entries of $\bv{A}$ according to the given condition doesn't change it's eigenvalues by too much (Lemma \ref{lem:l2-zeroed}.  Then, we again split $\bv{A}'=\bv{A}'_o+\bv{A}'_m$ such that $\norm{\bv{A}'_{m}}_2 \leq \epsilon \sqrt{\delta}\norm{\bv{A}}_F$. We argue that after sampling, we have $\norm{\bv{A}'_{m,S}}_2 \le \epsilon \norm{\bv A}_F$ and the eigenvalues of $\bv{A}'_{o,S}$ approximate those of $\bv{A}'_o$ up to $\pm \epsilon \norm{\bv{A}}_F$ error.

\begin{algorithm}[h]
\caption{Eigenvalue estimator using $\ell_2$ norm-based sampling}
\label{alg:l2_eig_est}
\begin{algorithmic}[1]
\STATE {\bfseries Input:} Symmetric $\bv A \in \mathbb{R}^{n\times n}$, Accuracy $\epsilon \in (0,1)$, failure prob. $\delta \in (0,1)$.  $\norm{\bv A_i}_2$ for all $i \in [n]$.
\STATE Fix $s = \frac{c_1\log^{10} n }{\epsilon^8\delta^4}$ where $c_1$ is a sufficiently large constant.
\STATE Add each $i \in [n]$ to sample set $S$ independently, with probability $p_i=\min \left (1,\frac{s\norm{\bv A_i}^2_2}{\norm{\bv A}^2_F}+\frac{1}{n^2}\right )$. Let the principal submatrix of $\bv A$ corresponding to $S$ be $\bv A_S$. 
\STATE Let $\bv{A}_S = \bv{D} \bv{A}_S \bv{D}$ where $\bv{D} \in \R^{|S| \times |S|}$ is diagonal with $\bv{D}_{i,i} = \frac{1}{\sqrt{p_j}}$ if the $i^{th}$ element of $S$ is $j$. 
\STATE Construct $\bv{A}'_S \in \R^{|S| \times |S|}$ from $\bv{A}_S$  as follows:
\begin{align*}
    \bv [\bv{A}'_S]_{i,j} &=   
    \begin{cases}
    0 & \text{if $i = j$ and }\norm{\bv{A}_i}_2^2 < \frac{\epsilon^2}{4}\norm{\bv{A}}_F^2 \\
    0 & \text{if $i \neq j$ and }\norm{\bv A_i}_2^2 \cdot \norm{\bv A_j}_2^2 < \frac{\epsilon^2 \norm{\bv A}_F^2 \cdot \lvert \bv A_{ij}\rvert^2}{c_2\log^4 n }\text{ for sufficient large constant $c_2$}\\
     [\bv A_S]_{i,j} & \text{otherwise}.
    \end{cases}
\end{align*}
\STATE Compute the eigenvalues of $\bv A'_S$: $\lambda_1(\bv{A}'_S) \ge \ldots \ge \lambda_{|S|}(\bv{A}'_S)$.
\STATE For all $i \in [|S|]$ with $\lambda_i(\bv{A}'_S) \ge 0$, let $\tilde \lambda_i(\bv{A}) =  \lambda_i(\bv{A}'_S)$. For all $i \in [|S|]$ with $\lambda_i(\bv{A}'_S) < 0$, let $\tilde \lambda_{n-(|S|-i)}(\bv{A}) =  \lambda_i(\bv{A}'_S)$. For all remaining $i \in [n]$, let $\tilde \lambda_i(\bv{A}) = 0$. 
\STATE {\bfseries Return:} Eigenvalue estimates $\tilde \lambda_1(\bv{A}) \ge \ldots \ge \tilde \lambda_n(\bv{A})$.
\end{algorithmic}
\end{algorithm}

\begin{lemma}\label{lem:l2-zeroed}
    Let $\bv A \in \R^{n \times n}$ be symmetric.  Let $\bv{A}' \in \R^{n \times n}$ have $\bv A'_{ij} = 0$ if either 1) $i=j$ and $\norm{\bv{A}_i}_2^2 < \frac{\epsilon^2}{4}\norm{\bv{A}}_F^2$  or 2) $i \neq j$ and $\norm{\bv A_i}_2^2 \cdot \norm{\bv A_j}_2^2 < \frac{\epsilon^2 \norm{\bv A}_F^2 \cdot \lvert \bv A_{ij}\rvert^2}{c_2\log^4 n }$ for a sufficiently large constant $c_2$. Otherwise, $\bv A'_{ij} =  \bv A_{ij}$. Then, for all $i \in [n]$, $$|\lambda_i(\bv A) - \lambda_i(\bv A')| \le \epsilon \norm{\bv A}_F.$$
\end{lemma}

\begin{proof}
Consider the matrix $\bv{A}''$, which is defined identically to $\bv{A}'$ except we only set $\bv{A}''_{ij} = 0$  if $i \neq j$ and $\norm{\bv A_i}_2^2 \cdot \norm{\bv A_j}_2^2 < \frac{\epsilon^2 \norm{\bv A}_F^2\lvert \bv A_{ij}\rvert^2}{c_2\log^4 n}$. I.e., we do not zero out any entries on the diagonal as in $\bv{A}'$. We will show that $\norm{\bv A - \bv{A}''}_2 \le \frac{\epsilon}{2}  \norm{\bv{A}}_F$.
 If $\bv{A}_{ii}$ is zeroed out in $\bv{A}'$  this implies that $\norm{\bv{A}_i}_2^2 <\frac{\epsilon^2}{4}\norm{\bv{A}}_F^2 $. Thus, $|\bv{A}_{ii}| \le \norm{\bv A_i}_2 \le \frac{\epsilon}{2} \norm{\bv A}_F$ and so $\norm{\bv{A}''-\bv{A}'}_2 \le \frac{\epsilon}{2} \norm{\bv A}_F$. So, by triangle inequality, we will then have $\norm{\bv A -\bv A'}_2 \le \epsilon \cdot \norm{\bv A}_F$. The lemma then follows from Weyl's inequality

To show that $\norm{\bv A - \bv{A}''}_2 \le \frac{\epsilon}{2}  \norm{\bv{A}}_F$, we use a variant of Girshgorin's theorem, as in the proof of Lemma \ref{lem:nnz-zeroed}. First, we split the entries of $\bv{A}$ into level sets, according to their magnitudes.
Let $\bv{A}=\sum_{k=0}^{\log \frac{n}{\epsilon}}\bv{A}_k$ where $(\bv{A}_0)_{ij}=\bv A_{ij}$ if $\lvert \bv{A}_{ij} \rvert \in \left[0, \frac{\epsilon}{n}\norm{\bv{A}}_F \right)$ and $(\bv{A}_0)_{ij}=0$ otherwise. For $1\leq k\leq \log \frac{n}{\epsilon}$, $(\bv{A}_k)_{ij}=\bv A_{ij}$ if $\lvert \bv{A}_{ij} \rvert \in \left[\frac{\norm{\bv{A}}_F }{2^{k}}, \frac{\norm{\bv{A}}_F }{2^{k-1}}\right)$ and $(\bv{A}_k)_{ij}=0$ otherwise.
We can also define $\bv A''=\sum_{k=0}^{\log \frac{n}{\epsilon}}\bv{A}''_k$ where each $\bv{A}''_k$ are defined similarly.
By triangle inequality, $\norm{\bv{ A}-\bv{A}''}_2 \le \sum_{k=0}^{\log n/\epsilon} \norm{\bv{A}_k-\bv{A}_k''}_2$. First observe that $\norm{\bv{A}_0-\bv{A}_0''}_2 \le \norm{\bv{A}_0-\bv{A}_0''}_F \le n \cdot \norm{\bv A_0}_\infty \le \epsilon \norm{\bv A}_F$. Further, we can assume without loss of generality that $\epsilon > 1/n$ and so $\log(n/\epsilon) \le 2 \log n$, as otherwise our algorithm can afford to read all of $\bv{A}$. So, it suffices to show that for all $k \ge 1$,
\begin{align}\label{eq:kbound}
\norm{\bv{A}_k-\bv{A}_k''}_2 \le \frac{\epsilon}{\log n} \cdot \norm{\bv A}_F.
\end{align}
This will give  $\norm{\bv{ A}-\bv{A}''}_2 \le \epsilon \cdot \norm{\bv A}_F + \sum_{k=1}^{\log n/\epsilon}   \frac{\epsilon}{\log n} \cdot \norm{\bv A}_F \le 3 \epsilon \cdot \norm{\bv A}_F$, which gives the lemma after adjusting $\epsilon$ by a constant factor. 

We now prove \eqref{eq:kbound} for each $k \ge 1$. For $p \in \{0,1,\ldots \log(n^2)\}$, let $\mathcal{I}_p \subset [n]$ be the set of rows/columns in $\bv{A}_k$ with $\nnz((\bv A_k)_i) \in \left[\frac{\nnz(\bv A_k)}{2^p}, \frac{\nnz(\bv A_k)}{2^{p-1}}\right)$ and let $\bv A_{kpq} = \bv A_k(\mathcal{I}_p, \mathcal{I}_q)$ be the submatrix of $\bv{A}_k$ formed with rows in $\mathcal{I}_p$ and columns in $\mathcal{I}_q$. Define the submatrix $\bv{A}''_{kpq}$ of $\bv{A}''_k$ in the same way. Let $\bv{\widehat A}_{kpq}=\bv{A}_{kpq}-\bv{A}''_{kpq}$ and finally, let $\bar{\bv A}_{kpq} \in \mathbb{R}^{n \times n}$ be the symmetric error matrix such that $\bar{\bv A}_{kpq}(\mathcal{I}_p, \mathcal{I}_q)= \bv{\widehat A}_{kpq} $ and $\bar{\bv A}_{kpq}(\mathcal{I}_q, \mathcal{I}_p)= \bv{\widehat A}_{kpq}^T$.

Note that all rows from which we zero out entries must have at least one non-zero entry $\nnz((\bv{A}_k)_i) \geq 1$ (otherwise all entries in that row/column are already zero), thus all such rows have $\nnz((\bv{A}_k)_i) \ge \frac{\nnz(\bv{A}_k)}{n^2}$ and so are covered by the submatrices $\bv{A }_{kpq}$.
Thus, by triangle inequality, we can bound
\begin{align}\label{eq:kbound2}\norm{\bv A_k- \bv{A}_k''}_2 \le \sum_{p=0}^{\log(n^2)}  \sum_{q=0}^{\log(n^2)} \norm{\bv{\bar A}_{kpq}}_2.
\end{align}
To prove \eqref{eq:kbound}, we need to bound $\norm{\bv{A}_{kpq}-\bv{A}_{kpq}''}_2$  for all $k \ge 1$ and $p,q$.
We use a case analysis.

\medskip

\noindent{\textbf{Case 1:} $\frac{4 \nnz(\bv{A}_k)^2 \cdot c_2 \log^4n}{\epsilon^2 \cdot 2^{2k}} > 2^{p+q}.$} In this case, first observe that since the nonzero entries of $\bv{A}_k$ lie in $\left[\frac{\norm{\bv{A}}_F }{2^{k}}, \frac{\norm{\bv{A}}_F }{2^{k-1}}\right)$, for any $i \in \mathcal{I}_p$, $j \in \mathcal{I}_j$, 
\begin{align*}\norm{\bv A_i}_2^2 \cdot \norm{\bv A_j}_2^2 &\ge  \norm{(\bv{A}_{k})_{i}}_2^2 \cdot \norm{(\bv{A}_{k})_{j}}_2^2 \\ 
&\ge \frac{\norm{\bv A}_F^4}{2^{4k}} \cdot \nnz((\bv{A}_k)_i) \cdot \nnz((\bv{A}_k)_j)\\
&\ge \frac{\norm{\bv A}_F^4}{2^{4k}\cdot 2^{p+q}} \cdot \nnz(\bv{A}_k)^2.
\end{align*}
Thus, by the assumed bound on $2^{p+q}$, we have for any $i,j$ where $(\bv{A}_{k})_{ij}$ is nonzero,
\begin{align*}
\norm{\bv A_i}_2^2 \cdot \norm{\bv A_j}_2^2 \ge \frac{\epsilon^2 \norm{\bv{A}}_F^4}{4 \cdot 2^{2k} c_2 \log^4 n}  \ge \frac{\epsilon^2 \norm{\bv{A}}_F^2 \cdot |\bv A_{ij}|^2}{ c_2 \log^4 n},
\end{align*}
where the second inequality follows again from the fact that the nonzero entries of $\bv{A}_k$ lie in $\left[\frac{\norm{\bv{A}}_F }{2^{k}}, \frac{\norm{\bv{A}}_F }{2^{k-1}}\right)$. Thus, any $i,j$ with $(\bv{A}_{kpq})_{ij}$  nonzero is \emph{not zeroed out} in line 5 of Algorithm \ref{alg:l2_eig_est}. So $\bv{\bar A}_{kpq} = \bv{0}$.
Plugging into \eqref{eq:kbound2}, we thus have:
\begin{align}\label{eq:kbound3}\norm{\bv A_k- \bv{A}_k''}_2 \le \sum_{p=0}^{\log(n^2)}  \sum_{q: 2^{p+q} \ge \frac{16 \nnz(\bv{A}_k)^2 \cdot c_2 \log^4n}{\epsilon^2 \cdot 2^{2k}}} \norm{\bv{\bar A}_{kpq}}_2.
\end{align}

\medskip

\noindent{\textbf{Case 2:} $\frac{16 \nnz(\bv{A}_k)^2 \cdot c_2 \log^4n}{\epsilon^2 \cdot 2^{2k}} \le  2^{p+q}.$}
In this case, observe that  $(\bv{\widehat A}_{kpq} \bv{\widehat A}_{kpq}^T)_{m} = (\bv{\widehat A}_{kpq})_{m}\bv{\widehat A}_{kpq}^T$. 
We can see that $(\bv{\widehat A}_{kpq})_{m}$ has at most $\nnz((\bv{A}_k)_{m}) \le \frac{\nnz(\bv A_k)}{2^{p-1}}$ non-zero entries. Similarly, each row of $\bv{\widehat A}_{kpq}^T$ has at most $\frac{\nnz(\bv A_k)}{2^{q-1}}$ non-zero elements. Thus, for all $m \in |\mathcal{I}_p|$, using the fact that all non-zero entries of $\bv{A}_{kpq}$ are bounded by $\frac{\norm{\bv{A}}_F}{2^{k-1}}$, we have:
\begin{align*}
    \norm{(\bv{\widehat A}_{kpq} \bv{\widehat A}^T_{kpq})_{m}}_1 \leq \frac{\nnz(\bv A_k)^2}{2^{p+q-2}} \cdot \frac{\norm{\bv A}_F^2}{2^{2k-2}}.
\end{align*}
Applying Girshgorin's circle theorem (Theorem \ref{thm:girshgorin}) we thus have:
\begin{align*}
    \norm{\bv{\widehat A}_{kpq}}_2^2 = \norm{\bv{\widehat A}_{kpq}\bv{\widehat A}_{kpq}^T}_2 \le  \frac{\nnz(\bv A_k)^2}{2^{p+q-2}} \cdot \frac{\norm{\bv A}_F^2}{2^{2k-2}}
    \end{align*}
    and so 
\begin{align*}
    \|\bar{\bv A}_{kpq}\|_2 
    \leq 2\|\bv{\widehat A}_{kpq}\|_2 
    \leq \frac{8\cdot \norm{\bv{A}}_F \cdot \nnz(\bv A_k)}{2^{k}2^{\frac{(p+q)}{2}}}.
\end{align*}
Plugging to \eqref{eq:kbound3}, we thus have: 
\begin{align*}
\norm{\bv A_k- \bv{A}_k''}_2 &\le \sum_{p=0}^{\log(n^2)} \sum_{q: 2^{p+q} \ge \frac{16 \nnz(\bv{A}_k)^2 \cdot c_2 \log^4n}{\epsilon^2 \cdot 2^{2k}}} \frac{8\cdot \norm{\bv{A}}_F \cdot \nnz(\bv A_k)}{2^{k}2^{\frac{(p+q)}{2}}}\\
&\le \sum_{p=0}^{\log(n^2)} \frac{2\epsilon \cdot \norm{\bv{A}}_F}{\sqrt{c_2} \log^2 n} \cdot \sum_{i=0}^\infty \frac{1}{\sqrt{2}} \le \frac{8 \epsilon \norm{\bv{A}}_F}{\sqrt{c_2}}.
\end{align*}
Setting $c_2 \ge 64$, we thus have \eqref{eq:kbound}, and in turn the lemma.
\end{proof}

We next give a bound on the incoherence of the outlying eigenvectors of $\bv{A}'$. This bound is again similar to Lemmas \ref{lemma:row_norm} and~\ref{lemma:general_row_norm}.

\begin{lemma}[Incoherence of outlying eigenvectors in terms of $\ell_2$ norms]
\label{lemma:l2_row_norm}
Let $\bv A, \bv{A}' \in \mathbb{R}^{n\times n}$ be as in Lemma \ref{lem:l2-zeroed}. 
Let $\bv A'_o= \bv V'_o\bv \Lambda'_o \bv V_o^{'T}$ where $\bv \Lambda'_o$ is diagonal, with the eigenvalues of $\bv A'$ with magnitude $\ge \epsilon \sqrt{\delta} \norm{\bv{A}}_F$ on its diagonal, and $\bv{V}'_o$ has columns equal to the corresponding eigenvectors. Let $\bv V'_{o,i}$ denote the $i$\textsuperscript{th} row of $\bv V'_o$. Then, \begin{align*}
\|\bv{\Lambda}_o^{'1/2}\bv{V}'_{o,i} \|_2^2 \leq \frac{\norm{\bv{A}_i}_2^2}{\epsilon \sqrt{\delta} \norm{\bv{A}}_F}\hspace{1em}and\hspace{1em}\|\bv V'_{o,i}\|^2_2 \leq \frac{\norm{\bv A_i}_2^2}{\epsilon^2\delta \norm{\bv{A}}_F^2}.
\end{align*}
\end{lemma}
\begin{proof}
The proof is again nearly identical to that of Lemma \ref{lemma:row_norm}.
Observe that $\bv A' \bv V'_o =\bv V'_o \bv \Lambda'_o $. Letting $[\bv A' \bv V'_o]_i$ denote the $i$\textsuperscript{th} row of the $\bv A' \bv V'_o$, we have
\begin{equation}\label{Eq: row_norm1_l2}
    \|[\bv A' \bv V'_o]_i\|_2^2 = \|[\bv V'_o \bv \Lambda'_o]_i\|_2^2 = \sum_{j=1}^r \lambda_j^2 \cdot \bv V_{o,i, j}^{'2},
\end{equation}
where $r=\rank(\bv A'_o)$, $\bv V'_{o,i,j}$ is the $(i,j)$\textsuperscript{th} element of $\bv V'_o$  and $\lambda_j=\bv \Lambda'_o(j,j)$. Since $\bv{V}'_o$ has orthonormal columns, we have $ \|[\bv A' \bv V'_o]_i\|_2^2 \le  \|\bv A'_i\|_2^2 \le \norm{\bv A_i}_2^2 $.
Therefore, by \eqref{Eq: row_norm1_l2},
\begin{equation}\label{Eq: eig_bound_l2 }
    \sum_{j=1}^r \lambda_j^2 \cdot  \bv V_{o,i, j}^{'2} \leq \norm{\bv A_i}_2^2.
\end{equation}
Since by definition $\lvert \lambda_j \rvert \geq \epsilon \sqrt{\delta} \norm{\bv{A}}_F$ for all $j$, we can conclude that $
   \|\bv{\Lambda}_o^{'1/2}\bv{V}'_{o,i} \|_2^2 =\sum_{j=1}^r \lambda_j \cdot  \bv V_{o,i, j}^{'2} \leq \frac{\norm{\bv A_i}_2^2}{\epsilon \sqrt{\delta} \norm{\bv{A}}_F}$
and
$
    \|\bv V'_{o,i}\|_2^2 = \sum_{j=1}^r \bv V_{o,i, j}^{'2} 
    \leq \frac{\norm{\bv A_i}_2^2}{\epsilon^2\delta \norm{\bv{A}}_F^2}$, which completes the lemma.
\end{proof}

\subsection{Outer and Middle Eigenvalue Bounds}\label{sec:accuracy bounds_l2}

Using Lemma \ref{lemma:l2_row_norm}, we next argue that the eigenvalues of $\bv{A}_{o,S}'$ will approximate those of $\bv{A}'$, and in turn those of $\bv{A}$. The proof is very similar to Lemmas~\ref{lemma: orthonormality} and~\ref{lemma:nnz_large}.

\begin{lemma}[Concentration of outlying eigenvalues with $\ell_2$ norm based sampling]\label{lemma:l2_large}
Let $\bv{A},\bv{A}' \in\mathbb{R}^{n\times n}$ be as in algorithm \ref{alg:l2_eig_est}.
Let $\bv A' = \bv A'_m + \bv A'_o$, where $\bv A'_m = \bv V'_m\bv{\Lambda}'_m\bv {\bv V'}_m^{T}$, and 
$\bv A'_o = \bv V'_o\bv{\Lambda}'_o\bv {\bv V'}_o^{T}$ are projections onto the eigenspaces with magnitude $< \epsilon\sqrt{\delta}\norm{\bv{A}}_F$ and $\ge \epsilon\sqrt{\delta}\norm{\bv{A}}_F$ respectively. For all $i \in [n]$ let $p_i=\min \left (1,\frac{s\norm{\bv A_i}^2_2}{\norm{\bv A}^2_F}+\frac{1}{n^2}\right )$ and let $\bar{\bv S}$ be a scaled diagonal sampling matrix such that the $\bar{\bv S}_{ii}=\frac{1}{\sqrt{p_i}}$ with probability $p_i$ and $\bar{\bv S}_{ii}=0$ otherwise. If $s \geq \frac{c\log(1/(\epsilon \delta)) }{\epsilon^3 \sqrt{\delta}}$ for a large enough constant $c$, then with probability at least $1-\delta$, $\|\bv \Lambda_o^{'1/2}\bv V_o^{'T} \bar{\bv S} \bar{\bv S}^T \bv V'_o\bv \Lambda_o^{'1/2} - \bv \Lambda'_o \|_2 \leq \epsilon \norm{\bv{A}}_F$.   
\end{lemma}
\begin{proof}
We define the random variables $\bv{Q}_1, \cdots \bv{Q}_n$ and the set $P=\{i \in [n]: p_i < 1\}$ exactly as in the proof of Lemma~\ref{lemma:nnz_large}.  Then, as explained in the proof of Lemma~\ref{lemma:nnz_large} it is sufficient to bound $\sum_{i\in P}\mathbb{E}[\bv Q_i^2]$. From~\ref{var_ineq_nnz} we have $\sum_{i\in P}\mathbb{E}[\bv Q_i^2] \preceq \sum_{i\in P} \frac{1}{p_{i}} \cdot \norm{\bv \Lambda_o^{1/2} \bv V_{o,i}}_2^2 \cdot  (\bv \Lambda_o^{1/2}\bv V_{o,i}  \bv V_{o,i}^T\bv \Lambda_o^{1/2})$. Also from Lemma~\ref{lem:l2-zeroed}, we have  $\norm{\bv \Lambda_o^{1/2} \bv V_{o,i}}_2^2 \le \frac{\norm{\bv{A}_i}_2^2}{\epsilon \sqrt{\delta} \norm{\bv{A}}_F}$ and for all $i \in P$, $\frac{1}{p_i} \leq \frac{\norm{\bv{A}}_F^2}{s\norm{\bv{A}_i}_2^2}$. We thus get,
 \begin{align*}
 \sum_{i\in P} \mathbb{E}[\bv Q_i^2] &\preceq\sum_{i\in P} \frac{1}{p_{i}} \cdot \frac{\norm{\bv{A}_i}_2^2}{\epsilon \sqrt{\delta} \norm{\bv{A}}_F} \cdot  (\bv \Lambda_o^{1/2}\bv V_{o,i} \bv V_{o,i}^T \bv \Lambda_o^{1/2}) \\
 &\preceq \frac{\norm{\bv{A}}_F}{s \epsilon \sqrt{\delta}}(\sum_{i\in P} \Lambda_o^{1/2}\bv V_{o,i} \bv V_{o,i}^T \bv \Lambda_o^{1/2}) \\
 &= \frac{\norm{\bv{A}}_F}{s \epsilon \sqrt{\delta}}\bv{\Lambda}_o \preceq  \frac{\norm{\bv{A}}_F^2}{s \epsilon \sqrt{\delta}} \cdot \bv{I}.
 \end{align*}
 Since $\bv{Q}_i^2$ is PSD this establishes that $v \leq \|\textbf{Var(E)} \|_2 \leq \frac{\norm{\bv{A}}_F^2}{s \epsilon \sqrt{\delta}}$. We can then apply the matrix Bernstein inequality exactly as in the proof of Lemma \ref{lemma: orthonormality} to show that when $s \geq \frac{c }{\epsilon^3 \sqrt{\delta}}$ for large enough $c$, with probability at least $1-\delta$, $\left\|\bv E\right\|_2 \leq \epsilon \norm{\bv{A}}_F $.
\end{proof}

We now bound the middle eignevalues.

\begin{lemma}[Concentration of middle eigenvalues with $\ell_2$- norm based sampling]
\label{lem:l2-middle}
Let $\bv{A},\bv{A}' \in\mathbb{R}^{n\times n}$ be as in Lemma \ref{lemma:l2_row_norm}.
Let $\bv A' = \bv A'_m + \bv A'_o$, where $\bv A'_m = \bv V'_m\bv{\Lambda}'_m\bv {\bv V'}_m^{T}$, and 
$\bv A'_o = \bv V'_o\bv{\Lambda}'_o\bv {\bv V'}_o^{T}$ are projections onto the eigenspaces with magnitude $< \epsilon\sqrt{\delta}\norm{\bv{A}}_F$ and $\ge \epsilon\sqrt{\delta}\norm{\bv{A}}_F$ respectively (analogous to Definition \ref{def:split}).
As in Algorithm \ref{alg:nnz eigenvalue estimate}, for all $i \in [n]$ let $p_i=\min\left (1,\frac{s \norm{\bv{A}_i}_2^2}{\norm{\bv{A}}_F^2}+\frac{1}{n^2}\right )$ and let $\bar{\bv S}$ be a scaled diagonal sampling matrix such that the $\bar{\bv S}_{ii}=\frac{1}{\sqrt{p_i}}$ with probability $p_i$ and $\bar{\bv S}_{ii}=0$ otherwise. If $s \geq  \frac{c\log^{10} n }{\epsilon^8\delta^4}$ for a large enough constant $c$, then with probability at least $1-\delta$, $$\norm{\bar{\bv S}\bv A'_{m}\bar{\bv S}}_2 \leq \epsilon \norm{\bv{A}}_F.$$
\end{lemma}
\begin{proof}
First observe that since $s \geq \frac{4}{\epsilon^2}$ (for large enough $c$), the results of Lemmas~\ref{lem:l2-zeroed} and~\ref{lemma:l2_row_norm} still hold. The proof follows the same structure as the proof of bounding the middle eigenvalues for sparsity sampling in  Lemma~\ref{lem:nnz-middle}.
From Lemma \ref{lemma:l2_row_norm}, we have $\|{\bv V'}_{o,i}\|_2 \leq \frac{\norm{\bv{A}_i}_2}{\epsilon\sqrt{\delta} \norm{\bv{A}}_F}$. 
Also, following the proof of Lemma~\ref{lemma:l2_row_norm},  we have $\|{\bv \Lambda'}_o {\bv V'}^T_{o,j} \|_2 = \|[{\bv A'}{\bv V'}_o]_j \|_2 \leq \norm{\bv{A}_j}_2$.
Thus, for all $i,j \in [n]$, using Cauchy Schwarz's inequality, we have 
\begin{align}
\label{eq:l2-aoij}
    |{\bv A'}_{o,i,j}|=|{\bv V'}_{o,i} {\bv \Lambda'}_o {\bv V'}_{o,j}^T| \leq \|{\bv V'}_{o,i}\|_2 \cdot \|{\bv \Lambda'}_o {\bv V'}_{o,j}^T\|_2 \leq \frac{\norm{\bv{A}_i}_2}{\epsilon \sqrt{\delta} \norm{\bv{A}}_F} \cdot \norm{\bv{A}_j}_2.
\end{align} 
Let ${\bv A'}_m = \bv H_m + \bv D_m$ where $\bv{H}_m$ and $\bv{D}_m$ contain the off-diagonal and diagonal elements of $\bv{A}'_m$ respectively. Then following the proof of Lemma~\ref{lem:nnz-middle}, we get:
\begin{align}\label{eq:l2_decoupled2}
    \E_2 \|\bar{\bv S}{\bv A'}_m \bar{\bv S}\|_2 \le 10\sqrt{\log n}\left (\E_2\|\bar{\bv S}\bv H_m \hat{\bv S}\|_{1\to 2} + \E_2\|\bv H_m\hat{\bv S}\|_{1\to 2}\right ) + 2\|\bv H_m\|_2 + \E_2\|\bar{\bv S}\bv D_m \bar{\bv S}\|_2
\end{align}

We now proceed to bound each of the terms on the right hand side of \eqref{eq:l2_decoupled2}. We start with $\E_2\|\bar{\bv S}\bv D_m \bar{\bv S}\|_2$. First, observe that $\E_2\|\bar{\bv S}\bv D_m \bar{\bv S}\|_2 \leq \max_i \frac{1}{p_i}\lvert (\bv{D}_m)_{ii} \rvert$. We consider two cases.

\smallskip

\noindent \textbf{Case 1:} $p_i<1$. Then, as $p_i \geq \frac{s \norm{\bv{A}_i}_2^2}{\norm{\bv{A}}_F^2}$ we have $\norm{\bv{A}}_F^2 \leq \frac{1}{s}\norm{\bv{A}_i}_2^2$ since $\frac{1}{s} < \frac{\epsilon^2}{4}$. So we must have that have $\lvert (\bv{D}_m)_{ii} \rvert= \lvert ({\bv{A}'}_m)_{ii} \rvert = \lvert (\bv{A}'_o)_{ii} \rvert $ (since $\bv{A}'_{ii}=0$). Then  by \eqref{eq:l2-aoij}, we have $\frac{1}{p_i}\lvert (\bv{D}_m)_{ii} \rvert \leq \frac{\norm{\bv{A}}_F}{s\epsilon \sqrt{\delta}}  $.

 \smallskip

\noindent\textbf{Case 2:} $p_i=1$. Then we have $\frac{1}{p_i}\lvert (\bv{D}_m)_{ii} \rvert=\lvert (\bv{D}_m)_{ii} \rvert \leq \max_j \lvert (\bv{D}_m)_{jj} \rvert \leq \|\bv{A}'_m \|_2 \leq \epsilon \sqrt{\delta}\norm{\bv{A}}_F$. \\
From the two cases above, for $s \geq \frac{1}{\epsilon^2 \delta}$, we have:
\begin{align}
    \label{eq:l2-sds-bound}
    \E_2\|\bar{\bv S}\bv D_m \bar{\bv S}\|_2 \leq \epsilon\sqrt{\delta}\norm{\bv{A}}_F.
\end{align} 
We can bound $\norm{\bv{H}_m}_2$ similarly. Since $\bv H_m = {\bv A'}_m - \bv D_m$ and $\|{\bv A'}_m\|_2 \leq \epsilon\sqrt{\delta}\norm{\bv{A}}_F.$,
\begin{align}
    \|\bv H_m\|_2 &\leq \|{\bv A'}_m\|_2 + \|\bv D_m\|_2\notag\\
    &\leq \epsilon\sqrt{\delta}\norm{\bv{A}}_F + \epsilon\sqrt{\delta}\norm{\bv{A}}_F.\notag\\
    &= 2\epsilon\sqrt{\delta}\norm{\bv{A}}_F.\label{eq:l2-hm-bound}
\end{align}
where the second step follows from the fact that $\norm{\bv{D}_m}_2 \le \max_i \lvert (\bv{D}_m)_{ii} \rvert  \le \norm{\bv{A}'_m}_2$.

We next bound the term $\E_2\|\bv H_m\hat{\bv S}\|_{1\to 2}$.
 Observe that $\E_2\|\bv H_m\hat{\bv S}\|_{1\to 2} \leq \frac{\max_i \|{\bv A'}_{m,i}\|_2}{\sqrt{p_i}}$, where $\bv{A'}_{m,i}$ is the $i$\textsuperscript{th} column/row of $\bv{A}'_m$. We again consider the two cases when $p_i = 1$ and $p_i<1$:
 
 \smallskip
 
 \noindent\textbf{Case 1:} $p_i = 1$. Then $\|{\bv A'}_{m,i}\|_2 \leq \|{\bv A'}_m\|_2 \leq \epsilon\sqrt{\delta}\norm{\bv{A}}_F$.
 
  \smallskip

 \noindent\textbf{Case 2:} $p_i < 1$. Then $\|{\bv A'}_{m,i}\|_2 \leq \|{\bv A'}_i\|_2 \leq \norm{\bv{A}}_F$. Thus, setting $s \geq \frac{1}{\epsilon^2\delta}$ we have:
 \begin{align*}
     \frac{\|{\bv A'}_{m,i}\|_2}{\sqrt{p_i}}
    &\leq \frac{\norm{\bv{A}}_F}{\sqrt{s}\norm{\bv{A}_i}_2} \cdot \|{\bv A'}_i\|_2\\
    &\leq \frac{\norm{\bv{A}}_F}{\sqrt{s}} \leq \epsilon\sqrt{\delta}\norm{\bv{A}}_F.
\end{align*}
Thus, from the two cases above, for all $i \in [n]$, adjusting $\epsilon$ by a $\frac{1}{\sqrt{\log n}}$ factor, we have
for $s \geq \frac{\log n}{\epsilon^2\delta}$:
\begin{align}
\label{eq:l2-hms-bound}
    \E_2\|\bv H_m\hat{\bv S}\|_{1 \to 2} &\leq \frac{\epsilon\sqrt{\delta}\norm{\bv{A}}_F}{\sqrt{\log n}}.
\end{align}

Overall, plugging \eqref{eq:l2-sds-bound}, \eqref{eq:l2-hm-bound}, and \eqref{eq:l2-hms-bound} back into \eqref{eq:l2_decoupled2}, we have :
\begin{align}\label{eq:l2_decoupled3}
    \E_2 \|\bar{\bv S}{\bv A'}_m \bar{\bv S}\|_2 \le 10\sqrt{\log n} \cdot \E_2\|\bar{\bv S}\bv H_m \hat{\bv S}\|_{1\to 2} + 15 \epsilon \sqrt{\delta}\norm{\bv{A}}_F.
\end{align}
Finally we bound $\E_2\|\bar{\bv S}\bv H_m \hat{\bv S}\|_{1\to 2}$. As in the proof of Lemma~\ref{lem:nnz-middle}, we have $\E_2\|\bar{\bv S}\bv H_m\hat{\bv S}\|_{1\to 2} \leq \E_2 \left (\max_{i: i\in [n]} \frac{\|(\bar{\bv S}\bv H_m)_{:,i}\|_2}{\sqrt{p_i}}\right )$ and we will argue that $\max_{i: i\in [n]} \frac{\|(\bar{\bv S}\bv H_m)_{:,i}\|_2}{\sqrt{p_i}}$ is bounded by $\epsilon\sqrt{\delta} \norm{\bv{A}}_F$ with probability $1-1/\poly(n)$. Also as argued in the proof of Lemma~\ref{lem:nnz-middle}, since $p_i \geq \frac{1}{n^2}$, it suffices to bound $\frac{\|(\bar{\bv S}\bv A'_m)_{:,i}\|_2}{\sqrt{p_i}}$ for all $i\in [n]$ with high probability. Again, for a fixed $i$ and any $j \in [n]$, define the random variables $z_j$ as:
\begin{align*}
    z_j &=
    \begin{cases}
    \frac{1}{p_j}|{\bv A'}_{m,i,j}|^2 & \text{with probability $p_j$}\\
    0 & \text{otherwise}.
    \end{cases}
\end{align*}
Then $\sum_{j=1}^n z_j = \|(\bar{\bv S}{\bv A'}_{m})_{:,i}\|_2^2$ and $\mathbb{E}[\sum_{j=1}^n z_j] = \|\bv{A}'_{m,i} \|_2^2 \leq \|\bv{A}'_{i} \|_2^2 \leq \norm{\bv{A}}_F^2$. We will again use Bernstein's inequality to bound $\sum_{j=1}^n z_j = \|(\bar{\bv S}{\bv A'}_{m})_{:,i}\|_2^2$ by bounding bound $|z_j|$ for all $j \in [n]$ and $\bv{Var}\left(\sum_{j=1}^n z_j\right)$.  We consider the cases of $p_i <1$ and $p_i = 1$ separately.

\smallskip

\noindent \textbf{Case 1:} $p_i<1$. Then, we have $p_i \geq s\norm{\bv{A}_i}_2^2 / \norm{\bv{A}}_F^2$. If ${\bv A'}_{i,j} \neq 0$ then
\begin{align*}
    |z_j| &\leq \frac{1}{p_j} |{\bv A'}_{m,i,j}|^2 \leq \max\left(1, \frac{\norm{\bv{A}}_F^2}{s\norm{\bv{A}_j}_2^2}\right) |{\bv A'}_{m,i,j}|^2\\
    &\leq |{\bv A'}_{m,i,j}|^2 + \frac{2\norm{\bv{A}}_F^2}{s\norm{\bv{A}_j}_2^2}\left(|{\bv A'}_{i,j}|^2 + |{\bv A'}_{o,i,j}|^2\right)\\
    &\leq |{\bv A'}_{m,i,j}|^2 + \frac{2\norm{\bv{A}}_F^2}{s\norm{\bv{A}_j}_2^2}\left(|{\bv A'}_{i,j}|^2 + \frac{\norm{\bv{A}_i}_2^2\norm{\bv{A}_j}_2^2}{\epsilon^2\delta \norm{\bv{A}}_F^2}\right)\\
    &\leq |{\bv A'}_{m,i,j}|^2 + \frac{2\norm{\bv{A}}_F^2}{s\norm{\bv{A}_j}_2^2}|{\bv A'}_{i,j}|^2 + \frac{2\norm{\bv{A}_i}_2^2}{\epsilon^2\delta s},
\end{align*}
where the fourth inequality uses \eqref{eq:l2-aoij}.
By the thresholding procedure which defines $\bv{A}'$, if $i \neq j$ and $\bv A'_{ij} \neq 0$,
\begin{align}
    \norm{\bv{A}_i}_2^2 \cdot \norm{\bv{A}_j}_2^2 
    \geq \frac{\epsilon^2 \norm{\bv{A}}_F^2 |\bv{A}'_{ij}|^2}{c_2\log^4 n }
    \Rightarrow
    \frac{\norm{\bv{A}_j}_2^2}{|\bv{A}'_{i,j}|^2} 
    \geq \frac{\epsilon^2 \norm{\bv{A}}_F^2}{c_2 \cdot \log^4 n \cdot \norm{\bv{A}_i}_2^2 },\label{eq:l2-aj}
\end{align}
and thus for $p_i < 1$ and ${\bv A'}_{ij} \neq 0$ we have
\begin{align*}
    |z_j| &\leq |{\bv A'}_{m,i,j}|^2 + \frac{2c_2 \log^4 n \cdot \norm{\bv{A}_i}_2^2}{s\epsilon^2} + \frac{2\norm{\bv{A}_i}_2^2}{\epsilon^2\delta s}.
\end{align*}
 Also $\bv{A}'_{ii}=0$ since we must have $\norm{\bv{A}_i}_2^2 < \frac{\epsilon^2}{4}\norm{\bv{A}}_F^2$ as $p_i <1$. If ${\bv A'}_{i,j} = 0$ or $i=j$, then we simply have
\begin{align*}
    |z_j| &\leq |{\bv A'}_{m,i,j}|^2 + \frac{2\norm{\bv{A}_i}_2^2}{s\epsilon^2\delta}.
\end{align*}
Overall for all $j \in [n]$,
\begin{align}
    |z_j| &\leq |{\bv A'}_{m,i,j}|^2 + \frac{2\norm{\bv{A}_i}_2^2}{s\epsilon^2\delta} + \frac{2c_2\log^4 n  \cdot \norm{\bv{A}_i}_2^2}{s\epsilon^2}, \label{eq:l2_partial_abs_z}
\end{align}
and since $|{\bv A'}_{m,i,j}|^2 \leq  \sum_{j=1}^n |{\bv A'}_{m,i,j}|^2 = \|{\bv A'}_{m,i}\|_2^2 \leq \|{\bv A'}_i\|_2^2 \leq \norm{\bv{A}_i}_2^2$, 
\begin{align}
    |z_j| &\leq \norm{\bv{A}_i}_2^2 + \frac{2\norm{\bv{A}_i}_2^2}{s\epsilon^2\delta} + \frac{2 c_2 \cdot \log^4 n  \cdot \norm{\bv{A}_i}_2^2}{s\epsilon^2}. \label{eq:l2_abs_z}
\end{align}
For $s \geq c\left (\frac{\log^4 n }{\epsilon^2} + \frac{1}{\epsilon^2 \delta}\right )$ and large enough $c$, we thus have $|z_j| \leq 2\norm{\bv{A}_i}_2^2$. 

\noindent We next bound the variance by:
\begin{align*}
    \bv{Var}\left(\sum_{j=1}^n z_j\right) &\leq \sum_{j=1}^n \E [z_j^2] \leq \sum_{j=1}^n p_j\frac{1}{p_j^2} |{\bv A'}_{m,i,j}|^4\\
    &= \sum_{j=1}^n \max \left(1, \frac{\norm{\bv{A}}_F^2}{s\norm{\bv{A}_j}_2^2}\right)|{\bv A'}_{m,i,j}|^4 \\
    &\leq \sum_{j=1}^n |{\bv A'}_{m,i,j}|^4 + \sum_{j=1}^n \frac{12\norm{\bv{A}}_F^2}{s\norm{\bv{A}_j}_2^2} \left(|{\bv A'}_{i,j}|^4 + |{\bv A'}_{o,i,j}|^4\right)\\
    &\leq \|{\bv A'}_{m,i}\|_2^4 + \sum_{j=1}^n \frac{12\norm{\bv{A}}_F^2}{s\norm{\bv{A}_j}_2^2} \left(|\bv A_{i,j}'|^4 + \frac{\norm{\bv{A}_i}_2^4\norm{\bv{A}_j}_2^4}{\epsilon^4\delta^2\norm{\bv{A}}_F^4}\right),
\end{align*}
where the last inequality uses \eqref{eq:l2-aoij}. We thus get:
\begin{align}
    \bv{Var}\left(\sum_{j=1}^n z_j\right) &\leq \|{\bv A'}_{m,i}\|_2^4 + \sum_{j:{\bv A'}_{i,j}\neq 0}\frac{12\norm{\bv{A}}_F^2 |\bv{A}'_{ij}|^4}{s\norm{\bv{A}_j}_2^2} + \sum_{j=1}^n \frac{12\norm{\bv{A}_i}_2^4\norm{\bv{A}_j}_2^2}{s\epsilon^4\delta^2\norm{\bv{A}}_F^2}. \label{eq:l2_unsolved_var_sum_zj_1}
\end{align}
\noindent Now $\bv A_{ii}' = 0$ as $p_i <1 $ (and thus, $\norm{\bv{A}}_i^2 < \frac{\epsilon^2}{4}\norm{\bv{A}}_F^2$). Combining \eqref{eq:l2-aj} with the second term to the right of \eqref{eq:l2_unsolved_var_sum_zj_1} we have
\begin{align*}
    \bv{Var}\left(\sum_{j=1}^n z_j\right) &\leq \|{\bv A'}_{m,i}\|_2^4 + \sum_{j:{\bv A'}_{i,j}\neq 0}\frac{12c_2\log^4 n \cdot \norm{\bv{A}_i}_2^2 \cdot |\bv{A}'_{ij}|^2}{s\epsilon^2} + \sum_{j=1}^n \frac{12\norm{\bv{A}_i}_2^4\norm{\bv{A}_j}_2^2}{s\epsilon^4\delta^2\norm{\bv{A}}_F^2}, 
\end{align*}
and since $\sum_{j} |\bv{A}'_{ij}|^2= \norm{\bv{A}_i}_2^2$, we have
\begin{align}
    \bv{Var}\left(\sum_{j=1}^n z_j\right) &\leq \|{\bv A'}_{m,i}\|_2^4 + \frac{12c_2\log^4 n  \cdot \norm{\bv{A}_i}_2^4}{s\epsilon^2} + \sum_{j=1}^n \frac{12\norm{\bv{A}_i}_2^4\norm{\bv{A}_j}_2^2}{s\epsilon^4\delta^2\norm{\bv{A}}_F^2}. \label{eq:l2_unsolved_var_sum_zj_2}
\end{align}
\noindent Finally since $\sum_{j=1}^n \norm{\bv{A}_j}_2^2 = \norm{\bv{A}}_F^2$ and $\|{\bv A'}_{m,i}\|_2^4 \leq \norm{\bv{A'}_i}_2^4 \le \norm{\bv{A}_i}_2^4$ we have
\begin{align}
    \bv{Var}\left(\sum_{j=1}^n z_j\right) &\leq \norm{\bv{A}_i}_2^4 + \frac{12c_2 \log^4 n  \cdot \norm{\bv{A}_i}_2^4}{s\epsilon^2} + \frac{12\norm{\bv{A}_i}_2^4}{s\epsilon^4\delta^2}.\label{eq:l2_var_sum_zj}
\end{align}

\noindent For $s \geq c\left (\frac{\log^4 n }{\epsilon^2} + \frac{1}{\epsilon^4\delta^2}\right )$ for large enough $c$, we have $\bv{Var}\left(\sum_{j=1}^n z_j\right) \leq 2\norm{\bv{A}_i}_2^4$. 

\noindent Therefore, using \eqref{eq:l2_abs_z} and \eqref{eq:l2_var_sum_zj} with $s \geq c\left (\frac{\log^4 n }{\epsilon^2} + \frac{1}{\epsilon^4\delta^2}\right )$, we can 
apply Bernstein inequality (Theorem \ref{thm:bernstein}) (for some constant $c$) to get
\begin{align*}
    \Pr\left(\|(\bar{\bv S}\bv A'_m)_{:,i}\|_2^2 \geq \E\|(\bar{\bv S}\bv A'_m)_{:,i}\|_2^2 + t \right) &\leq \Pr\left(\sum_{j=1}^n z_j \geq \norm{\bv{A}_i}_2^2 + t \right)\\
    &\leq \exp\left(\frac{-t^2/2}{c\norm{\bv{A}_i}_2^4 + ct\norm{\bv{A}_i}_2^2/3}\right).
\end{align*}
If we set $t= \log n \cdot \norm{\bv{A}_i}_2^2$, for some constant $c'$ we have
\begin{align*}
     \Pr\left(\|(\bar{\bv S}\bv A'_m)_{:,i}\|_2^2 \geq \E\|(\bar{\bv S}\bv A'_m)_{:,i}\|_2^2+\log n \cdot \norm{\bv{A}_i}_2^2 \right)
     &\leq \exp\left(\frac{-(\log n)^2/2}{c+c(\log n)/3}\right) \leq  \exp(-c'\log n) \leq 1/n^{c'}.
\end{align*}
Since ${\bv A'}_m = \bv H_m + \bv D_m$, we have $\|(\bar{\bv S}{\bv A'}_m)_{:,i}\|_2 \geq \|(\bar{\bv S}\bv H_m)_{:,i}\|_2$. Then with probability at least $1-1/n^{c'} \ge 1 -\delta$, for any row $i$ with $p_i <1$, we have 
\begin{align*}
    \frac{1}{p_i} \cdot \|(\bar{\bv S}\bv H_m)_{:,i}\|_2^2 \leq \frac{\norm{\bv{A}}_F^2}{s\norm{\bv{A}_i}_2^2}\cdot c(\log n)\norm{\bv{A}_i}_2^2 \leq \frac{\epsilon^2 \delta \norm{\bv{A}}_F^2}{\log n},
\end{align*}
for $s \geq c\left (\frac{\log^4 n }{\epsilon^2} + \frac{1}{\epsilon^4\delta^2}\right )$ for large enough $c$.
Observe that, as in Lemma \ref{lemma: orthonormality} w.l.o.g. we have assumed $1-\frac{1}{n^{c'}} \ge 1-\delta$, since otherwise, our algorithm would read all $n^2$ entries of the matrix.

\smallskip

\noindent \textbf{Case 2:} $p_i = 1$. Then, we have $\norm{\bv{A}_i}_2^2 \geq \norm{\bv{A}}_F^2/s$. As in the $p_i < 1$ case, when $\bv{A}_{ii}=0$, (and this $\bv{A}'_{ii}=\bv{A}_{ii}=0$) we have from~\eqref{eq:l2_partial_abs_z}: 
\begin{align*}
    |z_j| &\leq |{\bv A'}_{m,i,j}|^2 + \frac{2\norm{\bv{A}_i}_2^2}{s\epsilon^2\delta} + \frac{2c_2\log^4 n  \cdot \norm{\bv{A}_i}_2^2}{s\epsilon^2}.
\end{align*}
Now, we observe that $|{\bv A'}_{m,i,j}|^2 \leq \sum_{j=1}^n |{\bv A'}_{m,i,j}|^2 \leq \| \bv{A}'_{m,i}\|^2_2 \leq \| \bv{A}'_m\|^2_2 \leq \epsilon^2\delta \norm{\bv{A}}_F^2$, which gives us
\begin{align}
    |z_j| &\leq \epsilon^2 \delta \norm{\bv{A}}_F^2 + \frac{2\norm{\bv{A}_i}_2^2}{s\epsilon^2\delta} + \frac{2c_2\log^4 n  \cdot \norm{\bv{A}_i}_2^2}{s\epsilon^2}.\label{eq: l2_abs_zj_p=1}
\end{align}
Note that if $\bv{A}_{ii} \neq 0$, the second term in~\eqref{eq:l2_partial_abs_z} is bounded as $\frac{2\norm{\bv{A}}_F^2}{s\norm{\bv{A}_i}_2^2}\cdot |\bv{A}'_{ii}|^2 \leq \frac{2\norm{\bv{A}}_F^2}{s} \leq 2\epsilon^2\delta \norm{\bv{A}}_F^2$ for $s\geq O(\frac{1}{\epsilon^2\delta})$.
Thus, for $s \geq c\left (\frac{\log^4 n }{\epsilon^4\delta} + \frac{1}{\epsilon^4\delta^2}\right )$ for a large enough constant $c$ and adjusting for other constants we have $|z_j| \leq 2\epsilon^2\delta \norm{\bv{A}}_F^2$.
Also observe that the expectation of $\sum z_j$ can be bounded by:
\begin{align*}
    \E\left(\sum_{j=1}^n z_j\right) = \E\|(\bar{\bv S}{\bv A'}_m)_{:,i}\|_2^2 = \|{\bv A'}_{m,i}\|_2^2 \leq \|{\bv A'}_m\|_2^2 \leq \epsilon^2\delta\norm{\bv{A}}_F^2.
\end{align*}

\noindent Next, the variance of the sum of the random variables $\{z_j\}$ can again be bounded by following the analysis presented in and prior to \eqref{eq:l2_unsolved_var_sum_zj_2} and \eqref{eq:l2_var_sum_zj} we have
\begin{align}
    \bv{Var}\left(\sum_{j=1}^n z_j\right) &\leq \|{\bv A'}_{m,i,j}\|_2^4 + \frac{12c_2 \log^2 n \cdot \norm{\bv{A}_i}_2^4}{s\epsilon^2} + \frac{12\norm{\bv{A}_i}_2^4}{s\epsilon^4\delta^2}\notag\\
    &\leq \epsilon^4\delta^2\norm{\bv{A}}_F^4 + \frac{12c_2 \log^2 n \cdot \norm{\bv{A}_i}_2^4}{s\epsilon^2} + \frac{12 \norm{\bv{A}_i}_2^4}{s\epsilon^4\delta^2},\label{eq:l2_var_z_j_p=1}
\end{align}
where we again bound $\|{\bv A'}_{m,i,j}\|_2^4$ using $$|{\bv A'}_{m,i,j}|^2 \leq \sum_{j=1}^n |{\bv A'}_{m,i,j}|^2 \leq \| \bv{A}'_{m,i}\|^2_2 \leq \| \bv{A}'\|^2_2 \leq \epsilon^2\delta \norm{\bv{A}}_F^2.$$

\noindent Then for $s \geq c(\frac{\log^4 n }{\epsilon^6\delta^2} + \frac{1}{\epsilon^8\delta^4})$, we have $\bv{Var}\left(\sum_{j=1}^n z_j\right) \leq 2\epsilon^4\delta^2\norm{\bv{A}}_F^4$ for large enough constant $c$.

\noindent Using \eqref{eq: l2_abs_zj_p=1} and \eqref{eq:l2_var_z_j_p=1} and noting that $\sum_{j=1}^n\E\left( z_j^2\right) \geq \bv{Var}\left(\sum_{j=1}^n z_j\right) - \E^2\left(\sum_{j=1}^n z_j\right)$ we can apply the Bernstein inequality (Theorem \ref{thm:bernstein}):
\begin{align*}
    \Pr\left(\|(\bar{\bv S}\bv A'_m)_{:,i}\|_2^2 \geq \E\|(\bar{\bv S}\bv A'_m)_{:,i}\|_2^2+t \right) &\leq \Pr\left(\sum_{j=1}^n z_j \geq \epsilon^2\delta\norm{\bv{A}_i}_2^2+t \right) \\ &\leq \exp\left(\frac{-t^2/2}{c\epsilon^4\delta^2\norm{\bv{A}}_F^4+c\epsilon^2\delta\norm{\bv{A}}_F^2t/3}\right).
\end{align*}
If we set $t=(\log n)\epsilon^2\delta\norm{\bv{A}}_F^2$, then for some constant $c'$ we have
\begin{align*}
    \Pr\left(\|(\bar{\bv S}\bv A'_m)_{:,i}\|_2^2 \geq \E\|(\bar{\bv S}\bv A'_m)_{:,i}\|_2^2+t \right) &\leq \exp(-c'\log n) \leq 1/n^{c'}.
\end{align*}
This, since $\|(\bar{\bv S}\bv H_m)_{:,i}\|_2^2 \leq \|(\bar{\bv S}\bv A'_m)_{:,i}\|_2^2$, when $p_i = 1$, setting  $s \geq c(\frac{\log^4 n }{\epsilon^6\delta^2} + \frac{1}{\epsilon^8\delta^4})$ for large enough $c$,  we have with probability $\ge 1-1/n^{c'}$
$
    \frac{1}{p_i}\|(\bar{\bv S}\bv H_m)_{:,i}\|_2^2 = \|(\bar{\bv S}\bv H_m)_{:,i}\|_2^2 \leq \|(\bar{\bv S}\bv A'_m)_{:,i}\|_2^2 \leq (\log n)\epsilon^2\delta\nnz({\bv A}).
$

We have proven that with probability $\ge 1-1/n^{c'}$, for both cases when $p_i < 1$ and $p_i = 1$, $\frac{\|(\bar{\bv{S}}\bv H_m)_{:,i}\|_2^2}{p_i} \le (\log n)\epsilon^2\delta\norm{\bv{A}}_F^2$. Taking a union bound over all $i \in [n]$,  with probability at least  $1-1/n^{c'-1}$, $ \max_i \frac{\|(\bar{\bv S}\bv H_m)_{:,i}\|_2}{\sqrt{p_i}} \leq \sqrt{\log n}\epsilon\sqrt{\delta}\norm{\bv{A}}_F$ for $s \geq c(\frac{\log^4 n }{\epsilon^6\delta^2} + \frac{1}{\epsilon^8\delta^4})$. Also, since $p_i \geq \frac{1}{n^2}$ for all $i \in [n]$, $ \frac{\|(\bar{\bv S}\bv H_m)_{:,i}\|_2}{\sqrt{p_i}} \leq \sqrt{\sum_{j=1}^n \frac{\bv{A}_{m,i,j}^2}{p_i\cdot p_j}}\leq \frac{n\cdot\norm{\bv{A}}_F}{\sqrt{s}} $. Thus, $\max_i \frac{\|(\bar{\bv S}\bv H_m)_{:,i}\|_2}{\sqrt{p_i}} \leq n\norm{\bv{A}}_F$ and we get, 
\begin{align*}
    \E_2 \left (\max_{i: i\in [n]} \frac{\|(\bar{\bv S}\bv H_m)_{:,i}\|_2}{\sqrt{p_i}}\right ) \leq \sqrt{\log n}\epsilon\sqrt{\delta}\norm{\bv{A}}_F+\frac{1}{n^{c'-3}} \leq \sqrt{\log n}\epsilon\sqrt{\delta}\norm{\bv{A}}_F.
\end{align*} 
after adjusting $\epsilon$ by at most some constants.
Overall, we finally get 
$$\E_2\|\bar{\bv S}\bv H_m\hat{\bv S}\|_{1\to 2} \leq \E_2 \left (\max_{i: i\in [n]} \frac{\|(\bar{\bv S}\bv H_m)_{:,i}\|_2}{\sqrt{p_i}}\right ) \leq \epsilon\sqrt{\log n}\sqrt{\delta}\norm{\bv{A}}_F.$$
Plugging this bound into \eqref{eq:l2_decoupled3}, we have for $s \geq c(\frac{\log^4 n }{\epsilon^6\delta^2} + \frac{1}{\epsilon^8\delta^4})$,
\begin{align*}
    \E_2\|\bar{\bv S}{\bv A'}_m \bar{\bv S}\|_2 &\leq (\log n)\epsilon\sqrt{\delta}\norm{\bv{A}}_F.
\end{align*}
Finally after adjusting $\epsilon$ by a $\frac{1}{\log n}$ factor, we have for $s \geq c(\frac{\log^{10} n }{\epsilon^6\delta^2} + \frac{\log^8 n}{\epsilon^8\delta^4})$ or $s \geq \frac{c\log^{10} n }{\epsilon^8\delta^4}$,
\begin{align*}
    \E_2\|\bar{\bv S}{\bv A'}_m \bar{\bv S}\|_2 &\leq \epsilon\sqrt{\delta}\norm{\bv{A}}_F.
\end{align*}
The final bound then follows via Markov's inequality on $\|\bar{\bv S}{\bv A'}_m \bar{\bv S}\|_2$.
\end{proof}

\subsection{Main Accuracy Bound}

We are finally ready to state our main result for $\ell_2$ norm based sampling.

\rownormSamp*

\begin{proof}
The proof follows exactly the same structure as the proofs of Theorems~\ref{thm:main_bound} and~\ref{thm:nnz_main_bound} for uniform and sparsity based sampling respectively. We use the results of Lemmas~\ref{lem:l2-middle} and~\ref{lemma:l2_large} on the concentration of the middle and large eigenvalues for $\ell_2$ norm based sampling.  

Analogous to Theorem \ref{thm:nnz_main_bound}, from Lemma \ref{lemma:l2_large} with error parameter $\frac{\epsilon}{\log n}$ the eigenvalues of $\bv A'_{o,S}$ approximate those of $\bv{A}_o'$ up to error $\epsilon\|\bv A\|_F$ with probability $1-\delta$ if $s \geq  \frac{c \log(1/(\epsilon \delta)) \cdot \log^3 n}{\epsilon^3 \sqrt{\delta}}$. We also require $s \geq \frac{c\log^{10}n}{\epsilon^8\delta^4}$ for $\|\bv A_{m,S}'\|_2 \leq \epsilon\|\bv A\|_F$ to hold with probability $1-\delta$ by Lemma \ref{lem:l2-middle}. Thus, for both conditions to hold simultaneously with probability $1-2\delta$ by a union bound, if suffices to set   $s = \max\left(\frac{c \log(1/(\epsilon \delta)) \cdot \log^3 n}{\epsilon^3 \sqrt{\delta}}, \frac{c\log^{10}n}{\epsilon^8\delta^4}\right) = \frac{c\log^{10}n}{\epsilon^8\delta^4}$, where we use that $\log(1/(\epsilon \delta) \le \log n$, as otherwise our algorithm can take $\bv{A}_S$ to be the full matrix $\bv{A}$. Adjusting $\delta$ to $\delta/2$ completes the theorem.
\end{proof}


\section{Eigenvalue Approximation via Entrywise Sampling}\label{app:entry}

In this section we show that sampling $\tilde O(n/\epsilon^2)$ entries from a bounded entry matrix preserves its eigenvalues up to error $\pm \epsilon  n$. We use this result to improve the sample complexity of Theorem \ref{thm:main_bound} from $\tilde O \left (\frac{\log^6 n}{\epsilon^6} \right )$ to $\tilde O \left (\frac{\log^3 n}{\epsilon^5} \right )$ by applying entrywise sampling to further sparsify the submatrix $\bv{A}_S$ that is sampled in Algorithm \ref{alg:eigenvalue estimate}. Entrywise sampling results similar to what we show are well-known in the literature. See for example \cite{achlioptas2007fast} and \cite{braverman2021near}. For completeness, we give a proof here using standard matrix concentration bounds.

\begin{theorem}[Entrywise sampling -- spectral norm bound]\label{thm:entrywise sampling}
Consider $\bv A \in \R^{\nbyn}$ with $\|\bv A\|_\infty \leq 1$. Let $\bv C \in \R^{n \times n}$ be constructed by setting $\bv C_{i,i} = \bv{A}_{i,i}$ for all $i \in [n]$ and 
\begin{align*}
   \bv C_{j,i} =  \bv C_{i,j} &= 
    \begin{cases}
        \frac{1}{p} \cdot \bv A_{i,j} & \text{with probability } $p$\\
        0 & \text{otherwise}.
    \end{cases}
\end{align*}
For any $\epsilon, \delta \in (0,1)$, if $p \geq \frac{c\log(n/\delta)}{n\epsilon^2}$ for a large enough constant $c$, then with probability at least $1-\delta$, $\|\bv A - \bv C\|_2 \leq \epsilon n$.
\end{theorem}
Note that by Weyl's inequality (Fact \ref{thm:eigenvalue_perturbation_theorem}), Theorem \ref{thm:entrywise sampling} immediately implies that the eigenvalues of $\bv C$ approximate those of $\bv A$ up to $\pm \epsilon n$ error with good probability.

\begin{proof}
For any $i<j$, define the symmetric random matrix $\bv E^{(ij)}$ with 
\begin{align*}
    \bv E^{(ij)}_{i,j} = \bv E^{(ij)}_{j,i} &=
    \begin{cases}
    (\frac{1}{p}-1)\cdot \bv A_{i,j} & \text{with probability } p\\
    -\bv A_{i,j} & \text{otherwise}.
    \end{cases}
\end{align*}.
Observe that $\bv C -\bv A =  \sum_{i,j\in[n], i < j} \bv E^{(ij)}$.
Further, each $\bv E^{(ij)}$ has just two non-zero values in different rows and columns. So \begin{align*}
   \|\bv E^{(ij)}\|_2 = |\bv C_{i,j} - \bv A_{i,j}]| \leq \left (\frac{1}{p} - 1\right ) \cdot |\bv A_{i,j}|\leq \frac{1}{p},
\end{align*}
where the last inequality uses that $\|\bv A\|_\infty \leq 1$. 
Additionally, $\bv E^{(ij)}\bv E^{(ij)T}$ is diagonal with two diagonal entries at $(i,i)$ or $(j,j)$ equal to $(\bv C_{i,j} - \bv A_{i,j})^2$. Thus, $\bv{V} = \sum_{i,j\in[n], i< j}\E [\bv E^{(ij)}\bv E^{(ij)T}]$ is also diagonal. We have 
\begin{align*}
\bv{V}_{i,i} =  \sum_{j \neq i} \E[(\bv C_{i,j} - \bv A_{i,j})^2] &= \sum_{j \neq i} \bv{A}_{i,j}^2 \cdot \left (p \cdot \left (\frac{1}{p}-1\right )^2 + (1-p) \cdot (-1)^2 \right )\\
&= \sum_{j \neq i} \bv{A}_{i,j}^2 \cdot \left (\frac{1}{p} -1 \right ) \le \frac{n}{p},
\end{align*}
where in the final inequality we use that $\norm{\bv A}_\infty \le 1$. Thus, since $\bv{V}$ is diagonal, $\norm{\bv{V}}_2 \le \frac{n}{p}$. 
Putting the above together using Theorem \ref{thm:matrix bernstein} we get,
\begin{align*}
   \Pr\left (\norm{\bv A - \bv C}_2 \ge \epsilon n\right ) = \Pr\left(\left\|\sum_{i,j\in[n], i< j}\bv E^{(ij)}\right\|_2 \geq \epsilon n\right) &\leq 2n \cdot \exp\left(\frac{-\epsilon^2n^2/2}{\frac{n}{p}+\frac{\epsilon n}{3p}}\right).
\end{align*} 
Thus, for $p \geq \frac{c\log (n/\delta)}{n\epsilon^2}$ for large enough $c$, with probability at least $1-\delta$ we have $\norm{\bv A - \bv C}_2 \leq \epsilon n$.
\end{proof}

\medskip

\subsection{Improved Sample Complexity via Entrywise Sampling}
We can combine Theorem \ref{thm:entrywise sampling} directly with Theorem \ref{thm:main_bound} to give an improved sample complexity for eigenvalue estimation. we have:
\entrywise*
\begin{proof}
Letting $s = \frac{c_1 \log(1/(\epsilon \delta)) \cdot \log^3 n}{\epsilon^3 \delta}$ for large enough constant $c_1$, by Theorem \ref{thm:main_bound}, for a random principal submatrix $\bv{A}_S$ formed by sampling each index with probability $s/n$, the eigenvalues of $\bv{A}_S$, after scaling up by a factor of $n/s$ approximate those of $\bv A$ to error $\pm \epsilon n$ with probability at least $1-\delta$. By Theorem \ref{thm:entrywise sampling}, if we sample off-diagonal entries of $\bv{A}_S$ with probability $p \ge \frac{c_2 \log(|S|/\delta)}{|S| \cdot \epsilon^2}$ to produce $\bv{C}$, then we preserve its eigenvalues to error $\pm \epsilon |S|$. Thus, after scaling by $\frac{n}{s}$, the eigenvalues of $\bv{C}$ approximate those of $\bv{A}$ to error $\pm \left (\epsilon n + \frac{n}{s} \cdot \epsilon |S| \right )$. Finally, observe that by a standard Chernoff bound, $|S| \le 2 s$ with probability at least $1-\delta$. So adjusting $\epsilon$ by a constant, the scaled eigenvalues of $\bv C$ give  $\pm \epsilon n$ approximations to $\bv{A}$'s eigenvalues. The expected number of entries read is $|S| + p \cdot |S|^2 = \tilde O \left (\frac{s \cdot \log(1/\delta)}{\epsilon^2} \right ) = \tilde O \left ( \frac{\log^3 n}{\epsilon^5 \delta}\right )$. Additionally, by a standard Chernoff bound at most $\tilde O\left ( \frac{\log^3 n }{\epsilon^5 \delta}\right )$ are read with probability at least $1-\delta$.
\end{proof}


\section{Singular Value Approximation via Sampling}\label{sec:singular}

We now show how to estimate the singular values of a bounded-entry matrix via random subsampling. Unlike in eigenvalue estimation, instead of sampling a random principal submatrix, here we sample a random submatrix with independent rows and columns. This allows us to apply known interior eigenvalue matrix Chernoff bounds to bound the perturbation in the singular values \cite{gittens2011tail,BakshiChepurkoJayaram:2020}. We first state a simplified version of Theorem 4.1 from \cite{gittens2011tail} (also stated as Theorem 4.6 in \cite{BakshiChepurkoJayaram:2020}), simplified using standard upper bounds on the Chernoff bounds in \cite{mitzenmacher2017probability}.
\begin{theorem}[Interior Eigenvalue Matrix Chernoff bounds -- Theorem 4.1 of~\cite{gittens2011tail}]\label{thm:interior}
Let $\{\bv{X}_j\}$ be a finite sequence of independent, random, positive-semidefinite matrices with dimension $n$, and assume that $\|\bv{X}_j \|_2 \leq L$ for some value $L$ almost surely. Given an integer $k \leq n$, define $$\mu_k=\lambda_k \left( \sum_j \mathbb{E}[\bv{X}_j]\right).$$ Then we have the tail inequalities:
\begin{align*}
    \begin{cases}
     \Pr \left(\lambda_k(\sum_j \bv{X}_j) \geq (1+\Delta)\mu_k \right) \leq (n-k+1)\cdot e^{-\frac{\Delta \mu_k}{3L}}, & \text{for } \Delta \geq 1\\
     \Pr \left(\lambda_k(\sum_j \bv{X}_j) \geq (1+\Delta)\mu_k \right) \leq (n-k+1)\cdot e^{-\frac{\Delta^2 \mu_k}{3L}}, & \text{for } \Delta \in [0,1)\\
    \Pr \left(\lambda_k(\sum_j \bv{X}_j) \leq (1-\Delta)\mu_k \right) \leq k\cdot e^{-\frac{\Delta^2 \mu_k}{2L}}, & \text{for } \Delta \in [0,1)
    \end{cases}
\end{align*}
\end{theorem}
We are now ready to state and prove the main theorem. 
\begin{theorem}\label{thm:singular}
Let $\bv A \in \mathbb{R}^{n\times n}$ be a matrix with $\|\bv A\|_\infty \leq 1$ and singular values $\sigma_1(\bv{A}) \ge \ldots \ge \sigma_n(\bv{A})$. 
Let $\bv{\bar S} \in \R^{n \times n}$ be a scaled diagonal sampling matrix such that $\bv{\bar S}_{ii}=\sqrt{\frac{n}{s}}$ with probability $\frac{s}{n}$ and $\bv{\bar S}_{ii}=0$ otherwise. Let $\bv{\bar T} \in \R^{n \times n}$ be an independent and identically distributed random sampling matrix. Let  $\bv{Z}=\bv{\bar{S}A \bar{T}}$ be the sampled submatrix from $\bv{A}$ with singular values $\sigma_1(\bv{Z}) \ge \ldots \ge \sigma_{n}(\bv{Z})$.  Then, if $s \geq \frac{c \log (n/\delta) }{\epsilon^2}$ for some constant $c$, with probability at least $1-\delta$, for all $i \in [n]$,
\begin{align*}
    \sigma_i(\bv A) -\epsilon n \leq  \sigma_i(\bv Z) \leq \sigma_i(\bv A) +\epsilon n.
\end{align*}
\end{theorem}
\begin{proof}
We first prove that singular values of $\bar{\bv{S}}\bv{A}$ are close to those of $\bv{A}$. Let $\bv X_i \in \R^{n \times n}$ be matrix valued r.v.'s for $i \in [n]$ such that: 
\begin{align*}
    \bv X_i = 
    \begin{cases}
    \frac{n}{s}\bv{A}_{i}\bv{A}_{i}^T , & \text{with probability } s/n\\
    0 & \text{otherwise}
    \end{cases}
\end{align*}
where $\bv{A}_{i}$ is the $i$\textsuperscript{th} row of $\bv{A}$ written as a column vector. Then, $\sum_i \bv{X}_i = (\bv{\bar{S}}\bv A)^T(\bv{\bar{S}}\bv A)$ and $\mathbb{E}[\sum_i\bv{X}_i]=\bv{A}^T\bv{A}$. We have $\|\bv{X}_i \|_2 \leq \max_j \frac{n}{s}\|\bv{A}_{j} \|^2_2 \leq \frac{n^2}{s}$ and $\lambda_k(\mathbb{E}[\sum_i\bv{X}_i])= \lambda_k(\bv{A}^T\bv{A}) = \sigma_k^2(\bv{A})$ for $k \in [n]$. 

\medskip

\noindent \textbf{Case 1}: We will first prove that $\sigma_k(\bv{A})-\epsilon n \leq \sigma_k(\bar{\bv{S}}\bv{A})$ for all $k \in [n]$. Note that when $\sigma_k(\bv{A}) \leq \epsilon n$, $\sigma_k(\bv{A})-\epsilon n \leq \sigma_k(\bar{\bv S}\bv{A})$ is trivially true. We now consider all $k \in [n]$ such that $\sigma_k(\bv{A})>\epsilon n$. Setting $\mu_k=\lambda_k(\bv{A}^T\bv{A})$, $L= \frac{n^2}{s}$ and $\Delta=\frac{\epsilon n}{\sigma_k(\bv{A})}$ (note that $\Delta<1$) in Theorem~\ref{thm:interior}, we get:  
\begin{align*}
    \Pr\left( \lambda_k((\bv{\bar{S}}\bv A)^T(\bv{\bar{S}}\bv A)) \leq (1-\Delta)\lambda_k(\bv{A}^T\bv{A})  \right) &\leq   k \cdot e^{-c\frac{\Delta^2_1 \lambda_k(\bv{A}^T\bv{A}) }{L}} \leq  k \cdot e^{-c\frac{\epsilon^2 n^2 }{\lambda_k(\bv{A}^T\bv{A}) } \cdot \frac{\lambda_k(\bv{A}^T\bv{A}) }{(n^2/s)}}
\end{align*}
where $c$ is constant. So, for $s \geq O(\frac{\log (n/\delta)}{\epsilon^2})$ for any $k$, we have $\lambda_k((\bar{\bv S}\bv{A})^T(\bar{\bv S}\bv{A}))=\sigma_k^2 (\bar{\bv S}\bv{A}) \geq (1-\Delta) \sigma_k^2 (\bv{A})$ with probability at least $1-\frac{\delta}{n}$. Taking a square root on both sides we get $\sigma_k(\bar{\bv S}\bv{A}) \geq \sqrt{1-\Delta}\sigma_k(\bv{A}) \geq (1-\Delta)\sigma_k(\bv{A})= \sigma_k(\bv{A})-\epsilon n $. Taking a union bound over all $k$ with $\sigma_k(\bv{A})> \epsilon n$, $\sigma_k(\bv{A})-\epsilon n \leq \sigma_k(\bar{\bv{S}}\bv{A})$ holds for all such $k$ with probability at least $1-\delta$.

\medskip

\noindent \textbf{Case 2}: We now prove that $\sigma_k(\bar{\bv{S}}\bv{A}) \leq \sigma_k(\bv{A})+\epsilon n$ for all $k \in [n]$. We first consider the case when $\sigma_k(\bv{A}) \leq \epsilon n$. Setting $\mu_k=\lambda_k(\bv{A}^T\bv{A})$, $L= \frac{n^2}{s}$ and $\Delta=\frac{\epsilon^2 n^2}{\sigma^2_k(\bv{A})}$ (note that $\Delta \geq 1$) in Theorem~\ref{thm:interior}, we get (for some constant $c$):  
\begin{align*}
    \Pr\left( \lambda_k((\bv{\bar{S}}\bv A)^T(\bv{\bar{S}}\bv A)) \geq (1+\Delta)\lambda_k(\bv{A}^T\bv{A})  \right) &\leq   n.e^{-\frac{c\Delta \lambda_k(\bv{A}^T\bv{A})}{L}} \\ &\leq  n \cdot e^{-\frac{c\epsilon^2 n^2 }{\lambda_k(\bv{A}^T\bv{A})} \cdot \frac{\lambda_k(\bv{A}^T\bv{A})}{(n^2/s)}}
\end{align*}
Thus, if $s\geq O(\frac{\log (n/\delta)}{\epsilon^2})$, we have $\lambda_k((\bv{\bar{S}}\bv A)^T(\bv{\bar{S}}\bv A)) \leq (1+\Delta)\lambda_k(\bv{A}^T\bv{A}) \leq \lambda_k(\bv{A}^T\bv{A})+\epsilon^2 n^2$ for all $k \in [n]$ such that $\sigma_k(\bv{A}) \leq \epsilon n$ with probability at least $1-\delta$ via a union bound. Taking square root on both sides and using the facts that $\lambda_k(\bv{A}^T\bv{A})=\sigma_k^2(\bv{A})$, $\lambda_k((\bv{\bar{S}}\bv A)^T(\bv{\bar{S}}\bv A))=\sigma_k^2(\bv{\bar{S}}\bv A)$ and $\sqrt{a+b}<\sqrt{a}+\sqrt{b}$, we get $\sigma_k(\bv{\bar{S}}\bv A) \leq \sigma_k(\bv{A})+\epsilon n$.

We now consider the case $\sigma_k(\bv{A}) > \epsilon n$. Setting $\mu_k=\lambda_k(\bv{A}^T\bv{A})$, $L= \frac{n^2}{s}$ and $\Delta=\frac{\epsilon n}{\sigma_k(\bv{A})}$ (note that $\Delta < 1$) in Theorem~\ref{thm:interior}, we get (for some constant $c$):  
\begin{align*}
    \Pr\left( \lambda_k((\bv{\bar{S}}\bv A)^T(\bv{\bar{S}}\bv A)) \geq (1+\Delta)\lambda_k(\bv{A}^T\bv{A})  \right) &\leq   n.e^{-\frac{c\Delta^2 \lambda_k(\bv{A}^T\bv{A})}{L}} \\ &\leq  n \cdot e^{-\frac{c\epsilon^2 n^2 }{\lambda_k(\bv{A}^T\bv{A})} \cdot \frac{\lambda_k(\bv{A}^T\bv{A})}{(n^2/s)}}.
\end{align*}
Thus, if $s\geq O(\frac{\log (n/\delta)}{\epsilon^2})$, we have $\lambda_k((\bv{\bar{S}}\bv A)^T(\bv{\bar{S}}\bv A)) \leq (1+\Delta)\lambda_k(\bv{A}^T\bv{A})$ for all $k \in [n]$ such that $\sigma_k(\bv{A}) > \epsilon n$ with probability at least $1-\delta$ via a union bound. Taking square root on both sides and using the fact that $\lambda_k(\bv{A}^T\bv{A})=\sigma_k^2(\bv{A})$, $\lambda_k((\bv{\bar{S}}\bv A)^T(\bv{\bar{S}}\bv A))=\sigma_k^2(\bv{\bar{S}}\bv A)$ and $\sqrt{a}<a$ for any $a>1$, we get $\sigma_k(\bv{\bar{S}}\bv A) \leq (1+\Delta)\sigma_k(\bv{A}) \leq \sigma_k(\bv{A})+\epsilon n$. Thus, via a union bound over all $k \in [n]$, we have $\sigma_k(\bv{\bar{S}}\bv A) \leq \sigma_k(\bv{A})+\epsilon n$ with probability $1-2\delta$.

Thus, via a union bound over the two cases above, for all $k \in [n]$ with probability at least $1-3\delta$ for $s \geq O(\frac{\log (n/\delta)}{\epsilon^2})$ we have, for all $k \in [n]$,
\begin{align}\label{eq:sigma1}
    |\sigma_k(\bv{\bar{S}}\bv A)-\sigma_k(\bv{A})| \leq \epsilon n.
\end{align}
Next we prove that the singular values of $\bv{\bar{S}}\bv A \bar{\bv{T}}$ are close to those of $\bv{\bar{S}}\bv{A}$, using essentially the same approach as above.
Let $\bv{Y}_i$ be a matrix values random variable for $i \in [n]$ such that:
\begin{align*}
    \bv Y_i = 
    \begin{cases}
    \frac{n}{s}(\bv{\bar{S}}\bv A)_i(\bv{\bar{S}}\bv A)^T_i , & \text{with probability } s/n\\
    0 & \text{otherwise}
    \end{cases}
\end{align*} 
where $(\bv{\bar{S}}\bv A)_i$ is the $i$\textsuperscript{th} column of $\bv{\bar{S}}\bv A$. Then, $\sum_i \bv{Y}_i = (\bv{\bar{S}}\bv A \bar{\bv{T}})^T(\bv{\bar{S}}\bv A \bar{\bv{T}})$. Also, we have  $\lambda_k(\mathbb{E}[\sum_i \bv{Y}_i])=\lambda_k((\bv{\bar{S}}\bv A )^T(\bv{\bar{S}}\bv A ))=\sigma_k^2(\bv{\bar{S}}\bv A)  $. First, using a standard Chernoff bound, we can claim that $\bar{\bv{S}}$ will sample at most $2s$ rows from $\bv{A}$ with probability at least $ 1-\delta$ for any $s \geq O(\log (1/\delta))$. Thus, we have $\| \bv{Y}_i\|_2 =\frac{n}{s}\|\bv{\bar{S}}\bv A \|_2^2 \leq \frac{n}{s}\cdot \frac{n}{s}\cdot 2s \leq \frac{2n^2}{s}$ with probability $1-\delta$. Let this event be called $E_2$. We now consider two cases conditioned on the event $E_2$.

\medskip

\noindent \textbf{Case 1}: We first prove that $\sigma_k(\bv{\bar{S}}\bv A)-\epsilon n \leq \sigma_k(\bar{\bv{S}} \bv{A} \bar{\bv{T}})$ for all $k \in [n]$. Again note that when $\sigma_k(\bv{\bar{S}}\bv A) \leq \epsilon n$ this is trvially true. So we consider all $k \in [n]$ such that $\sigma_k(\bv{\bar{S}}\bv A) > \epsilon n$. Setting $\mu_k=\lambda_k((\bv{\bar{S}}\bv A)^T(\bv{\bar{S}}\bv A))$, $L= \frac{2n^2}{s}$ (as we have conditioned on $E_2$) and $\Delta=\frac{\epsilon n}{\sigma_k(\bv{\bar{S}}\bv A)}$ (note that $\Delta<1$) in Theorem~\ref{thm:interior}, we get:  
\begin{align*}
    \Pr\left( \lambda_k((\bar{\bv{S}} \bv{A} \bar{\bv{T}})^T(\bar{\bv{S}} \bv{A} \bar{\bv{T}})) \leq (1-\Delta)\lambda_k(\bv{A}^T\bv{A})  \right) &\leq   k \cdot e^{-c\frac{\Delta^2_1 \lambda_k((\bv{\bar{S}}\bv A)^T(\bv{\bar{S}}\bv A)) }{L}} \leq  k \cdot e^{-c\frac{\epsilon^2 n^2 }{\lambda_k((\bv{\bar{S}}\bv A)^T(\bv{\bar{S}}\bv A)) } \cdot \frac{\lambda_k((\bv{\bar{S}}\bv A)^T(\bv{\bar{S}}\bv A)) }{(n^2/s)}}
\end{align*}
where $c$ is some constant. So, for $s \geq O(\frac{\log (n/\delta)}{\epsilon^2})$ for any $k$, we have $\lambda_k((\bar{\bv{S}} \bv{A} \bar{\bv{T}})^T(\bar{\bv{S}} \bv{A} \bar{\bv{T}}))=\sigma_k^2 (\bar{\bv S}\bv{A}\bar{\bv{T}}) \geq (1-\Delta) \sigma_k^2 (\bar{\bv S}\bv{A})$ with probability at least $1-\frac{\delta}{n}$. Taking a square root on both sides we get $\sigma_k(\bar{\bv{S}} \bv{A} \bar{\bv{T}}) \geq \sqrt{1-\Delta}\sigma_k(\bar{\bv S}\bv{A}) \geq (1-\Delta)\sigma_k(\bar{\bv S}\bv{A})= \sigma_k(\bar{\bv S}\bv{A})-\epsilon n $. Taking a union bound over all $k$ with $\sigma_k(\bv{A})> \epsilon n$, $\sigma_k(\bar{\bv{S}}\bv{A})-\epsilon n \leq \sigma_k(\bar{\bv{S}}\bv{A} \bar{\bv{T}})$ holds for all such $k$ with probability at least $1-\delta$.

\medskip
\noindent \textbf{Case 2}: We now prove $\sigma_k(\bar{\bv{S}} \bv{A} \bar{\bv{T}}) \leq \sigma_k(\bv{\bar{S}}\bv A)+\epsilon n $ for all $k \in [n]$. We again first consider the case $\sigma_k(\bv{\bar{S}}\bv A) \leq \epsilon n$. Setting $\mu_k=\lambda_k(\bv{A}^T\bv{A})$, $L= \frac{n^2}{s}$ and $\Delta=\frac{\epsilon^2 n^2}{\sigma^2_k(\bv{\bar{S}}\bv A)}$ (note that $\Delta \geq 1$) in Theorem~\ref{thm:interior}:  
\begin{align*}
    \Pr\left( \lambda_k((\bar{\bv{S}} \bv{A} \bar{\bv{T}})^T(\bar{\bv{S}} \bv{A} \bar{\bv{T}})) \geq (1+\Delta)\lambda_k((\bv{\bar{S}}\bv A)^T(\bv{\bar{S}}\bv A))  \right) &\leq   n \cdot e^{-\frac{c\Delta \lambda_k((\bv{\bar{S}}\bv A)^T(\bv{\bar{S}}\bv A))}{L}} \\ &\leq  n \cdot e^{-\frac{c\epsilon^2 n^2 }{\lambda_k((\bv{\bar{S}}\bv A)^T(\bv{\bar{S}}\bv A))} \cdot \frac{\lambda_k((\bv{\bar{S}}\bv A)^T(\bv{\bar{S}}\bv A))}{(n^2/s)}}
\end{align*}
Then, similar to the case $\sigma_k(\bv{A}) \leq \epsilon n$ in the previous case 2, taking square root of both sides and via a union bound, we get $\sigma_k(\bar{\bv{S}} \bv{A} \bar{\bv{T}}) \leq \sigma_k(\bv{\bar{S}}\bv A)+\epsilon n $ for all $k \in [n]$ such that $\sigma_k(\bar{\bv{S}}\bv{A}) \leq \epsilon n$ with probability at least $1-\delta$ for $s \geq O(\frac{\log (n/\delta)}{\epsilon^2})$. The case $\sigma_k(\bar{\bv{S}}\bv{A}) >\epsilon n$ will again be similar as $\sigma_k(\bv{A}) > \epsilon n$ in the previous case 2. We set $\Delta=\frac{\epsilon n}{\sigma_k(\bar{\bv{S}}\bv{A})}$
and apply Theorem~\ref{thm:interior} and take the square root on both sides to get 
$\sigma_k(\bar{\bv{S}} \bv{A} \bar{\bv{T}}) \leq \sigma_k(\bv{\bar{S}}\bv A)+\epsilon n $ with probability $1-\delta$ for all $k \in [n]$ for $s \geq O(\frac{\log (n/\delta)}{\epsilon^2})$. 
Thus, with probability $1-2\delta$, conditioned on the event $E_2$, we have $\sigma_k(\bar{\bv{S}} \bv{A} \bar{\bv{T}}) \leq \sigma_k(\bv{\bar{S}}\bv A)+\epsilon n $ for all $k \in [n]$. Finally, via a union bound over the two cases above, and conditioned on $E_2$, for all $k \in [n]$ with probability at least $1-2\delta$ for $s \geq O(\frac{\log (n/\delta)}{\epsilon^2})$ we get
\begin{align} \label{eq:sigma2}
    |\sigma_k(\bar{\bv{S}} \bv{A} \bar{\bv{T}})-\sigma_k(\bv{\bar{S}}\bv A)| \leq \epsilon n.
\end{align} 
Thus, taking a union bound over all the cases above (including $E_2$), from equation \eqref{eq:sigma1} and \eqref{eq:sigma2} and via a triangle inequality, we get:
$
    |\sigma_k(\bar{\bv{S}} \bv{A} \bar{\bv{T}})-\sigma_k(\bv A)| \leq 2\epsilon n
$
with probability at least $1-c\delta$ (where $c$ is a small constant) for $s \geq O(\frac{\log (n/\delta)}{\epsilon^2})$. Adjusting $\epsilon$ and $\delta$ by constant factors gives us the final bound. 
\end{proof}

\noindent \textbf{Remark on Rectangular Matrices}: Though we have considered  $\bv{A}$ to be a square matrix for simplicity, notice that Theorem~\ref{thm:singular} also holds for any arbitrary (non-square) matrix $\bv{A} \in \R^{n \times m}$, with $n$ replaced by $\max(n,m)$ in the sample complexity bound.

\medskip

\noindent \textbf{Remark on Non-Uniform Sampling}: As discussed in Section \ref{sec:sparsityOverview}, simple non-uniform random submatrix sampling via row/column sparsities or norms does not suffice to estimate the singular values up to improved error bounds of $\epsilon \sqrt{\nnz(\bv A)}$ or $\epsilon \norm{\bv A}_F$. A more complex strategy, such as the zeroing out used in Theorems \ref{thm:nnz_main_bound} and \ref{thm:l2_main_bound} must be used.  It is worth noting that following the same proof as Theorem \ref{thm:singular}, it is easy to show that if $\bv{\bar S}$ is sampled according to row norms or sparsities and appropriately weighted, then the singular values of $\bv{\bar S} \bv A$ do approximate those of $\bv A$ up to these improved error bounds. The proof breaks down when analyzing $\bv{\bar S} \bv A \bv{\bar T}$. $\bv{\bar T}$ would have to be sampled according to the row norms/sparsities of $\bv{\bar S} \bv A $, not $\bv{A}$, for the proof to go through. However, in general, these sampling probabilities can differ significantly between $\bv{\bar S} \bv A $ and  $\bv{A}$.

\end{document}